\documentclass[a4paper,11pt]{article}

\usepackage{a4wide}

\usepackage[noadjust]{cite}
\usepackage{amsmath}
\usepackage{amssymb}
\usepackage{amsthm} 
\usepackage{bm}
\usepackage{color}

\usepackage{array}
\usepackage{enumerate}
\usepackage{caption} 

\newenvironment{myitemize}
{ \begin{itemize}
    \setlength{\itemsep}{0pt}
    \setlength{\parskip}{0pt}
    \setlength{\parsep}{0pt}     }
{ \end{itemize}                  } 

\newenvironment{myenumerate}
{ \begin{enumerate}
    \setlength{\itemsep}{0pt}
    \setlength{\parskip}{0pt}
    \setlength{\parsep}{0pt}     }
{ \end{enumerate} 				 }

\newcounter{rowno}

\usepackage{tikz}
\usetikzlibrary{fit,positioning, intersections, shapes.geometric,arrows,decorations.markings,decorations.pathreplacing,calc}
\usepackage{varwidth}

\newcommand\ellipsebyfoci[4]{
	\path[#1] let \p1=(#2), \p2=(#3), \p3=($(\p1)!.5!(\p2)$)
	in \pgfextra{
		\pgfmathsetmacro{\angle}{atan2(\y2-\y1,\x2-\x1)}
		\pgfmathsetmacro{\focal}{veclen(\x2-\x1,\y2-\y1)/2/1cm}
		\pgfmathsetmacro{\lentotcm}{\focal*2*#4}
		\pgfmathsetmacro{\axeone}{(\lentotcm - 2 * \focal)/2+\focal}
		\pgfmathsetmacro{\axetwo}{sqrt((\lentotcm/2)*(\lentotcm/2)-\focal*\focal}
	}
	(\p3) ellipse[x radius=\axeone cm,y radius=\axetwo cm, rotate=\angle];
}

\theoremstyle{plain}
\newtheorem{thm}{Theorem}
\newtheorem*{thm*}{Theorem}
\newtheorem{cor}[thm]{Corollary}
\newtheorem{lem}[thm]{Lemma}
\newtheorem{obs}[thm]{Observation}
\theoremstyle{definition}
\newtheorem{defn}[thm]{Definition}
\newtheorem{rem}[thm]{Remark}
\newtheorem*{leminducedWR3General}{Lemma~\ref{lem:inducedWR3General}}
\newtheorem*{leminducedCycleGeneral}{Lemma~\ref{lem:inducedCycleGeneral}}

\newcommand{\Hom}[1]{\mathrm{\#\textsc{Hom}}\left(#1\right)}

\newcommand{\Ret}[1]{\mathrm{\#\textsc{Ret}}(#1)}
\newcommand{\LHom}[1]{\mathrm{\#\textsc{LHom}}(#1)}
\newcommand{\bis}{\#\mathrm{\textsc{BIS}}}
\newcommand{\sat}{\#\mathrm{\textsc{SAT}}}
\newcommand{\csp}{\mathrm{\#\textsc{CSP}}}
\newcommand{\TCut}[1]{\mathrm{\#\textsc{MultiterminalCut}(#1)}}
\newcommand{\largecut}{\mathrm{\#\textsc{LargeCut}}}

\newcommand{\Zivny}{{\v{Z}}ivn{\'y}}

\renewcommand{\P}{\mathrm{P}}
\newcommand{\FP}{\mathrm{FP}}
\newcommand{\NP}{\mathrm{NP}}
\newcommand{\RP}{\mathrm{RP}}
\newcommand{\numP}{\#\mathrm{P}}
\newcommand{\leap}{\le_\mathrm{AP}}
\newcommand{\eqap}{\equiv_\mathrm{AP}}
\newcommand{\WR}[1]{\mathrm{WR}_{#1}}

\newcommand{\boldS}{\mathbf{S}}
\newcommand{\calA}{\mathcal{A}}

\newcommand{\calG}{\mathcal{G}}
\newcommand{\calH}{\mathcal{H}}
\newcommand{\calI}{\mathcal{I}}
\newcommand{\calL}{\mathcal{L}}

\newcommand{\hatN}{\widehat{N}}

\newcommand{\Ecut}{\mathrm{Cut}}
\newcommand{\NH}{\Gamma_H}
\newcommand{\NHb}{\Gamma_{H_b}}
\newcommand{\Nb}[1]{\Gamma_{#1}}

\newcommand{\abs}[1]{\left\vert #1 \right\vert}
\newcommand{\ceil}[1]{\left\lceil #1 \right\rceil}
\newcommand{\floor}[1]{\left\lfloor #1 \right\rfloor}
\newcommand*\from{\colon}
\let\epsilon=\varepsilon
\newcommand{\eps}{\ensuremath{\varepsilon}}
\newcommand{\ucp}[2]{#1 \times #2}
\newcommand{\stirling}{\genfrac\{\}{0pt}{}}

\let\originalleft\left
\let\originalright\right
\renewcommand{\left}{\mathopen{}\mathclose\bgroup\originalleft}
\renewcommand{\right}{\aftergroup\egroup\originalright}

\renewcommand{\hom}[3][]{{N^{#1}\bigl(#2 \rightarrow #3\bigr)}}

\newcommand{\Imp}{\ensuremath{\mathrm{Imp}}}
\newcommand{\Iv}{I_{\mathrm{v}}}
\newcommand{\Ie}{I_{\mathrm{e}}}
\newcommand{\Cv}{C_{\mathrm{v}}}
\newcommand{\Ce}{C_{\mathrm{e}}}
\newcommand{\Hve}[1]{H_{#1}}
\newcommand{\Dv}[1]{D_v(#1)}
\newcommand{\De}[1]{D_e(#1)}

\newcommand{\X}[3]{X(#1,#2,#3)}
\newcommand{\Jpqt}{J(p,q,t)}
\newcommand{\bisgraphs}{\mathcal{H}_{\mathrm{BIS}}}

\newcommand{\prob}[3]{
\vbox{
\begin{description}
 \item[\bf Name:] #1
 \vspace{-1.75ex}
 \item[\bf Input:] #2  
 \vspace{-1.75ex}
 \item[\bf Output:] #3
\end{description}
}
}

\newcommand{\examplebox}[1]{
	\bigskip
	\noindent\fbox{\parbox{\dimexpr \textwidth-2\fboxsep-2\fboxrule}{
			#1
	}}
	\bigskip
}
\usepackage{filecontents}

\usepackage[hidelinks]{hyperref}

\title{The Complexity of Approximately Counting Retractions to Square-Free Graphs}
\author{Jacob Focke, Leslie Ann Goldberg and  Stanislav \Zivny 
	\thanks{
		The research leading to these results has received funding from 
		the European Research Council under the European Union's Seventh Framework Programme (FP7/2007-2013) ERC grant agreement no.\ 334828 and under the European Union's Horizon 2020 research and innovation programme (grant agreement No 714532). Jacob Focke has received funding from the Engineering and Physical Sciences Research Council (grant ref: EP/M508111/1). Stanislav \Zivny\ was supported by a Royal Society University Research Fellowship. The paper 
		reflects only the authors' views and not the views of the ERC or the European Commission. The European Union is not liable for any use that may be made of the information contained therein.}}
\date{24 February 2021}

\begin{document}

\maketitle
\begin{abstract}
	A \emph{retraction} is a homomorphism from a graph $G$ to an induced subgraph $H$ of $G$ that is the identity on $H$. In a long line of research, retractions have been studied under various algorithmic settings. Recently, the problem of approximately counting retractions was considered. We give a complete trichotomy for the complexity of approximately counting retractions to all square-free graphs (graphs that do not contain a cycle of length $4$).
	It turns out there is a rich and interesting class of graphs for which this problem is complete in the class $\bis$. As retractions generalise homomorphisms, our easiness results extend to the important problem of approximately counting homomorphisms. By giving new $\bis$-easiness results we now settle the complexity of approximately counting homomorphisms for a whole class of non-trivial graphs which were previously unresolved.
\end{abstract}

\section{Introduction}
A function $h$ that maps the vertices of a graph $G$ to the vertices of a graph $H$ is a \emph{homomorphism} from $G$ to $H$ if $h$ preserves the edges of $G$, i.e.~if for every pair of adjacent vertices $u,v\in V(G)$ we have $\{h(u),h(v)\}\in E(H)$. It is well-known that homomorphisms represent graph-theoretic structures including proper vertex colourings and independent sets. For example, consider the graphs $I$, $K_3$ and $C_4$ given in Figure~\ref{fig:introExampleH}.
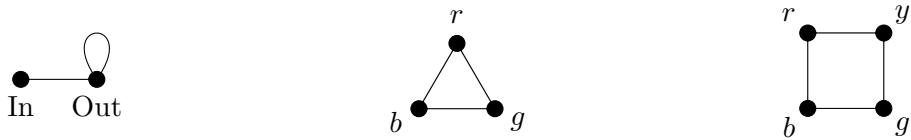
\begin{figure}[ht]
	\centering
	\begin{minipage}{.32 \textwidth}
		\centering
		\begin{tikzpicture}[scale=1, every loop/.style={min distance=10mm,looseness=10}]
		
		\filldraw (0,0) node(a){} circle[radius=3pt] --++ (0:1cm) node(b){} circle[radius=3pt];
		
		\path[-] (b.center) edge  [in=125,out=55,loop] node {} ();
		
		\node at ($(a)+(-90:.35cm)$) {In};
		\node at ($(b)+(-90:.35cm)$) {Out};	
		\end{tikzpicture}
	\end{minipage}
	\centering
	\begin{minipage}{.32 \textwidth}
		\centering
		\begin{tikzpicture}[scale=1]
		
		\filldraw (0,0) node(a){} circle[radius=3pt] --++ (60:1cm) node(b){} circle[radius=3pt] --++ (-60:1cm) node(c){} circle[radius=3pt] -- (a.center);
		
		\node at ($(a)+(210:.35cm)$) {$b$};
		\node at ($(b)+(90:.35cm)$) {$r$};
		\node at ($(c)+(-30:.35cm)$) {$g$};		
		\end{tikzpicture}
	\end{minipage}
	\centering
	\begin{minipage}{.32 \textwidth}
		\centering
		\begin{tikzpicture}[scale=1]
		
		\filldraw (0,0) node(a){} circle[radius=3pt] --++ (90:1cm) node(b){} circle[radius=3pt] --++ (0:1cm) node(c){} circle[radius=3pt] --++ (-90:1cm) node(d){} circle[radius=3pt]-- (a.center);
		
		\node at ($(a)+(225:.35cm)$) {$b$};
		\node at ($(b)+(135:.35cm)$) {$r$};
		\node at ($(c)+(45:.35cm)$) {$y$};
		\node at ($(d)+(-45:.35cm)$) {$g$};		
		\end{tikzpicture}
	\end{minipage}
	\caption{The graphs $I$ (on the left), $K_3$ (in the middle) and $C_4$ (on the right).}
	\label{fig:introExampleH}
\end{figure}
A homomorphism from a graph $G$ to $I$ corresponds to an independent set in $G$, whereas a homomorphism from $G$ to $K_3$ corresponds to a proper $3$-colouring of $G$. Finally, a homomorphism from $G$ to $C_4$ corresponds to a $4$-colouring of $G$ that uses the colours red ($r$), blue ($b$), green ($g$) and yellow ($y$), but for which $\{r,g\}$- and $\{b,y\}$-coloured edges are forbidden. 

If for each vertex $v \in V(G)$ we specify a so-called ``list'' $S_v\subseteq V(H)$ and set $\boldS=\{S_v \mid v\in V(G)\}$, then $h$ is a homomorphism from $(G,\boldS)$ to $H$ if it is a homomorphism from $G$ to $H$ such that for all $v\in V(G)$ it holds that $h(v)\in S_v$. This generalisation of a homomorphism is known as a \emph{list homomorphism}. For example one could consider list homomorphisms from the Petersen graph (see Figure~\ref{fig:petersenGraph} on the left) to the graph $K_3$. These homomorphisms then correspond to proper $3$-colourings of the Petersen graph, where some vertices have pre-assigned colours. 

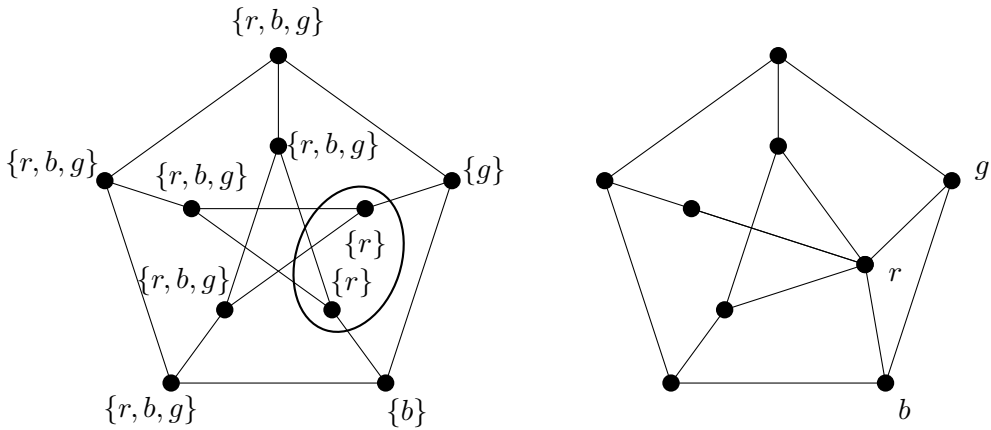
\begin{figure}[ht]
	\centering
	\begin{minipage}{.45 \textwidth}
		\centering
		\begin{tikzpicture}[scale=1.2]
		
		\foreach \x [count=\i] in {18, 90,162,234,306}
		{
			\node[circle,fill=black,inner sep=0pt,minimum size=6.5pt] (v\i) at (\x:1cm){};
			\node[circle,fill=black,inner sep=0pt,minimum size=6.5pt] (u\i) at (\x:2cm){};
			\draw (v\i.center) -- (u\i.center);
		}
		\foreach \i [evaluate=\i as \j using {mod(\i,5)+1}, 
		evaluate=\i as \k using {mod(\i+1,5)+1}] in {1,2,3,4,5}
		{
			\pgfmathtruncatemacro\intj{round(\j)}
			\pgfmathtruncatemacro\intk{round(\k)}
			\draw ({u\i}.center) -- ({u\intj}.center);
			\draw ({v\i}.center) -- ({v\intk}.center);
		}

		\ellipsebyfoci{draw,thick, name=R}{v1.center}{v5.center}{1.4};
		
		\node at ($(v1)+(-90:.4cm)$) {$\{r\}$};
		\node at ($(v2)+(0:.6cm)$) {$\{r,b,g\}$};
		\node at ($(v3)+(72:.37cm)$) {$\{r,b,g\}$};
		\node at ($(v4)+(144:.55cm)$) {$\{r,b,g\}$};
		\node at ($(v5)+(54:.37cm)$) {$\{r\}$};
		
		\node at ($(u1)+(18:.37cm)$) {$\{g\}$};
		\node at ($(u2)+(90:.37cm)$) {$\{r,b,g\}$};
		\node at ($(u3)+(162:.6cm)$) {$\{r,b,g\}$};
		\node at ($(u4)+(234:.37cm)$) {$\{r,b,g\}$};
		\node at ($(u5)+(306:.4cm)$) {$\{b\}$};		
		\end{tikzpicture}
	\end{minipage}
	\centering
	\begin{minipage}{.45 \textwidth}
		\centering
		\begin{tikzpicture}[scale=1.2]
		
		\foreach \x [count=\i] in {18, 90,162,234,306}
		{
			\node[circle,fill=black,inner sep=0pt,minimum size=6.5pt] (u\i) at (\x:2cm){};
		}
		\foreach \x [count=\i] in {-18, 90,162,234,-18}
		{
			\node[circle,fill=black,inner sep=0pt,minimum size=6.5pt] (v\i) at (\x:1cm){};
		}
		\foreach \i [evaluate=\i as \j using {mod(\i,5)+1}, 
		evaluate=\i as \k using {mod(\i+1,5)+1}] in {1,2,3,4,5}
		{
			\pgfmathtruncatemacro\intj{round(\j)}
			\pgfmathtruncatemacro\intk{round(\k)}
			\draw ({u\i}.center) -- ({u\intj}.center);
			\draw ({v\i}.center) -- ({v\intk}.center);
			\draw (v\i.center) -- (u\i.center);
		}

		\ellipsebyfoci{draw, name=R}{v1.center}{v5.center}{1.4};
		
		\node at ($(v1)+(-18:.35cm)$) {$r$};
		\node at ($(u1)+(18:.35cm)$) {$g$};
		\node at ($(u5)+(306:.35cm)$) {$b$};	
		
		\path ($(u2)+(90:.37cm)$) node[opacity=0] {$\{r,b,g\}$};
		\path ($(u5)+(306:.4cm)$) node[opacity=0] {$\{b\}$};	
		\end{tikzpicture}
	\end{minipage}
	\caption{Petersen graph with lists (on the left) and Petersen graph with vertices identified according to single-vertex lists (on the right).}
	\label{fig:petersenGraph}
\end{figure}

If every list contains either just a single vertex or all of the vertices of $H$, then such a list homomorphism is called a \emph{retraction}. This definition of a retraction is equivalent to the definition from the abstract, where a retraction was defined as a homomorphism from a graph $G$ to an induced subgraph $H$ of $G$ that is the identity on $H$. The equivalence (in the sense of parsimonious polynomial-time interreducibility) was shown by Feder and Hell~\cite[Theorem 4.1]{FederLHomRefl}.
The intuition behind the equivalence is that one can identify all of the vertices in $G$ that have the same single-vertex list $\{u\}$ with the corresponding vertex $u$ of $H$. 
Homomorphisms from this new graph 
to $H$ are  retractions in the sense of the abstract.
For the Petersen graph in our example, the modified graph is displayed in Figure~\ref{fig:petersenGraph} on the right and contains $K_3$ as an induced subgraph.

Retractions are also known under the names  \emph{one-or-all list homomorphisms} (e.g.~\cite{FederLHomRefl, FederLHomIrrefl}) and \emph{pre-colouring extensions} (e.g.~\cite{BHT1992,BJW1994,JS1997,KS1997,Tuza1997,Marx2006,FHH2009}).
Related work on retractions is described in Section~\ref{sec:relatedwork}.

In the study of approximate counting there are three important classes of problems~\cite{DGGJApprox}:
(1)~problems that have fully-polynomial-time randomised approximation
schemes (FPRASes),
(2)~problems that are approximation-equivalent to $\bis$, the problem of counting
independent sets in a bipartite graph, and
(3)~problems that are approximation-equivalent to $\sat$, the problem of counting satisfying assignments to a Boolean formula (these problems have no
FPRAS, unless $\NP=\RP$). It is believed that these three classes are disjoint, so there are no FPRASes for the $\bis$-equivalent and $\sat$-equivalent problems.
The problems that are interreducible with $\bis$ under approximation-preserving (AP-)reductions are  complete in a complexity class 
which is sometimes
called $\#\mathrm{RH}\Pi_1$ and is sometimes just called~$\bis$.
For convenience we say that a graph $H$ is $\bis$-easy or $\bis$-hard
if the problem of approximately counting retractions to $H$ is $\bis$-easy or $\bis$-hard, respectively. 
We similarly use the terms $\sat$-easy, and $\sat$-hard.

In this work we give a complete complexity trichotomy for approximately counting retractions to all square-free graphs 
(graphs that do not contain a $4$-cycle)
and we show that all of these problems fall within the three  given complexity classes. An interesting feature is that the class of $\bis$-equivalent graphs turns out to be surprising and rich. 

First we give some illustrative examples. Afterwards we describe the class of $\bis$-equivalent graphs in detail.
A key idea that emerges in the proofs is the role of triangles (3-cycles). It
turns out that triangles in graphs can induce hardness, but they can
also ``turn'' $\sat$-hard cases into $\bis$-easy ones. For example, consider Figure~\ref{fig:introBIS1}. The graph on the left was shown to be $\sat$-hard~\cite[Lemma 2.15]{FGZRet}. In comparison, we will show that the graph on the right is actually $\bis$-easy. Note that this is not because a vertex was added: If one deletes any of the edges of the triangle, the resulting graph is $\sat$-hard again.
\begin{figure}[h!]\centering
	\centering
	\begin{minipage}{.45 \textwidth}
		\centering
		\begin{tikzpicture}[scale=1, every loop/.style={min distance=10mm,looseness=10}]
		
		\filldraw (0,0) node(o){} circle[radius=3pt] --++ (0:1.5cm) node(t0){} circle[radius=3pt] --++ (0:1.5cm) node(t1){} circle[radius=3pt];
		
		\path[-] (o.center) edge  [in=125,out=55,loop] node {} ();	
		\path[-] (t0.center) edge  [in=125,out=55,loop] node {} ();
		\path[-] (t1.center) edge  [in=125,out=55,loop] node {} ();
		
		\filldraw (t0.center) -- (1.25,-1.5cm) node(a22){} circle[radius=3pt];
		\filldraw (t0.center) -- (1.75,-1.5cm) node(a23){} circle[radius=3pt];
		
		\path[opacity=0] ($(t0)+(60:1.5cm)$) node(t) {};
		\path[opacity=0] (t.center) edge  [in=125,out=55,loop] node {} ();	
		\end{tikzpicture}
	\end{minipage}
	\begin{minipage}{.45 \textwidth}
		\centering
		\begin{tikzpicture}[scale=1, every loop/.style={min distance=10mm,looseness=10}]
		
		\filldraw (0,0) node(o){} circle[radius=3pt] --++ (0:1.5cm) node(t0){} circle[radius=3pt] --++ (60:1.5cm) node(t1){} circle[radius=3pt] --++ (-60:1.5cm) node(t2){} circle[radius=3pt];
		
		\draw (t0.center) -- (t2.center);
		
		\path[-] (o.center) edge  [in=125,out=55,loop] node {} ();	
		\path[-] (t0.center) edge  [in=125,out=55,loop] node {} ();
		\path[-] (t1.center) edge  [in=125,out=55,loop] node {} ();
		\path[-] (t2.center) edge  [in=125,out=55,loop] node {} ();	
		
		\filldraw (t0.center) -- (1.25,-1.5cm) node(a22){} circle[radius=3pt];
		\filldraw (t0.center) -- (1.75,-1.5cm) node(a23){} circle[radius=3pt];
		
		\end{tikzpicture}
	\end{minipage}
	\caption{Triangles can induce $\bis$-easiness: The graph on the left is $\sat$-hard whereas the graph on the right is $\bis$-easy.}
	\label{fig:introBIS1}
\end{figure}
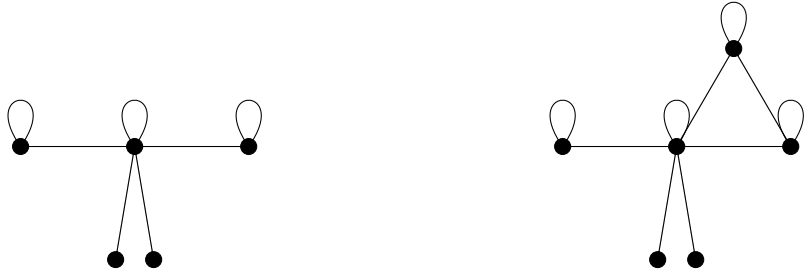
To give an even more striking example, it seems surprising that the graph on the left in Figure~\ref{fig:introBIS1} is $\sat$-hard, but the graph depicted in Figure~\ref{fig:introBIS2} turns out to be $\bis$-easy.
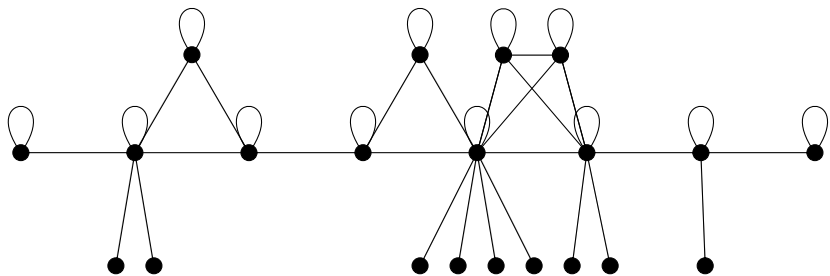
\begin{figure}[h!]\centering
	{\def\scaleFactor{1}
		\begin{tikzpicture}[scale=1, every loop/.style={min distance=10mm,looseness=10}]
		
		\filldraw (0,0) node(o){} circle[radius=3pt] --++ (0:1.5cm) node(t0){} circle[radius=3pt] --++ (60:1.5cm) node(t1){} circle[radius=3pt] --++ (-60:1.5cm) node(t2){} circle[radius=3pt] --++ (0:1.5cm) node(p0){} circle[radius=3pt] --++ (60:1.5cm) node(p1){} circle[radius=3pt] --++ (-60:1.5cm) node(p2){} circle[radius=3pt] --++ (75:1.34cm) node(p3a){} circle[radius=3pt] --++ (0:.75cm) node(p3b){} circle[radius=3pt] --++ (-75:1.34cm) node(p4){} circle[radius=3pt] --++ (0:1.5cm) node(p5){} circle[radius=3pt] --++ (0:1.5cm) node(p6){} circle[radius=3pt];
		
		\draw (t0.center) -- (t2.center);
		\draw (p0.center) -- (p2.center);
		\draw (p2.center) -- (p4.center);
		\draw (p3a.center) -- (p2.center);
		\draw (p3a.center) -- (p4.center);
		\draw (p3b.center) -- (p2.center);
		\draw (p3b.center) -- (p4.center);
		
		\path[-] (o.center) edge  [in=125,out=55,loop] node {} ();	
		\path[-] (t0.center) edge  [in=125,out=55,loop] node {} ();
		\path[-] (t1.center) edge  [in=125,out=55,loop] node {} ();
		\path[-] (t2.center) edge  [in=125,out=55,loop] node {} ();	
		\path[-] (p0.center) edge  [in=125,out=55,loop] node {} ();
		\path[-] (p1.center) edge  [in=125,out=55,loop] node {} ();
		\path[-] (p2.center) edge  [in=125,out=55,loop] node {} ();
		\path[-] (p3a.center) edge  [in=125,out=55,loop] node {} ();
		\path[-] (p3b.center) edge  [in=125,out=55,loop] node {} ();		
		\path[-] (p4.center) edge  [in=125,out=55,loop] node {} ();
		\path[-] (p5.center) edge  [in=125,out=55,loop] node {} ();
		\path[-] (p6.center) edge  [in=125,out=55,loop] node {} ();

		\filldraw (t0.center) -- (1.25,-1.5cm) node(a22){} circle[radius=3pt];
		\filldraw (t0.center) -- (1.75,-1.5cm) node(a23){} circle[radius=3pt];
		\filldraw (p2.center) -- (5.25,-1.5cm) node(b21){} circle[radius=3pt];
		\filldraw (p2.center) -- (5.75,-1.5cm) node(b22){} circle[radius=3pt];
		\filldraw (p2.center) -- (6.25,-1.5cm) node(b23){} circle[radius=3pt];
		\filldraw (p2.center) -- (6.75,-1.5cm) node(b24){} circle[radius=3pt];
		\filldraw (p4.center) -- (7.25,-1.5cm) node(b41){} circle[radius=3pt];
		\filldraw (p4.center) -- (7.75,-1.5cm) node(b42){} circle[radius=3pt];
		\filldraw (p5.center) -- (9,-1.5cm) node(b51){} circle[radius=3pt];
		
		\end{tikzpicture}
	}
	\caption{Example graph which turns out to be $\bis$-easy.}
	\label{fig:introBIS2}
\end{figure}
There exists an interesting underlying balancing process between looped cliques in the neighbourhood of a looped vertex and its number of unlooped neighbours, which decides whether a graph is $\bis$-easy or $\sat$-hard.  

We now informally define the $\bis$-easy class $\bisgraphs$
(it is defined formally in Definition~\ref{def:bisgraphs}). 
This class also contains graphs with squares and is more general than what we will need to classify all square-free graphs.
A graph in $\bisgraphs$ is a path $P$ of looped vertices with some attached unlooped degree-$1$ vertices (bristles, depicted below the path in Figure~\ref{fig:introBIS2}) and some attached looped vertices forming cliques with two consecutive vertices of $P$ (depicted above the path in Figure~\ref{fig:introBIS2}). For each vertex $p$ of the path $P$ the number of attached bristles satisfies the following properties:
\begin{myitemize}
	\item If $p$ is an endpoint of $P$ then it does not have a bristle.
	\item If $p$ is not an endpoint of $P$ then it is part of exactly two reflexive cliques $K_L$ (``to the left'' of $p$) and $K_R$ (``to the right'' of $p$). Then $p$ has at most $(\abs{K_L}-1)\cdot (\abs{K_R}-1)$ bristles.
\end{myitemize}

Consider again the graph in Figure~\ref{fig:introBIS2}. Note that the graph is not square-free as this is
not required by the definition of $\bisgraphs$. The third vertex from the right
on the path is part of a reflexive $4$-clique to the left and a reflexive $2$-clique to the right.
Hence the vertex can have at most $(4-1)\cdot (2-1)=3$ bristles (and has in fact
only $2$ bristles).

\newcommand{\ThmbisEasyRet}{
	Let $H$ be a graph in $\bisgraphs$. Then 
	approximately counting retractions to~$H$	
	is $\bis$-equivalent under approximation-preserving reductions.
}

\begin{thm}\label{thm:bisEasyRet}
	\ThmbisEasyRet
\end{thm} 

$\bisgraphs$ is a fairly broad class of graphs but it turns out that approximately counting retractions to any member of $\bisgraphs$ is $\bis$-easy. For the class of square-free graphs, $\bisgraphs$ completely captures the truth --- together with the class of 
non-trivial irreflexive caterpillars (defined momentarily) they form precisely the class of $\bis$-equivalent square-free graphs. 
Here (following a couple of necessary definitions)
is the full complexity classification.

A graph $H$ is $\emph{reflexive}$ if every vertex of $H$ is looped and it is \emph{irreflexive} if it contains only unlooped vertices.
A \emph{square} is a cycle of length $4$. A connected irreflexive graph is a \emph{star} if it contains at most one vertex of degree greater than $1$, and it is a \emph{caterpillar} if it contains a path $P$ such that all vertices outside of $P$ have degree $1$. A graph is \emph{trivial} if it is a reflexive clique or an irreflexive complete bipartite graph.  

\newcommand{\ThmRetNoSquare}{
	
	Let $H$ be a {square-free} graph. 
	
	\renewcommand{\theenumi}{\roman{enumi})}
	\renewcommand{\labelenumi}{\theenumi}
	\begin{myenumerate}
		
		\item If  every connected component of $H$ is trivial then approximately counting retractions to $H$ is in FP.
		
		\item Otherwise, if every connected component of $H$ is 
		
		\begin{myitemize}
			
			\item trivial,
			
			\item in the class $\bisgraphs$, or
			
			\item is an irreflexive caterpillar
			
		\end{myitemize} 
		then approximately counting retractions to $H$ is $\bis$-equivalent. 
		
		\item Otherwise, approximately counting retractions to $H$ is $\sat$-equivalent. 
		
	\end{myenumerate}
}

\begin{thm}\label{thm:RetNoSquare}
	\ThmRetNoSquare
\end{thm}

As a second central contribution of this work we emphasise the implications of Theorem~\ref{thm:bisEasyRet} on the complexity of approximately counting homomorphisms.  
This problem has been studied intensively in the past~\cite{DGGJApprox, GJTreeHoms, DGJSampling, GGJList, GGJBIS, KelkThesis, GKP2004, FGZRet}, but despite these efforts its complexity is still mostly open. For example, there are graphs with as few as four vertices for which its complexity is unresolved. In a partial result, Galanis, Goldberg and Jerrum~\cite{GGJBIS} proved that approximately counting homomorphisms
to a connected graph~$H$ is $\bis$-hard unless 
$H$ is an irreflexive complete bipartite graph or a reflexive clique. However, the knowledge about $\bis$-easiness to complement this result is very fragmented. 
Since approximately counting homomorphisms reduces to
approximately counting retractions (\cite[Observation 1.2]{FGZRet}), our $\bis$-easiness results from Theorem~\ref{thm:bisEasyRet}
carry over to the homomorphism-counting domain. 
Thus, we are able to resolve the complexity of approximately counting homomorphisms to graphs in~$\bisgraphs$.

\begin{cor}\label{cor:bisEasyHom}
	Let $H$ be a graph in $\bisgraphs$. Then 
	approximately counting homomorphisms to~$H$	
	is $\bis$-equivalent under approximation-preserving reductions.
\end{cor}

\begin{rem}
	As a side note, the $\bis$-easiness result given in Corollary~\ref{cor:bisEasyHom} also extends to the problems of approximately counting surjective homomorphisms and approximately counting compactions. This follows from the reductions given in~\cite[Theorem 1.5]{FGZRet}.
\end{rem}

\subsection{Related Work}
\label{sec:relatedwork}

The concept of a retraction has been studied as early as the 1930's and originates in work on continuous functions from topological spaces into subspaces~\cite{Borsuk1931}. Subsequently, retractions between discrete structures, and graphs in particular, have  received a lot of research attention~\cite{HellThesis, Hell1974, HR1987, Pesch1988, FV1999}. See~\cite{HN2008} for an overview.
The algorithmic problem of deciding  whether
a retraction exists can be expressed naturally as a constraint satisfaction problem (CSP) and has been studied  
in e.g.,~\cite{HellNesetrilBook, VikasCompRetCSP, Vikas4Vertex, Vikas2017, FederPseudoForest}. 
Retractions have also been studied in the context of directed graphs. A work by Larose~\cite{LaroseSurvey} surveys CSPs with a digraph as the fixed template. The survey has a special focus on CSPs with additional unary constraints, such as the digraph retraction problem. 
Applications of retractions  reach from classical results such as Brouwer's fixed-point theorem and its equivalent no-retraction theorem~\cite[pp. 272-273]{HYTopology} to specific problems such as solving Sudoku puzzles~\cite{HMSudoku}.

The complexity of deciding whether there is a retraction to a fixed graph $H$ has been studied in different contexts, such as CSPs and list homomorphisms~\cite{FederLHomRefl, FederLHomIrrefl, HN2008,FederPseudoForest, HellNesetrilBook}, surjective homomorphisms~\cite{VikasCompRetCSP, Vikas4Vertex, Vikas2017, BKMsurvey}, as well as pre-colouring extensions and scheduling~\cite{BHT1992, HT1993, HT1996, BJW1994,JS1997,KS1997,Tuza1997,Marx2006,FHH2009}. A complete complexity dichotomy for the retraction decision problem is now known as a consequence of the CSP dichotomy~\cite{BulatovCSPDichotomy,ZhukCSPdichotomy} (assuming $\P \neq \NP$). However, a corresponding graph-theoretical characterisation is not known. 
A characterisation is known \cite{FederPseudoForest} for pseudoforests, which are graphs in which 
each connected component has at most one cycle.

The complexity of exactly counting retractions to a fixed graph $H$ is also classified completely~\cite{DG} (assuming $\FP \neq \numP$). 
In this case, there is a characterisation -- if every connected component of $H$ is trivial, then counting retractions is in $\FP$. Otherwise, counting retractions is $\numP$-complete.

The first result on approximately counting retractions is the following classification for graphs which are both square-free and triangle-free.
\begin{thm}[{\cite[Theorem 1.1]{FGZRet}}]\label{thm:RetGirth5}
	Let $H$ be a graph of girth at least $5$.
	{
		\renewcommand{\theenumi}{\roman{enumi})}
		\renewcommand{\labelenumi}{\theenumi}
		\begin{myenumerate}            
			\item If every connected component of $H$ is an irreflexive star, a single looped vertex, or an edge with two loops, then 
			approximately counting retractions to~$H$			
			is in $\FP$.
			\item Otherwise, if every connected component of $H$ is an irreflexive caterpillar or a partially bristled reflexive path, then 
			approximately counting retractions to~$H$			
			is approximation-equivalent to $\bis$.
			\item Otherwise,  approximately counting retractions to~$H$ is approximation-equivalent to $\sat$.
		\end{myenumerate}
	}
\end{thm}

Counting homomorphisms to square-free graphs has been studied before~\cite{GGRMod2, KBModp} though those results apply to counting modulo a prime number, and so they cannot be applied here.

In the past, graphs without squares have been studied in various other combinatorial settings: 
For example,~\cite{Arends2011} investigates the complexity of finding $101$-colourings of square-free graphs. Polarity graphs are another natural class of square-free graphs that have been studied, e.g.~\cite{Bondy1999,Abreu2010}. Wrochna showed that square-free graphs are multiplicative~\cite{Wrochna2017, WrochnaThesis}, solving a problem related to Hedetniemi's conjecture\footnote{This conjecture has been refuted lately~\cite{Shitov2019}.}~\cite[Problem 7.1]{TardifHedetniemiConjecture} and improving earlier results about the multiplicativity of square-free graphs that contain triangles~\cite{Delhomme2002}. Furthermore, square-free graphs play an important role in extremal graph theory and the study of the Tur\'{a}n number, e.g.~\cite{PikhurkoTuranFunction, Furedi1983, Furedi2013, FirkeSquarefreeExtremalGraphs, GarnickExtremalGraphs, Clapham1989}. Recently, it has been shown that certain bounds on the chromatic number with respect to distance-two colouring of planar graphs are unique to square-free graphs~\cite{Choi2018}.

\subsection{Preliminaries}

For a non-negative integer $k$ we use $[k]$ to denote the set $\{1,\dots, k\}$. For sets $X$ and $Y$ we define $\ucp{X}{Y}= \{\{x,y\} \mid x\in X, y\in Y\}$ as an unordered version of   the Cartesian product. The elements of $\ucp{X}{Y}$ are \emph{multisets} of size exactly $2$. Using this notation the set of edges $E(H)$ of a graph $H=(V(H), E(H))$ is a subset of $\ucp{V(H)}{V(H)}$. An edge with two identical elements is a \emph{loop}. Correspondingly, a vertex $v\in V(H)$ is called \emph{looped} if $\{v,v\}\in E(H)$ and unlooped otherwise. The \emph{girth} of a graph $H$ is the length of a shortest cycle in $H$. All cycles have length at least $3$.

We have already defined reflexive and irreflexive graphs. A graph $H$ is a \emph{mixed} graph if it contains both looped and unlooped vertices, i.e.~if it is neither reflexive nor irreflexive. Given a graph $H$ and a subset $U$ of $V(H)$, $H[U]$ is the \emph{subgraph of $H$ induced by $U$}.

Given graphs $G$ and $H$, $\calH(G,H)$ is the set of homomorphisms from $G$ to $H$
and $\hom{G}{H}$ denotes its size.
Analogously, given a corresponding set of lists $\boldS$, $\calH((G,\boldS),H)$ is the set of homomorphisms from $(G,\boldS)$ to $H$
and  $\hom{(G,\boldS)}{H}$ denotes its size. 

We use $\Ret{H}$ to denote the problem of approximately counting retractions to~$H$
(for a fixed graph~$H$ which may have loops but does not have multi-edges).
We use $\Hom{H}$ to denote the problem of approximately counting homomorphisms to~$H$.
Formally, these problems are defined as follows.

\prob
	{
		$\Ret{H}$.
	}
	{
		An irreflexive graph $G$ and a collection of lists $\boldS=\{S_v\subseteq V(H) \mid v\in V(G)\}$ such that, for all $v\in V(G)$, $\abs{S_v}\in \{1,\abs{V(H)}\}$.
	}
	{
		$\hom{(G,\boldS)}{H}$.
	}
\prob
	{
		$\Hom{H}$.
	}
	{
		An irreflexive graph $G$.
	}
	{
		$\hom{G}{H}$.
	}

The list homomorphisms counting problem, defined as follows, is a generalisation of $\Ret{H}$.

\prob
{
	$\LHom{H}.$
}
{
	An irreflexive graph $G$ and a collection of lists $\boldS=\{S_v\subseteq V(H)\mid v\in V(G)\}$.
}
{
	$\hom{(G,\boldS)}{H}$.
}

If there is an approximation-preserving reduction~\cite{DGGJApprox} from a problem $A$ to a problem $B$, we write $A \leap B$.
\begin{obs}[{\cite[Observation 1.2]{FGZRet}}]\label{obs:HomToRetToLHom}
	Let $H$ be a graph. Then $\Hom{H}\leap\Ret{H}\leap\LHom{H}$.
\end{obs}

For approximately counting list homomorphisms Galanis, Goldberg and Jerrum~\cite{GGJList} give the following complete classification.
\begin{thm}[\cite{GGJList}]\label{thm:LHomTricho}
	Let $H$ be a connected graph. 
	{
		\renewcommand{\theenumi}{\roman{enumi})}
		\renewcommand{\labelenumi}{\theenumi}
		\begin{myenumerate}            
			\item If $H$ is an irreflexive complete bipartite graph or a reflexive complete graph, then
			$\LHom{H}$ is in $\FP$.
			\item 
			Otherwise, if $H$ is an irreflexive bipartite permutation graph or a reflexive proper interval graph, then
			$\LHom{H}$ is approximation-equivalent to $\bis$. 
			\item Otherwise, $\LHom{H}$ is approximation-equivalent to $\sat$.
		\end{myenumerate}
	}
\end{thm}

\subsection{Paper Outline}

Theorem~\ref{thm:bisEasyRet} is proved in Section~\ref{sec:bis-easiness}. Theorem~\ref{thm:RetNoSquare} is proved in Section~\ref{sec:finaltheorems}. The corresponding $\bis$-easiness follows mainly from Theorem~\ref{thm:bisEasyRet}. The corresponding $\sat$-hardness results are collected in Section~\ref{sec:sat-hardness}. Proving $\sat$-hardness is the bulk of this work because of the combinatorial complexity of designing reductions which establish $\sat$-hardness for all square-free graphs
(apart from reflexive cliques, irreflexive caterpillars and those in $\bisgraphs$).

\section{$\bis$-Easiness Results}\label{sec:bis-easiness}
In this section we prove Theorem~\ref{thm:bisEasyRet}, which states that approximately counting retractions to any graph from the class $\bisgraphs$ (Definition~\ref{def:bisgraphs}) is $\bis$-equivalent. The proof is built on a method for generating $\bis$-easiness results from~\cite[Section 2.2.1]{FGZRet} which uses the framework of constraint satisfaction problems. Intuitively, the method takes as input two CSP instances, say $\Iv$ and $\Ie$, and produces a graph $\Hve{\Iv,\Ie}$ for which $\Ret{\Hve{\Iv,\Ie}}\leap \bis$. The challenge is to find the right instances $\Iv$ and $\Ie$ and to identify and generate corresponding general classes of $\bis$-easy graphs.
For the convenience of the reader we repeat some definitions introduced in~\cite{FGZRet}. 
Let $\calL$ be a set of  Boolean relations.

\prob
{
	$\csp(\calL)$.
}
{
	A set of variables $X$ and a set of constraints $C$, where each constraint applies a relation from~$\calL$ to a list of variables from~$X$.
}
{
	The number of assignments $\sigma\from X\to \{0,1\}$ that satisfy all constraints in $C$.
} 

$\Imp=\{(0,0), (0,1), (1,1)\}$ is an arity-two Boolean relation.
The constraint $\Imp(x,y)$  
ensures that, in any satisfying assignment~$\sigma$, we have
$\sigma(x) \implies \sigma(y)$.

\begin{defn}[{\cite[Definition 2.6]{FGZRet}}]\label{def:Hve}
	Let $\Iv=(X,\Cv)$ and $\Ie=(X,\Ce)$ be instances of $\csp(\{\Imp\})$.
	We define the   undirected graph $\Hve{\Iv,\Ie}$ as follows. 
	The vertices of $\Hve{\Iv,\Ie}$ are the satisfying assignments of $\Iv$. 
	Given any  assignments $\sigma$ and $\sigma'$ in $V(\Hve{\Iv,\Ie})$, there is an edge $\{\sigma, \sigma'\}$ in $\Hve{\Iv,\Ie}$ if and only if 
	the following holds:
	For every constraint $\Imp(x,y)$ in $\Ce$,
	we have  $\sigma(x) \implies \sigma'(y)$  and $\sigma'(x) \implies \sigma(y)$.
\end{defn}

\begin{lem}[{\cite[Lemmas 2.5 and 2.8]{FGZRet}}]\label{lem:OALHomToCSPImplies}
	Let 
	$\Iv=(X,\Cv)$ and $\Ie=(X,\Ce)$ be instances of $\csp(\{\Imp\})$.
	Then $\Ret{\Hve{\Iv,\Ie}} \leap \bis$.
\end{lem}

\begin{defn}\label{def:bisgraphs}
	A graph~$H$ is in  $\bisgraphs$ if it can be defined as follows.
	For some positive integer~$Q$,
	the vertex set $V(H)$ is of the form 
	$V(H)= \bigcup_{i=0}^Q K_i \cup \bigcup_{i=1}^Q B_i$
	where 
	$K_0,\ldots,K_{Q}$  induce reflexive cliques in~$H$, 
	and
	$B_1,\ldots,B_Q$ are   disjoint sets of unlooped degree-$1$ vertices (called \emph{bristles}).
	There are $Q+2$ vertices $p_0,\ldots,p_{Q+1}$ in~$V(H)$  such that
	each clique $K_i$ contains both~$p_i$ and~$p_{i+1}$.
	The intersection of the cliques is given as follows.
	\begin{itemize}
		\item For $i\in[Q]$,   $K_{i-1}\cap K_{i} =\{p_i\}$.
		\item For   $i,j\in \{0,\dots, Q\}$ with 
		$\abs{j-i}> 1$,
		$K_{i}\cap K_j =\emptyset$.
	\end{itemize}
	The size of each set $B_i$ of bristles satisfies
	$  0 \leq \abs{B_i}\le \left(\abs{K_{i-1}}-1\right)\cdot \left(\abs{K_i}-1\right)$.
	Finally, the edge set of $H$ is given as follows.
	$$E(H) = \bigcup_{i=0}^Q \left(\ucp{K_i}{K_i}\right) \cup \bigcup_{i=1}^Q \left(\{p_i\} \times B_i\right).$$
\end{defn}

For an example graph from the class $\bisgraphs$ see Figure~\ref{fig:HBISexample}.

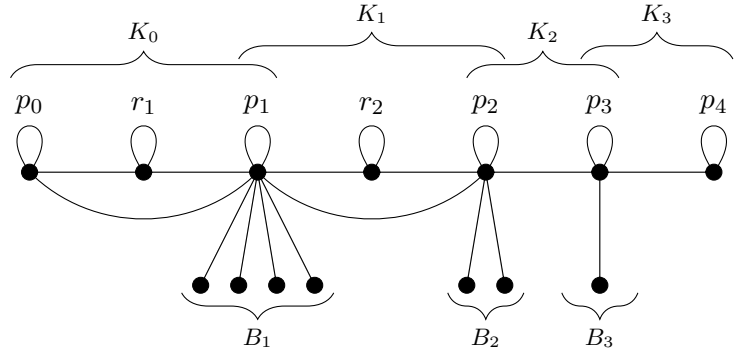
\begin{figure}[h!]\centering
	{\def\scaleFactor{1}
		\begin{tikzpicture}[scale=1, every loop/.style={min distance=10mm,looseness=10}]
		
		\filldraw (0,0) node(p0){} circle[radius=3pt] --++ (0:1.5cm) node(p1){} circle[radius=3pt] --++ (0:1.5cm) node(p2){} circle[radius=3pt] --++ (0:1.5cm) node(p3){} circle[radius=3pt] --++ (0:1.5cm) node(p4){} circle[radius=3pt] --++ (0:1.5cm) node(p5){} circle[radius=3pt] --++ (0:1.5cm) node(p6){} circle[radius=3pt];
		
		\path[-] (p0.center) edge  [in=-135,out=-45] (p2.center);
		\path[-] (p2.center) edge  [in=-135,out=-45] (p4.center);
		
		\path[-] (p0.center) edge  [in=125,out=55,loop] node {} ();
		\path[-] (p1.center) edge  [in=125,out=55,loop] node {} ();
		\path[-] (p2.center) edge  [in=125,out=55,loop] node {} ();
		\path[-] (p3.center) edge  [in=125,out=55,loop] node {} ();		
		\path[-] (p4.center) edge  [in=125,out=55,loop] node {} ();
		\path[-] (p5.center) edge  [in=125,out=55,loop] node {} ();
		\path[-] (p6.center) edge  [in=125,out=55,loop] node {} ();

		\filldraw (p2.center) -- (2.25,-1.5cm) node(b21){} circle[radius=3pt];
		\filldraw (p2.center) -- (2.75,-1.5cm) node(b22){} circle[radius=3pt];
		\filldraw (p2.center) -- (3.25,-1.5cm) node(b23){} circle[radius=3pt];
		\filldraw (p2.center) -- (3.75,-1.5cm) node(b24){} circle[radius=3pt];
		\filldraw (p4.center) -- (5.75,-1.5cm) node(b41){} circle[radius=3pt];
		\filldraw (p4.center) -- (6.25,-1.5cm) node(b42){} circle[radius=3pt];
		\filldraw (p5.center) -- (7.5,-1.5cm) node(b51){} circle[radius=3pt];

		\node[above=.5cm of p0]{$p_0$};
		\node[above=.5cm of p1]{$r_1$};
		\node[above=.5cm of p2]{$p_1$};
		\node[above=.5cm of p3]{$r_2$};
		\node[above=.5cm of p4]{$p_2$};
		\node[above=.5cm of p5]{$p_3$};
		\node[above=.5cm of p6]{$p_4$};

		\draw [decorate,decoration={brace,amplitude=10pt},xshift=0pt,yshift=1.25cm]
		(-.25,0) -- (3.25,0)node [black,midway,xshift=0pt, yshift=.6cm] {\footnotesize
			$K_0$};
		
		\draw [decorate,decoration={brace,amplitude=10pt},xshift=0pt,yshift=1.5cm]
		(2.75,0) -- (6.25,0)node [black,midway,xshift=0pt, yshift=.6cm] {\footnotesize
			$K_1$};
		
		\draw [decorate,decoration={brace,amplitude=10pt},xshift=0pt,yshift=1.25cm]
		(5.75,0) -- (7.75,0)node [black,midway,xshift=0pt, yshift=.6cm] {\footnotesize
			$K_2$};
		
		\draw [decorate,decoration={brace,amplitude=10pt},xshift=0pt,yshift=1.5cm]
		(7.25,0) -- (9.25,0)node [black,midway,xshift=0pt, yshift=.6cm] {\footnotesize
			$K_3$};
		
		\draw [decorate,decoration={brace,amplitude=10pt},xshift=0pt,yshift=.4cm]
		(4,-2) -- (2,-2)node [black,midway,xshift=0pt, yshift=-.6cm] {\footnotesize
			$B_1$};
		
		\draw [decorate,decoration={brace,amplitude=10pt},xshift=0pt,yshift=.4cm]
		(6.5,-2) -- (5.5,-2)node [black,midway,xshift=0pt, yshift=-.6cm] {\footnotesize
			$B_2$};
		
		\draw [decorate,decoration={brace,amplitude=10pt},xshift=0pt,yshift=.4cm]
		(8,-2) -- (7,-2)node [black,midway,xshift=0pt, yshift=-.6cm] {\footnotesize
			$B_3$};
		
		\end{tikzpicture}
	}
	\caption{Example graph from $\bisgraphs$ for $Q=3$. Note that $\abs{B_1}= 4 =  \left(\abs{K_{0}}-1\right)\cdot \left(\abs{K_1}-1\right)$, $\abs{B_2}= 2 =  \left(\abs{K_{1}}-1\right)\cdot \left(\abs{K_2}-1\right)$ and $\abs{B_3}= 1 =  \left(\abs{K_{2}}-1\right)\cdot \left(\abs{K_3}-1\right)$.}
	\label{fig:HBISexample}
\end{figure}

Let $H\in \bisgraphs$ be as defined in Definition~\ref{def:bisgraphs}.
The high-level-structure of the proof of Theorem~\ref{thm:bisEasyRet} is as follows. We first define two instances $\Iv$ and $\Ie$ of $\csp(\{\Imp\})$, then we establish that $H$ is isomorphic to $\Hve{\Iv, \Ie}$ (Lemma~\ref{lem:HisoHIvIe}), which then allows us to apply Lemma~\ref{lem:OALHomToCSPImplies}.

\examplebox{
	To give more intuition we will use a running example where $H$ is the graph depicted in Figure~\ref{fig:HBISexample}. To separate this example from the rest of the proof we use text boxes.
}

Let $V^*$ be the set of looped vertices in $H$, i.e.~$V^*=\bigcup_{i=0}^Q K_i$ and let $X=\{x_v \mid v\in V^*\setminus\{p_0\}\}$ be a set of Boolean variables. We fix an ordering ``$<$'' on the vertices of $V^*$ 
with two properties: (1) In $K_i$, $p_{i}$ is the smallest vertex and $p_{i+1}$ is the largest; 
(2) The order of the $K_i$'s is respected in the
sense that, for any pair of distinct vertices $u,v\in V^*$, if there is an $i<j$ such that   $u\in K_i$ and $v\in K_j$, then $u<v$. 
We define $U=\{\Imp(x_u, x_v) \mid u>v\}$.

Consider $\Iv^* = (X,U)$ and $\Ie^*=(X,U)$.
The graph $\Hve{\Iv^*,\Ie^*}$  is simply a reflexive path on $\abs{V^*}$ vertices. 
We will construct $\Iv$ from $\Iv^*$ by choosing a subset $\Cv$ of $U$ --- this allows the creation of bristles.
Similarly we will construct $\Ie$ from $\Ie^*$ by choosing a subset $\Ce$ of $U$  --- this  creates reflexive cliques
amongst the vertices of the reflexive path in $\Hve{\Iv^*,\Ie^*}$.
In order to define $\Cv$ and $\Ce$, for $i\in \{0, \dots, Q\}$, we define the sets of constraints $\De{i}$ (the constraints that we will delete from $U$ to define $\Ce$) as follows.
\begin{equation}\label{equ:De}
\De{i}= \{\Imp(x_u,x_v) \in U \mid u,v \in K_i\setminus\{p_i\}\}.
\end{equation}
For $i\ge 1$ (i.e.~for $i\in [Q]$), we define the sets of constraints $\Dv{i}$ (the constraints we will delete from $U$ to define $\Cv$). The definition of $\Dv{i}$ is a bit more involved and uses the following sets:
\begin{equation}\label{equ:Ai}
A(i)=\{\Imp(x_u,x_v) \in U \mid 
u\in K_i\setminus\{p_i\},\
v\in K_{i-1}\setminus\{p_{i-1}\}  \}.
\end{equation}
In order to model the set of bristles $B_i$ we will ``delete'' exactly $\abs{B_i}$ constraints that belong to $A(i)$ from $\Cv$. As a means to specify which constraints will be deleted we define an order on the constraints in $U$ (which uses the order $<$ on the vertices in $V^*$ which we fixed previously). There are several orders which would work.

\begin{defn}\label{def:OrderOnU}
	We define an order ``$\preceq$'' on $U$. Let $\Imp(x_u,x_v), \Imp(x_{u'}, x_{v'}) \in U$. Then $\Imp(x_u,x_v) \preceq \Imp(x_{u'}, x_{v'})$ if one of the following holds:
	\begin{itemize}
		\item $u < u'$.
		\item $u=u'$ and $v\ge v'$.
	\end{itemize}
	If $\Imp(x_u,x_v) \preceq \Imp(x_{u'}, x_{v'})$
	and the  ordered pair $(u,v)$ is distinct from the ordered pair $(u',v')$ then  $\Imp(x_u,x_v) \prec \Imp(x_{u'}, x_{v'})$.
\end{defn}
Now let $\Dv{i}$ be the $\abs{B_i}$ smallest elements of $A(i)$ with respect to $\preceq$ as given in Definition~\ref{def:OrderOnU}. Note that by
the definition of~$A(i)$ (Equation~\eqref{equ:Ai}) and the bound on $\abs{B_i}$   (Definition~\ref{def:bisgraphs}) this is well-defined since 
\[
\abs{B_i}\le (\abs{K_{i-1}}-1)\cdot (\abs{K_i}-1) = \abs{A(i)}.
\]
Finally, we   define $\Cv$ and $\Ce$ as follows.
\begin{equation}\label{equ:CvCe}
\Cv=U \setminus \left(\bigcup_{i=1}^Q \Dv{i} \right) \quad \text{ and } \quad \Ce=U \setminus \left(\bigcup_{i=0}^Q \De{i} \right).
\end{equation}

\examplebox{
	In our running example, order the variables in~$V^*$ from left 
	to right. 
	We have $X=\{x_{r_1}, x_{p_1}, x_{r_2}, x_{p_2}, x_{p_3}, x_{p_4}\}$ as the set of variables of $\Iv$ and $\Ie$. Since $\abs{B_1}=4$ the set $\Dv{1}$ contains the $4$ smallest elements of $A(1)=\{\Imp(x_{r_2},x_{p_1}),\allowbreak \Imp(x_{r_2},x_{r_1}),\allowbreak \Imp(x_{p_2},x_{p_1}), \Imp(x_{p_2},x_{r_1})\}$, which means $\Dv{1}=A(1)$. Similarly, since $\abs{B_2}=2$, the set $\Dv{2}$ contains the $2$ smallest elements of $A(2)=\{\Imp(x_{p_3},x_{p_2}),\allowbreak \Imp(x_{p_3},x_{r_2})\}$, which means $\Dv{2}=A(2)$. Finally, since $\abs{B_3}=1$, we have $\Dv{3}=A(3)=\{\Imp(x_{p_4},x_{p_3})\}$. Thus,
	\begin{align}
	\Cv=\{ &\Imp(x_{p_4},x_{p_2}), \Imp(x_{p_4},x_{r_2}), \Imp(x_{p_4},x_{p_1}), \Imp(x_{p_4},x_{r_1}),\nonumber\\ &\Imp(x_{p_3},x_{p_1}), \Imp(x_{p_3},x_{r_1}), \Imp(x_{p_2},x_{r_2}), \Imp(x_{p_1},x_{r_1})\}.\label{equ:CvExample}
	\end{align} 
	Regarding the edge constraints we have $\De{0}=\{\Imp(x_{p_1}, x_{r_1})\}$, $\De{1}=\{\Imp(x_{p_2}, x_{r_2})\}$, $\De{2}=\emptyset$ and $\De{3}=\emptyset$ and hence
	\begin{equation}\label{equ:CeExample}
	\Ce=U\setminus\{\Imp(x_{p_1}, x_{r_1}), \Imp(x_{p_2}, x_{r_2})\}.
	\end{equation}
}

Recall that the satisfying assignments of $\Iv$ correspond to the vertices of  $\Hve{\Iv,\Ie}$.
For $v\in V^*$ let $\sigma_v\from X \to \{0,1\}$ be the assignment with 
\[
\sigma_v(x_u)=
\begin{cases}
1, &\text{if } u\le v\\
0, &\text{otherwise,}
\end{cases}
\]
where $u\in V^*\setminus\{p_0\}$ (i.e.~$x_u\in X$). Note that $\sigma_{p_0}$ is the all-zero assignment since $p_0$ is the minimum vertex in $V^*$. 
The reason that we did not introduce a variable for~$p_0$ in the definition of~$X$ is 
that its role is captured by the all-zero assignment to~$X$. 
We will call assignments of the form $\sigma_v$ \emph{path assignments}. 
The path assignments inherit an order from the order on the set $V^*$ that we fixed, i.e.~$\sigma_u < \sigma_v$ if and only if $u<v$.

\begin{lem}\label{lem:canonicalPath}
	All path assignments satisfy $\Iv=(X, \Cv)$. If $\sigma_{p_0}, \dots, \sigma_{p_{Q+1}}$ are the path assignments ordered by $<$, then they form a reflexive path in $\Hve{\Iv, \Cv}$. The reflexive path is not necessarily induced by its vertices. 
\end{lem}
\begin{proof} 
	This proof merely requires that $\Cv \subseteq U$ and $\Ce \subseteq U$  --- it does not use the detailed definitions of~$\Cv$ and~$\Ce$.
	
	First, note that   the assignments $\sigma_v$ satisfy the $\csp(\{\Imp\})$-instance $(X,U)$.
	Thus,    they will still be satisfying assignments if we delete constraints from $U$.
	
	We now investigate the edges between the vertices $\sigma_{p_0}, \dots, \sigma_{p_{Q+1}}$ of $\Hve{\Iv,\Ie}$.
	From Definition~\ref{def:Hve} recall that,
	given any satisfying assignments $\sigma$ and $\sigma'$ of $\Iv$, there is an edge $\{\sigma, \sigma'\}$ in $\Hve{\Iv,\Ie}$ if and only if the following holds:
	\begin{equation}\label{equ:adjacentAssignments}
	\text{For every constraint $\Imp(x,y)$ in $\Ce$,
		we have  $\sigma(x) \implies \sigma'(y)$  and $\sigma'(x) \implies \sigma(y)$.}
	\end{equation}
	Using~\eqref{equ:adjacentAssignments} it can easily be checked that for $\Ce=U$ (and therefore for all $\Ce \subseteq U$), the vertices $\sigma_{p_0}, \dots, \sigma_{p_{Q+1}}$ form a reflexive path, i.e.~these vertices are looped and for $v\in V^*\setminus\{p_0\}$ and $v'= \max \{u \in V^* \mid u < v\}$ we have that $\{\sigma_{v'}, \sigma_{v}\}$ is an edge in $\Hve{\Iv,\Ie}$.
\end{proof}

Lemma~\ref{lem:canonicalPath} shows that even if we were to enforce all constraints in $U$ both in the modelling of vertices ($\Iv^*=(X,U)$) and of edges ($\Ie^*=(X,U)$), the graph $\Hve{\Iv^*, \Ie^*}$ always contains a reflexive path on $\abs{V^*}$ vertices. 
As noted earlier,   $\Hve{\Iv^*, \Ie^*}$ is precisely this reflexive path. 
Our definition of $\Ce$, given in~\eqref{equ:CvCe}, 
ensures that $\Ce \subseteq U$ so it uses a subset of the constraints,
which leads to the possibilitiy of more edges in $\Hve{\Iv, \Ie}$. The idea behind the construction is that a vertex $v\in V^*$ will be modelled by the satisfying assignment $\sigma_{v}$, which is a vertex in $\Hve{\Iv, \Ie}$. The previous lemma shows that these assignments form a  reflexive path. We now show how the clique structure is modelled. Intuitively,  the deleted constraints allow shortcuts that form cliques along the  underlying path of path assignments. 
The following observation follows immediately from the definition of~$\De{i}$ in~\eqref{equ:De}.

\begin{obs}\label{obs:De}
	For all $i\in [Q]$, $v\in V^*$ with $p_i<v$ and all $k\in \{0, \dots, Q\}$,  $\Imp(x_{v}, x_{p_i})\notin \De{k}$ and consequently $\Imp(x_{v}, x_{p_i}) \in \Ce$ (since $\Imp(x_{v}, x_{p_i}) \in U$).
\end{obs} 

\begin{lem}\label{lem:cliqueStructure}
	Two vertices $u,v \in V^*$ are adjacent in~$H$ if and only if $\sigma_{u}$~is adjacent to~$\sigma_{v}$ in~$\Hve{\Iv,\Ie}$.
\end{lem}
\begin{proof}
	First we consider the case where $u,v \in V^*$ are adjacent. Then there exists an index $i$ such that   $u,v \in K_i$. Consequently $p_i \le u\le p_{i+1}$ and $p_i \le v\le p_{i+1}$. Therefore, for $w\in V^*$, if $w\le p_{i}$ then $\sigma_{u}(x_w)=1$ and $\sigma_{v}(x_w)=1$, and if $w> p_{i+1}$ then $\sigma_{u}(x_w)=0$ and $\sigma_{v}(x_w)=0$. 
	Let $\Imp(x_t,x_s)$ be some constraint in $\Ce$. Then $\Imp(x_t,x_s)\notin \De{i}$ and, by
	the Definition of~$\De{i}$ in~\eqref{equ:De}, we have at least one of $s\le p_i$ or $t> p_{i+1}$. Using these facts, we can verify~\eqref{equ:adjacentAssignments} for $\sigma=\sigma_{u}$ and $\sigma'=\sigma_{v}$ which shows that $\sigma_{u}$ is adjacent to $\sigma_{v}$ in $\Hve{\Iv,\Ie}$. 
	
	Now assume that $u$ and $v$ are not adjacent. Then there exist indices $i\neq j$ from $\{0,\ldots,Q\}$ such that $u\in K_i\setminus\{p_{i+1}\}$, 
	$v\notin K_i$, $v \in K_j\setminus\{p_{j+1}\}$ and $u\notin K_j$. 	 
	Without loss of generality assume $i<j$. Then $u<p_{i+1}< v$ so $\sigma_{v}(x_{v})=1$ and $\sigma_{u}(x_{p_{i+1}})=0$. However,  by Observation~\ref{obs:De}, $\Imp(x_{v}, x_{p_{i+1}}) \in \Ce$ and consequently $\sigma_{u}$ and $\sigma_{v}$ are not adjacent in $\Hve{\Iv,\Ie}$.
\end{proof}

\examplebox{
	In our example, the vertices $\sigma_{p_0}, \sigma_{r_1}, \sigma_{p_1}, \sigma_{r_2}, \sigma_{p_2}, \sigma_{p_3}, \sigma_{p_4}$ form a reflexive path
	(as guaranteed by Lemma~\ref{lem:canonicalPath}). Furthermore, $\{\sigma_{p_0}, \sigma_{p_1}\}$ is an edge since $\Imp(x_{p_1}, x_{r_1})\notin \Ce$ by~\eqref{equ:CeExample}, and $\{\sigma_{p_1}, \sigma_{p_2}\}$ is an edge since $\Imp(x_{p_2}, x_{r_2})\notin \Ce$ by~\eqref{equ:CeExample}. The subgraph of $\Hve{\Iv, \Ie}$ induced by the path assignments is as depicted in Figure~\ref{fig:Hsigma} and is isomorphic to $H[V^*]=H[\{p_0, r_1,p_1,r_2,p_2,p_3, p_4\}]$ (as is guaranteed by Lemma~\ref{lem:cliqueStructure}).

	\bigskip
	\begin{center}
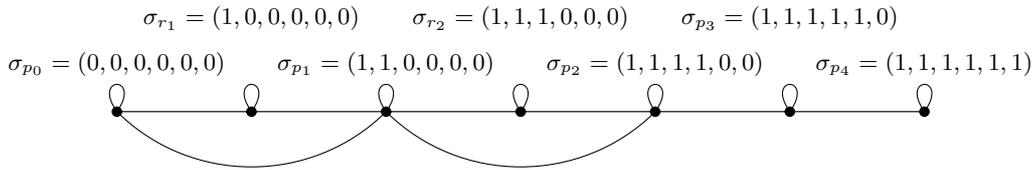

		\footnotesize%
		\begin{tikzpicture}[scale=0.59, every loop/.style={min distance=10mm,looseness=10}]
		
		\filldraw (0,0) node(p0){} circle[radius=3pt] --++ (0:3cm) node(p1){} circle[radius=3pt] --++ (0:3cm) node(p2){} circle[radius=3pt] --++ (0:3cm) node(p3){} circle[radius=3pt] --++ (0:3cm) node(p4){} circle[radius=3pt] --++ (0:3cm) node(p5){} circle[radius=3pt] --++ (0:3cm) node(p6){} circle[radius=3pt];
		
		\path[-] (p0.center) edge  [in=-135,out=-45] (p2.center);
		\path[-] (p2.center) edge  [in=-135,out=-45] (p4.center);
		
		\path[-] (p0.center) edge  [in=125,out=55,loop] node {} ();
		\path[-] (p1.center) edge  [in=125,out=55,loop] node {} ();
		\path[-] (p2.center) edge  [in=125,out=55,loop] node {} ();
		\path[-] (p3.center) edge  [in=125,out=55,loop] node {} ();		
		\path[-] (p4.center) edge  [in=125,out=55,loop] node {} ();
		\path[-] (p5.center) edge  [in=125,out=55,loop] node {} ();
		\path[-] (p6.center) edge  [in=125,out=55,loop] node {} ();
		
		\node[above=.25cm of p0]{$\sigma_{p_0}=(0,0,0,0,0,0)$};
		\node[above=.85cm of p1]{$\sigma_{r_1}=(1,0,0,0,0,0)$};
		\node[above=.25cm of p2]{$\sigma_{p_1}=(1,1,0,0,0,0)$};
		\node[above=.85cm of p3]{$\sigma_{r_2}=(1,1,1,0,0,0)$};
		\node[above=.25cm of p4]{$\sigma_{p_2}=(1,1,1,1,0,0)$};
		\node[above=.85cm of p5]{$\sigma_{p_3}=(1,1,1,1,1,0)$};
		\node[above=.25cm of p6]{$\sigma_{p_4}=(1,1,1,1,1,1)$};
		\end{tikzpicture}
		\captionof{figure}{The subgraph of $\Hve{\Iv,\Ie}$ induced by the path assignments. If a vertex corresponds to a path assignment $\sigma_v$ of $\Iv$ then its label is of the form $\sigma_v=(\sigma_v(x_{r_1}), \sigma_v(x_{p_1}), \sigma_v(x_{r_2}), \sigma_v(x_{p_2}), \sigma_v(x_{p_3}), \sigma_v(x_{p_4}))$.}
		\label{fig:Hsigma}
	\end{center}
}

With Lemma~\ref{lem:cliqueStructure} we have established the clique structure of the looped vertices. Next we will model the bristles that are attached to the vertices $p_1, \dots, p_Q$. To do this, we consider satisfying assignments of $\Iv$ that are not path assignments, i.e.~that are not of the form $\sigma_v$.  These assignments are called \emph{bristle assignments}.

\begin{defn}
	An assignment $\beta \from X \to \{0, 1\}$ is a bristle assignment if and only if there exist $u,v\in V^*\setminus\{p_0\}$ with $u<v$ such that $\beta(x_u)=0$ and $\beta(x_v)=1$.
\end{defn}

\begin{defn}\label{def:good}
	For $i\in [Q]$, $a \in K_{i-1}\setminus\{p_{i-1}\}$ and $b \in K_{i}\setminus\{p_{i}\}$ let $\beta_i[a, b]\from X \to \{0,1\}$ be the assignment with 
	\begin{equation}\label{equ:goodbristle}
	\beta_i[a, b](x_u)=
	\begin{cases}
	1, &\text{if } u<a\\
	0, &\text{if } a\le u \le p_i\\
	1, &\text{if } p_i < u\le b\\
	0, &\text{otherwise (if $b < u$),}
	\end{cases}
	\end{equation}
	where $u\in V^*\setminus\{p_0\}$. 
	Hence $\beta_i[a, b]$ is of the following form.
	\begin{center}
		\begin{tabular}{l|*{13}{c}}
			$x_u$
			& &$\dots$ &$ $
			&$x_{a}$ &$\dots$ &$x_{p_i}$
			& &$\dots$ &$x_{b}$
			& &$\dots$ &$x_{p_{Q+1}}$
			\\
			$\beta_{i}[a, b](x_u)$
			&$1$ & $\dots$ & $1$
			&$0$ &$\dots$ &$0$
			&$1$ &$\dots$ &$1$ 
			&$0$ &$\dots$ &$0$
		\end{tabular}
	\end{center}
	Note that $a$ is the minimum index for which $\beta_i[a, b]$ takes value $0$, and $b$ is the maximum index for which $\beta_i[a, b]$ takes value $1$.
	We say that a bristle assignment is \emph{good} if it is of the form $\beta_i[a, b]$.
\end{defn}

\begin{obs}\label{obs:satBristleEquivalence}
	Since $\Cv\subseteq U$ the following are equivalent:
	\begin{enumerate}[(I)]
		\item $\beta_i[a,b]$ is a satisfying assignment of $\Iv$.
		\item For all $s,t \in V^*\setminus\{p_0\}$ with $a\le s \le p_i$ and $p_i < t \le b$ it holds that $\Imp(x_t,x_s) \notin \Cv$.
	\end{enumerate}
\end{obs}

\begin{lem}\label{lem:goodBristle}
	Every bristle assignment that satisfies $\Iv=(X,\Cv) $ is good.
\end{lem}
\begin{proof}
	The property of $\Cv$ that is used in this proof is that $U\setminus \bigcup_{i\in [Q]} A(i) \subseteq \Cv$ and therefore if a constraint from $U$ is not in $\Cv$ it has to be in one of the sets $A(i)$.
	
	Let $\beta$ be a bristle assignment that satisfies $\Iv$. Let $a$ be the minimum index with $\beta(x_a)=0$ and let $b$ be the maximum index with $\beta(x_b)=1$. 
	Since $\beta$ is a bristle assignment, $b>a$.	
	Then, since $\beta$ is a satisfying assignment, $\Imp(x_b, x_a) \notin \Cv$. Therefore $\Imp(x_b, x_a) \in A(i)$ for some $i\in [Q]$, and therefore 
	(by the definition of~$A(i)$), $b\in K_i\setminus\{p_i\}$
	and
	$a \in K_{i-1} \setminus\{p_{i-1}\}$. 
	This is consistent with Definition~\ref{def:good}.
	We will show that $\beta=\beta_i[a, b]$.
	Let $u\in V^*\setminus\{p_0\}$. We investigate the value of $\beta(x_u)$ depending on $u$ to show that $\beta$ takes values as given in~\eqref{equ:goodbristle}:
	\begin{itemize}
		\item If $u<a$, by the minimality of $a$ we have $\beta(x_u)=1$.
		\item If $u=a$, then $\beta(x_u)=\beta(x_a)=0$ by the choice of $a$.
		\item If $a< u\le p_i$, then $u\notin K_{i}\setminus\{p_i\}$ which implies $\Imp(x_u,x_a)\notin A(i)$ and consequently, since $\Imp(x_u,x_a)$ is in $U$, we have $\Imp(x_u,x_a) \in \Cv$ (also using the fact that $\Imp(x_u,x_a)$ cannot be in any of the other sets $A(k)$ since $a \in K_{i-1}\setminus\{p_{i-1}\}$). Therefore it holds that $\beta(x_u)=0$.
		\item If $p_i <u < b$, then $u\notin K_{i-1}$ which implies $\Imp(x_b,x_u)\notin A(i)$ and consequently, since $\Imp(x_b,x_u)$ is in $U$, we have $\Imp(x_b,x_u) \in \Cv$ (also using the fact that $\Imp(x_b,x_u)$ cannot be in any of the other sets $A(k)$ since $b \in K_i\setminus\{p_i\}$). Therefore it holds that $\beta(x_u)=1$.
		\item If $u=b$, then $\beta(x_u)=\beta(x_b)=1$ by the choice of $b$.
		\item If $b<u$, by the maximality of $b$ we have $\beta(x_u)=0$.
	\end{itemize}
\end{proof}

\examplebox{
	We can now check which  of the (good) bristle assignments satisfy $\Iv$ in the running example. First, consider the set $\Dv{1}$. 
	We already showed, in the text box containing~\eqref{equ:CvExample}, that $\Dv{1}= \{\Imp(x_{r_2},x_{p_1}),\allowbreak \Imp(x_{r_2},x_{r_1}),\allowbreak \Imp(x_{p_2},x_{p_1}), \Imp(x_{p_2},x_{r_1})\}$. From the definition of $\Cv$
	(Equation~\eqref{equ:CvCe}),
	\begin{enumerate}[(i)]
		\setlength{\itemsep}{0pt}
		\setlength{\parskip}{0pt}
		\setlength{\parsep}{0pt} 
		\item $\Imp(x_{r_2},x_{p_1}) \notin \Cv$. 	\label{item:A11}
		\item $\Imp(x_{r_2},x_{r_1}) \notin \Cv$.	\label{item:A12}
		\item $\Imp(x_{p_2},x_{p_1}) \notin \Cv$.	\label{item:A13}
		\item $\Imp(x_{p_2},x_{r_1}) \notin \Cv.$	\label{item:A14}
	\end{enumerate} 
	The good bristle assignments of the form $\beta_1[a,b]$ 
	have $a \in K_{0}\setminus \{p_{0}\} = \{r_1,p_1\}$ and $b \in K_{1}\setminus\{p_{1}\}=\{r_2,p_2\}$	
	so they	
	are $\beta_1[p_1,r_2]$, $\beta_1[r_1,r_2]$, $\beta_1[p_1,p_2]$ and $\beta_1[r_1,p_2]$, which are as follows:
	\begin{center}
		\begin{tabular}{l|*{7}{c}}
			$x_u$ &$\begingroup \color{red}\bm{x_{r_1}}\endgroup$ &$\begingroup \color{red}\bm{x_{p_1}}\endgroup$ &$\begingroup \color{red}\bm{x_{r_2}}\endgroup$ & $\begingroup \color{red}\bm{x_{p_2}}\endgroup$ & $x_{p_3}$& $x_{p_4}$\\
			$\beta_1[p_1,r_2](x_u)$ &$\begingroup \color{red}\bm{1}\endgroup$ &$\begingroup \color{red}\bm{0}\endgroup$ &$\begingroup \color{red}\bm{1}\endgroup$ &$\begingroup \color{red}\bm{0}\endgroup$ &$0$ &$0$\\
			$\beta_1[r_1,r_2](x_u)$ &$\begingroup \color{red}\bm{0}\endgroup$ &$\begingroup \color{red}\bm{0}\endgroup$ &$\begingroup \color{red}\bm{1}\endgroup$ &$\begingroup \color{red}\bm{0}\endgroup$ &$0$ &$0$\\
			$\beta_1[p_1,p_2](x_u)$ &$\begingroup \color{red}\bm{1}\endgroup$ &$\begingroup \color{red}\bm{0}\endgroup$ &$\begingroup \color{red}\bm{1}\endgroup$ &$\begingroup \color{red}\bm{1}\endgroup$ &$0$ &$0$\\
			$\beta_1[r_1,p_2](x_u)$ &$\begingroup \color{red}\bm{0}\endgroup$ &$\begingroup \color{red}\bm{0}\endgroup$ &$\begingroup \color{red}\bm{1}\endgroup$ &$\begingroup \color{red}\bm{1}\endgroup$ &$0$ &$0$
		\end{tabular}
	\end{center}

	We now apply Observation~\ref{obs:satBristleEquivalence}.	
	From~\eqref{item:A11} it follows that $\beta_1[p_1,r_2]$ satisfies $\Iv$. Similarly, from~\eqref{item:A11} and~\eqref{item:A12} it follows that $\beta_1[r_1,r_2]$ satisfies $\Iv$. From~\eqref{item:A11} and~\eqref{item:A13} it follows that $\beta_1[p_1,p_2]$ satisfies $\Iv$. Finally, from~\eqref{item:A11},~\eqref{item:A12},~\eqref{item:A13} and~\eqref{item:A14} it follows that $\beta_1[r_1,p_2]$ satisfies $\Iv$.
	The four bristle assignments that we	have checked correspond to the four bristles  in~$B_1$ in Figure~\ref{fig:HBISexample}.	
	Similarly, one can check, from the definitions of $\Dv{2}$ and $\Dv{3}$, that $\beta_2[p_2,p_3]$, $\beta_2[r_2,p_3]$ and $\beta_3[p_3,p_4]$ are the only other satisfying bristle assignments.
}

\begin{lem}\label{lem:onlyNeighbourPi}
	For every bristle assignment $\beta$ that satisfies $\Iv=(X,\Cv)$ there exists $i\in [Q]$ such that $\sigma_{p_i}$ is the only neighbour of $\beta$ in $\Hve{\Iv,\Ie}$.
\end{lem}
\begin{proof}
	Here we will use the fact that $U\setminus \bigcup_{i\in [Q]} A(i) \subseteq \Cv$ and the fact that $\Ce=U\setminus \left(\bigcup_{k=0}^Q D_e(k)\right)$. 
	
	Let $\beta$ be a bristle assignment that satisfies $\Iv$. By Lemma~\ref{lem:goodBristle}, $\beta$ is good, i.e.~there exist $i\in [Q]$, $a \in K_{i-1}\setminus\{p_{i-1}\}$ and $b \in K_{i}\setminus\{p_{i}\}$ such that $\beta=\beta_{i}[a, b]$. We will show that $\sigma_{p_i}$ is the only neighbour of $\beta=\beta_{i}[a, b]$ in $\Hve{\Iv,\Ie}$. 
	
	Let $\psi$ be some satisfying assignment of $\Iv$ which is adjacent to $\beta_{i}[a, b]$. Note that $\beta_{i}[a, b](x_{p_i})=0$. This fact together with Observation~\ref{obs:De} (which states that for all $u\in V^*$ with $p_i<u$ we have $\Imp(x_u, x_{p_i}) \in \Ce$) implies that,
	for all $u>p_i$, $\psi(x_u)=0$. 
	
	Let $v$ be the minimum vertex with $p_i<v$. Then $p_i < v \le b$ and therefore $\beta_{i}[a, b](x_{v})=1$. 
	The  choice of~$v$ ensures that for all vertices $u\in V^*$,
	$u\leq p_i$ iff $u<v$.	  Therefore, if $u\le p_i$ we have $\Imp(x_{v}, x_u) \in \Ce$ (since $v\in K_i\setminus\{p_i\}$ but $u\notin K_i\setminus\{p_i\}$ and hence $\Imp(x_{v}, x_u) \notin \De{i}$) and consequently $\psi(x_u)=1$. 
	Summarising, we obtain
	\begin{equation*}
	\text{for all $u\leq p$, } 
	\psi(x_u)=1 
	\text{ and for all $u>p_i$, }
	\psi(x_u)=0.
	\end{equation*}
	Thus, $\psi=\sigma_{p_i}$. 
	
	It remains to check that $\beta_{i}[a, b]$ and $\sigma_{p_i}$ are in fact adjacent. To this end we verify~\eqref{equ:adjacentAssignments}:
	
	\bigskip
	\noindent{\bf\boldmath Claim:
		If $\Imp(x_t,x_s) \in \Ce$ then $\beta_{i}[a, b](x_t) \Rightarrow \sigma_{p_i}(x_s)$.
	}
	
	\smallskip
	\noindent{\it Proof of the claim:}
	We check for possible violations.
	The only relevant $s$ and $t$ are those for which $s<t$, $\beta_{i}[a, b](x_t)=1$ and $\sigma_{p_i}(x_s)=0$, 
	i.e., all~$s$ and $t$ satisfying $p_i<s <t \leq b \leq p_{i+1}$.
	However, constraints of this form are in $\De{i}$ and hence are not in $\Ce$.
	{\it (End of the proof of the claim.)}
	
	\bigskip
	\noindent{\bf\boldmath Claim: If $\Imp(x_t,x_s) \in \Ce$ then $\sigma_{p_i}(x_t) \Rightarrow \beta_{i}[a, b](x_s)$.
	}
	
	\smallskip
	\noindent{\it Proof of the claim:}
	Again we check for violations. The only relevant $s$ and $t$ are those for which $s<t$, $\sigma_{p_i}(x_t)=1$ and $\beta_{i}[a, b](x_s)=0$,  
	i.e., all $s$ and $t$ satisfying $p_{i-1} < a \leq s < t \leq p_i$.
	However, constraints of this form are in $\De{i-1}$ and hence are not in $\Ce$.
	{\it (End of the proof of the claim.)}
\end{proof}

\begin{lem}\label{lem:numberOfBristles}
	For each $i\in[Q]$ there are exactly $\abs{B_i}$ good bristle assignments that satisfy $\Iv$ and are adjacent to $\sigma_{p_i}$ in $\Hve{\Iv, \Ie}$.
\end{lem}
\begin{proof}
	Every good bristle assignment is of the form $\beta_{i}[a, b]$ for some $i\in[Q]$, $a \in K_{i-1}\setminus\{p_{i-1}\}$ and $b \in K_{i}\setminus\{p_{i}\}$. In the proof of Lemma~\ref{lem:onlyNeighbourPi} we have shown that the only neighbour of $\beta_{i}[a, b]$ is $\sigma_{p_i}$. Then the following claim completes the proof of the lemma:
	
	\bigskip
	\noindent{\bf\boldmath Claim:
		$\beta_{i}[a, b]$ satisfies $\Iv$ if and only if $\Imp(x_b, x_{a})$ is among the $\abs{B_i}$ smallest elements of $A(i)$.
	}
	
	\smallskip
	\noindent{\it Proof of the claim:}
	Since $a \in K_{i-1}\setminus\{p_{i-1}\}$ and $b \in K_{i}\setminus\{p_{i}\}$ we have $\Imp(x_b, x_a)\in A(i)$.
	
	First consider the case where $\Imp(x_b, x_a)$ is one of the $\abs{B_i}$ smallest elements of $A(i)$. From Observation~\ref{obs:satBristleEquivalence} we know that $\beta_{i}[b, a]$ satisfies $\Iv$ if and only if for all $s,t \in V^*\setminus\{p_0\}$ with $a\le s \le p_i$ and $p_i < t \le b$ it holds that $\Imp(x_t,x_s) \notin \Cv$. Note that each such $\Imp(x_t,x_s)$ is in $A(i)$ and 
	by Definition~\ref{def:OrderOnU}, 
	$\Imp(x_t,x_s) \preceq \Imp(x_b,x_a)$. Thus, $\Imp(x_t,x_s)\in D_v(i)$ and we obtain $\Imp(x_t,x_s) \notin \Cv$ as required.
	
	Now consider the remaining case where $\Imp(x_b, x_a)$ is not among the the $\abs{B_i}$ smallest elements of $A(i)$. Then, by our choice of $D_v(i)$, $\Imp(x_b, x_a)\notin D_v(i)$. Moreover, for all $k\neq i$, $\Imp(x_b, x_a)\notin D_v(k)$ since $\Imp(x_b, x_a)\in A(i)$, $D_v(k)\subseteq A(k)$ and $A(k)\cap A(i) = \emptyset$. Hence $\Imp(x_b, x_a)\in \Cv$ and consequently $\beta_{i}[b, a]$ does not satisfy $\Iv$ by Observation~\ref{obs:satBristleEquivalence}.
	{\it (End of the proof of the claim.)}
\end{proof}

\begin{lem}\label{lem:HisoHIvIe}
	Let $\Iv= (X,\Cv)$ and $\Ie=(X,\Ce)$. Then $\Hve{\Iv,\Ie}$ is isomorphic to~$H$.
\end{lem}
\begin{proof}
	Here we collect the previous results. By Lemma~\ref{lem:canonicalPath} the path assignments satisfy $\Iv$ (and hence are vertices of $\Hve{\Iv,\Ie}$). Lemma~\ref{lem:cliqueStructure} shows that the subgraph of $\Hve{\Iv,\Ie}$ induced by the path assignments is isomorphic to $H[V^*]$. By Lemma~\ref{lem:goodBristle} all other satisfying assignments are good bristle assignments and by Lemma~\ref{lem:onlyNeighbourPi} each of these good bristle assignments has a unique neighbour, which is among $\sigma_{p_1}, \dots, \sigma_{p_Q}$. Then Lemma~\ref{lem:numberOfBristles} shows that there are exactly $\abs{B_i}$ good bristle assignments adjacent to $\sigma_{p_i}$. 
\end{proof}

We can now prove Theorem~\ref{thm:bisEasyRet}, which we re-state at this point for the convenience of the reader.
{\renewcommand{\thethm}{\getrefnumber{thm:bisEasyRet}}
	\begin{thm}
		\ThmbisEasyRet
	\end{thm}
	\addtocounter{thm}{-1}
}
\begin{proof}
	Let $H\in \bisgraphs$. The $\bis$-easiness part of Theorem~\ref{thm:bisEasyRet} follows directly from Lemmas~\ref{lem:HisoHIvIe} and~\ref{lem:OALHomToCSPImplies}. The $\bis$-hardness part follows from $\Hom{H}\leap \Ret{H}$ (Observation~\ref{obs:HomToRetToLHom}) together with the fact that $\bis \leap \Hom{H}$ for all connected graphs $H$ other than reflexive cliques and irreflexive stars~\cite[Theorem 1]{GGJBIS}.
\end{proof}

\examplebox{
	In our running example we conclude by showing in Figure~\ref{fig:HveExample} how $H$ is encoded as $\Hve{\Iv, \Ie}$. We have already demonstrated which assignments are satisfying and thus are vertices of $\Hve{\Iv, \Ie}$. 
	Using~\eqref{equ:CeExample} and~\eqref{equ:adjacentAssignments} it is straightforward to verify (with some work) that each satisfying bristle assignment of the form $\beta_i[a,b]$ is in fact adjacent only to $\sigma_{p_i}$.
	\begin{center}
		\hspace*{-.2cm} \footnotesize%
		\begin{tikzpicture}[scale=0.6, every loop/.style={min distance=10mm,looseness=10}]
		
		\filldraw (0,0) node(p0){} circle[radius=3pt] --++ (0:3cm) node(p1){} circle[radius=3pt] --++ (0:3cm) node(p2){} circle[radius=3pt] --++ (0:3cm) node(p3){} circle[radius=3pt] --++ (0:3cm) node(p4){} circle[radius=3pt] --++ (0:3cm) node(p5){} circle[radius=3pt] --++ (0:3cm) node(p6){} circle[radius=3pt];
		
		\path[-] (p0.center) edge  [in=-135,out=-45] (p2.center);
		\path[-] (p2.center) edge  [in=-135,out=-45] (p4.center);
		
		\path[-] (p0.center) edge  [in=125,out=55,loop] node {} ();
		\path[-] (p1.center) edge  [in=125,out=55,loop] node {} ();
		\path[-] (p2.center) edge  [in=125,out=55,loop] node {} ();
		\path[-] (p3.center) edge  [in=125,out=55,loop] node {} ();		
		\path[-] (p4.center) edge  [in=125,out=55,loop] node {} ();
		\path[-] (p5.center) edge  [in=125,out=55,loop] node {} ();
		\path[-] (p6.center) edge  [in=125,out=55,loop] node {} ();
		
		\filldraw (p2.center) -- (5.25,-2.25cm) node(b21){} circle[radius=3pt];
		\filldraw (p2.center) -- (5.75,-2.25cm) node(b22){} circle[radius=3pt];
		\filldraw (p2.center) -- (6.25,-2.25cm) node(b23){} circle[radius=3pt];
		\filldraw (p2.center) -- (6.75,-2.25cm) node(b24){} circle[radius=3pt];
		\filldraw (p4.center) -- (11.75,-2.25cm) node(b41){} circle[radius=3pt];
		\filldraw (p4.center) -- (12.25,-2.25cm) node(b42){} circle[radius=3pt];
		\filldraw (p5.center) -- (15,-2.25cm) node(b51){} circle[radius=3pt];
		
		\node[above=.3cm of p0]{$\sigma_{p_0}=(0,0,0,0,0,0)$};
		\node[above=.7cm of p1]{$\sigma_{r_1}=(1,0,0,0,0,0)$};
		\node[above=.3cm of p2]{$\sigma_{p_1}=(1,1,0,0,0,0)$};
		\node[above=.7cm of p3]{$\sigma_{r_2}=(1,1,1,0,0,0)$};
		\node[above=.3cm of p4]{$\sigma_{p_2}=(1,1,1,1,0,0)$};
		\node[above=.7cm of p5]{$\sigma_{p_3}=(1,1,1,1,1,0)$};
		\node[above=.3cm of p6]{$\sigma_{p_4}=(1,1,1,1,1,1)$};

		\node at (4.5,-3) {$\beta_1[p_1,r_2]=(\begingroup \color{red}\bm{1,0,1,0}\endgroup,0,0)$};
		\node at (4.5,-3.7) {$\beta_1[r_1,r_2]=(\begingroup \color{red}\bm{0,0,1,0}\endgroup,0,0)$};
		\node at (4.5,-4.4) {$\beta_1[p_1,p_2]=(\begingroup \color{red}\bm{1,0,1,1}\endgroup,0,0)$};
		\node at (4.5,-5.1) {$\beta_1[r_1,p_2]=(\begingroup \color{red}\bm{0,0,1,1}\endgroup,0,0)$};
		\node at (11,-3) {$\beta_2[p_2,p_3]=(1,1,\begingroup \color{red}\bm{1,0,1,0}\endgroup)$};
		\node at (11,-3.7) {$\beta_2[r_2,p_3]=(1,1,\begingroup \color{red}\bm{0,0,1,0}\endgroup)$};
		\node at (17.5,-3) {$\beta_3[p_3,p_4]=(1,1,1,\begingroup \color{red}\bm{1,0,1}\endgroup)$};
		\end{tikzpicture}
		
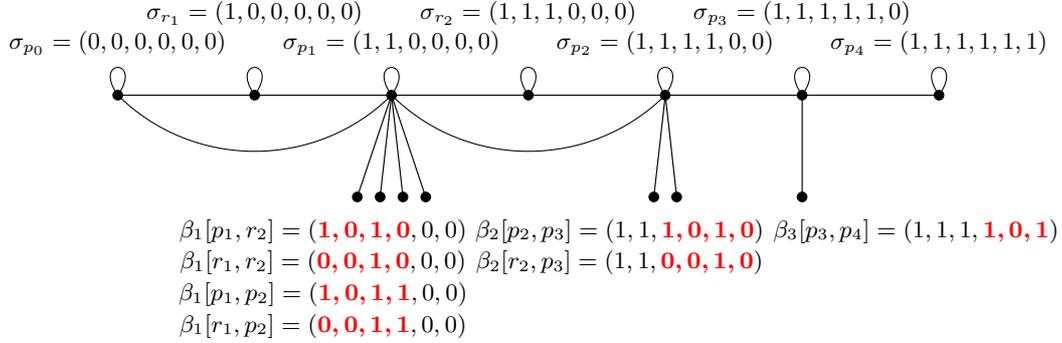
\captionof{figure}{The graph $\Hve{\Iv,\Ie}$. If a vertex corresponds to a satisfying assignment $\rho$ of $\Iv$ then its label is of the form $\rho=(\rho(x_{r_1}),\rho(x_{p_1}),\rho(x_{r_2}),\rho(x_{p_2}),\rho(x_{p_3}),\rho(x_{p_4}))$.}
		\label{fig:HveExample}
	\end{center}
}

\section{$\sat$-Hardness Results}\label{sec:sat-hardness}

From previous results we already know the following:
\begin{itemize}
	\item Theorem~\ref{thm:RetGirth5} classifies the complexity of approximately counting retractions to $H$ for all graphs $H$ that are both square-free and triangle-free (i.e.~have girth at least $5$).
	\item For irreflexive $H$, the proof of Theorem~\ref{thm:RetGirth5} does not use triangle-freeness --- \cite[Theorem~2.3]{FGZRet} gives a trichotomy for approximately counting retractions to the class of irreflexive square-free graphs.
\end{itemize}
Therefore, we investigate square-free graphs that contain at least one triangle and at least one looped vertex. It turns out that we have to work through a number of technical cases to cover all $\sat$-hard graphs with these properties (and hence all $\sat$-hard square-free graphs).

For a positive integer $q$ the graph $\WR{q}$ is a looped star on $q+1$ vertices (the underlying star has $q$ degree-$1$ vertices). The (reflexive) \emph{net} is a looped triangle where each vertex of the triangle has an additional looped neighbour, as shown in the following illustration:
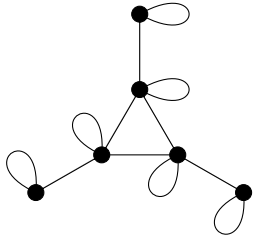
\begin{center}
	\begin{tikzpicture}[scale=1, baseline=0.36cm, every loop/.style={min distance=10mm,looseness=10}]
		
		\filldraw (0,0) node(d){} circle[radius=3pt] --++ (30:1cm) node(a){} circle[radius=3pt] --++ (60:1cm) node(b){} circle[radius=3pt] --++ (-60:1cm) node(c){} circle[radius=3pt]--++ (-30:1cm) node(f){} circle[radius=3pt];
		\filldraw (b.center) --++ (90:1cm) node(e){} circle[radius=3pt];
		\draw (a.center) -- (c.center);	
		
		
		\path[-] (a.center) edge  [in=155,out=85,loop] node {} ();
		\path[-] (d.center) edge  [in=155,out=85,loop] node {} ();
		\path[-] (b.center) edge  [in=30,out=-30,loop] node {} ();
		\path[-] (e.center) edge  [in=30,out=-30,loop] node {} ();
		\path[-] (c.center) edge  [in=-155,out=-85,loop] node {} ();
		\path[-] (f.center) edge  [in=-155,out=-85,loop] node {} ();

	\end{tikzpicture}
\end{center}
Here is an overview of the cases that we consider:
\begin{itemize}
	\item In Section~\ref{sec:mixedTriangles} we show that mixed triangles induce $\sat$-hardness.
	\item In Section~\ref{sec:NeighbourhoodOfLoopedVertex} we show in which cases the neighbourhood of a looped vertex induces $\sat$-hardness.
	\item In Sections~\ref{sec:inducedWR3},~\ref{sec:inducedNet} and~\ref{sec:inducedCycle} we show $\sat$-hardness for square-free graphs with an induced $\WR{3}$, with an induced net and with an induced reflexive cycle of length at least $5$, respectively. Essentially, these three sections deal with the graphs from the excluded subgraph characterisation of reflexive proper interval graphs (see~\cite[Section 1 and Appendix A]{GGJList} for the details about this characterisation).
\end{itemize}  

\subsection{Retractions and Neighbourhoods}
\begin{defn}\label{def:neighbourhoods}
	For a graph $H$ and a vertex $v\in V(H)$ we define the \emph{(distance-$1$) neighbourhood} of $v$ as $\NH(v)=\{u\in V(H) \mid \{u,v\} \in E(H)\}$. (In particular, this might include $v$ itself.) Then $\deg_H(v)=\abs{\NH(v)}$ is the \emph{degree} of $v$.
	More generally, the \emph{distance-$k$ neighbourhood} of $v$ is defined as $\NH^k(v)=\{u\in V(H) \mid \text{There is a walk } W=u,w_1,\dots, w_{k-1}, v\text{ (on $k$ edges) in }H\}$. Let $U$ be a subset of $V(H)$. Then $\NH(U)= \bigcap_{v\in U} \NH(v)$ is the \emph{set of common neighbours} of the vertices in $U$.
\end{defn}

The following well-known and simple observation shows that, for approximately counting retractions, hardness carries over from subgraphs that are induced by the neighbourhood of a vertex.
\begin{obs}
	\label{obs:PinNeighbourhood}
	Let $H$ be a graph and let $u$ be a vertex of $H$. Then $\Ret{H[\NH(u)]} \leap \Ret{H}$.
\end{obs}
\begin{proof}
	Let $(G,\boldS)$ be an input to $\Ret{H[\NH(u)]}$, let $v_1, \dots, v_n$ be the vertices of $G$ and $\boldS=\{S_v \mid v\in V(G)\}$. Let $w$ be a vertex distinct from the vertices in $G$. Then we construct the graph $G'$ with vertices $V(G')=V(G)\cup \{w\}$ and edges $E(G')=E(G) \cup \{\{w,v_i\} \mid i\in [n]\}$.
	We set $\boldS'=\{S'_v \mid v\in V(G')\}$, where
	\[
	S'_v=
	\begin{cases}
	\{u\}, &\text{if }v=w\\
	S_v, &\text{if }v\in V(G)\text{ and }\abs{S_v}=1\\
	V(H), & \text{otherwise.}
	\end{cases}
	\]
	Then $\hom{(G,\boldS)}{H[\NH(u)]} = \hom{(G',\boldS')}{H}$.
\end{proof}
\subsection{Square-Free Graphs with Mixed Triangles}\label{sec:mixedTriangles}
\begin{lem}\label{lem:hardtriangles1}
	Let $H$ be a square-free graph which contains a triangle with exactly two looped and one unlooped vertex. Then $\sat\leap \Ret{H}$.
\end{lem}
\begin{proof}
	Let  $b_1, b_2, r$ be a triangle in $H$, where $b_1$ and $b_2$ are looped and $r$ is unlooped. Consider the neighbourhood $\NH(b_1) \cap \NH(b_2)=\NH(\{b_1,b_2\})$. Since $H$ is square-free, $H[\NH(b_1) \cap \NH(b_2)]$ is precisely the triangle $b_1, b_2, r$. The problem $\Hom{H[\{b_1, b_2, r\}]}$ corresponds to counting independent sets where vertices not in the independent set have a weight of $2$ and vertices in the independent set have weight $1$. It is well-known that approximately counting weighted independent sets is $\sat$-hard, see for instance~\cite[Lemma 2]{GJPTwoSpin}. This gives $\sat \leap \Hom{H[\NH(b_1) \cap \NH(b_2)]}$. From Observation~\ref{obs:HomToRetToLHom} it follows immediately that $\sat\leap \Ret{H[\NH(b_1) \cap \NH(b_2)]}$.
	
	Finally, one can easily observe that $\Ret{H[\NH(b_1) \cap \NH(b_2)]}\leap \Ret{H}$: Let $(G,\boldS)$ be an input to $\Ret{H[\NH(b_1) \cap \NH(b_2)]}$ and let $\boldS=\{S_v \mid v\in V(G)\}$. Let $w_1$ and $w_2$ be vertices distinct from the vertices in $G$. Then we construct the graph $G'$ with vertices $V(G')=V(G)\cup \{w_1, w_2\}$ and edges $E(G')=E(G) \cup \bigl(\ucp{V(G)}{\{w_1, w_2\}}\bigr)$.
	We set $\boldS'=\{S'_v \mid v\in V(G')\}$, where
	\[
	S'_v=
	\begin{cases}
	\{b_1\}, &\text{if }v=w_1\\
	\{b_2\}, &\text{if }v=w_2\\
	S_v, &\text{if }v\in V(G)\text{ and }\abs{S_v}=1\\
	V(H), & \text{otherwise.}
	\end{cases}
	\]
	Then $\hom{(G,\boldS)}{H[\NH(b_1) \cap \NH(b_2)]} = \hom{(G',\boldS')}{H}$.
\end{proof}

\begin{lem}\label{lem:hardtriangles2}
	Let $H$ be a square-free graph which contains a triangle with exactly two unlooped and one looped vertex. Then $\sat\leap \Ret{H}$.
\end{lem}
\begin{proof}
	Let  $H_T$ be a triangle in $H$ with vertices $b$, $r_1$ and $r_2$, where $b$ is looped and both $r_1$ and $r_2$ are unlooped. Let $H'=H[\NH(b)]$. By Observation~\ref{obs:PinNeighbourhood} it holds that $\Ret{H'} \leap \Ret{H}$. Suppose we can show that $\Ret{H'[\{b,r_1\}]}\leap \Ret{H'}$. Then $H'[\{b,r_1\}]$ is a single edge with one looped ($b$) and one unlooped vertex ($r_1$) and it is well-known that counting homomorphisms to this graph corresponds to counting independent sets, which in turn is known to be $\sat$-hard (\cite[Theorem 3]{DGGJApprox}). Summarising we have
	\[
	\sat \leap \Hom{H'[\{b,r_1\}]} \leap \Ret{H'[\{b,r_1\}]}\leap \Ret{H'} \leap \Ret{H}, 
	\]
	where the second reduction is from Observation~\ref{obs:HomToRetToLHom}. This proves the lemma. It remains to prove the following claim.
	
	\bigskip
	\noindent{\bf\boldmath Claim: $\Ret{H'[\{b,r_1\}]}\leap \Ret{H'}$.}
	\smallskip
	
	\noindent {\it Proof of the claim:}\quad
	For $u\in V(H')$ let $w(u)$ be the number of common neighbours of $r_1$ and $u$ in $H'$. Then $w(b)=2$ since $r_1$ and $b$ have two common neighbours: $r_2$ and $b$, and these are their only common neighbours as $H'$ is square-free. Similarly, $w(r_1)=2$ as the ``common'' neighbours in this case are simply the neighbours of $r_1$, which are only $b$ and $r_2$ (since $H'$ is square-free). Now let $u\in V(H')\setminus\{b,r_1\}$. The vertex $b$ is a common neighbour of $u$ and $r_1$ since every vertex in $H'$ is a neighbour of $b$. It turns out that $b$ is the only common neighbour of $u$ and $r_1$: Suppose there exists a vertex $u'\neq b$ in $H'$ which is a common neighbour of $u$ and $r_1$. If $u'=u$ (see Figure~\ref{fig:mixedTriangle} on the left) then $u$ is adjacent to $r_1$. 
	\begin{figure}[ht]
		\centering
		\begin{minipage}{.45 \textwidth}
			\centering
			\begin{tikzpicture}[scale=1, baseline=0.36cm, every loop/.style={min distance=10mm,looseness=10}]
			
			\filldraw (0,0) node(r1){} circle[radius=3pt] --++ (60:1.5cm) node(b){} circle[radius=3pt] --++ (-60:1.5cm) node(r2){} circle[radius=3pt] -- (r1.center);
			
			\node[circle,fill=black,inner sep=0pt,minimum size=6.5pt] (u) at ($(r1)+(120:1.5cm)$){};
			\path[-] (b.center) edge (u.center);
			
			\node at ($(r1)+(0cm,-.35cm)$) {$r_1$};	
			\node at ($(b)+(0cm,.35cm)$) {$b$};	
			\node at ($(r2)+(0cm,-.35cm)$) {$r_2$};
			\node at ($(u)+(.41cm,.35cm)$) {$u=u'$};
			
			\path[-] (b.center) edge  [in=125,out=55,loop] node {} ();
			\path[-] (u.center) edge  [in=125,out=55,loop] node {} ();
			
			\path[dashed] (r1.center) edge (u.center);

			\end{tikzpicture}
		\end{minipage}
		\begin{minipage}{.45 \textwidth}
			\centering
			\begin{tikzpicture}[scale=1, baseline=0.36cm, every loop/.style={min distance=10mm,looseness=10}]
			
			\filldraw (0,0) node(r1){} circle[radius=3pt] --++ (60:1.5cm) node(b){} circle[radius=3pt] --++ (-60:1.5cm) node(r2){} circle[radius=3pt] -- (r1.center);
			
			\node[circle,fill=black,inner sep=0pt,minimum size=6.5pt] (u) at ($(r1)+(120:1.5cm)$){};
			\node[circle,fill=black,inner sep=0pt,minimum size=6.5pt] (uprime) at ($(r1)+(180:1.5cm)$){};
			\path[-] (b.center) edge (u.center);
			\path[-] (b.center) edge (uprime.center);
			
			\node at ($(r1)+(0cm,-.35cm)$) {$r_1$};	
			\node at ($(b)+(0cm,.35cm)$) {$b$};	
			\node at ($(r2)+(0cm,-.35cm)$) {$r_2$};
			\node at ($(u)+(0cm,.35cm)$) {$u$};	
			\node at ($(uprime)+(0cm,-.35cm)$) {$u'$};
			
			\path[-] (b.center) edge  [in=125,out=55,loop] node {} ();
			
			\path[dashed] (r1.center) edge (uprime.center);
			\path[dashed] (uprime.center) edge (u.center);

			\end{tikzpicture}
		\end{minipage}
		\caption{Contradictions to the square-freeness of the graph $H'$.}
		\label{fig:mixedTriangle}
	\end{figure}
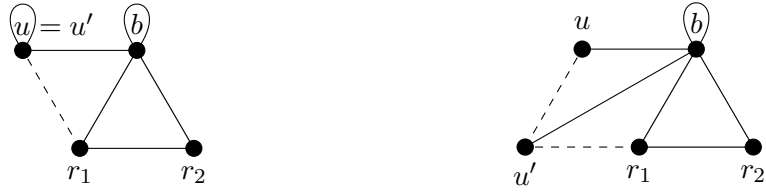
	Additionally, $u$ is then looped and hence $u\neq r_2$. Then $u, r_1, r_2, b$ is a square. If otherwise $u'\neq u$ then $u, u', r_1, b$ is a square (see Figure~\ref{fig:mixedTriangle} on the right), both cases give a contradiction. So we have shown that, for $u\in V(H')$,
	\begin{equation}\label{equ:}
	w(u)=2 \text{ if } u\in \{b,r_1\},\text{ and } w(u)=1\text{ otherwise.}
	\end{equation}
	Intuitively, we will use this fact to ``boost'' the vertices $b$ and $r_1$ and make them exponentially more likely to be used by a homomorphism to $H'$.

	Let $q$ be the number of vertices of $H'$. Now let $(G,\boldS)$ be an $n$-vertex input to $\Ret{H'[\{b,r_1\}]}$ and let $\eps$ be the desired precision. As usual, from $(G,\boldS)$ we define an input $(G',\boldS')$ to $\Ret{H'}$. We introduce a vertex $p$ distinct from the vertices of $G$ that will serve as a pin to the vertex $r_1$ in $H'$. Then, for each $v\in V(G)$ we introduce an independent set on $s$ vertices all of which are connected only to $p$ and $v$. The parameter $s$ will depend on the input size, specifically we set $s=n^2$. Intuitively it is clear that this gadget introduces a weight equal to $w(u)^s$ for each vertex $u\in V(H')$. For sufficiently large $s$, the image of $v$ is likely to be $b$ or $r_1$. This implies the statement of the lemma. The reader that trusts this intuition can skip reading the following calculations.
	
	We give the full details for the sake of completeness: For each $v\in V(G)$, let $I_v$ be an independent set of size $s$ with vertices distinct from the remaining vertices of $G'$. Then $G'$ is the graph with vertices $V(G')=V(G)\cup \{p\} \cup \bigcup_{v\in V(G)} I_v$
	and edges $E(G')=E(G) \cup \bigcup_{v\in V(G)} \bigl(\ucp{\{v,p\}}{I_v}\bigr)$.
	We set $\boldS'=\{S'_v \mid v\in V(G')\}$, where
	\[
	S'_v=
	\begin{cases}
	\{r_1\}, &\text{if }v=p\\
	S_v, &\text{if }v\in V(G)\text{ and }\abs{S_v}=1\\
	V(H'), & \text{otherwise.}
	\end{cases}
	\]
	
	We say that a homomorphism $h\in \calH((G', \boldS'), H')$ is \emph{full} if $h(V(G)) \subseteq \{b,r_1\}$.
	Let $Z^*$ be the number of full homomorphisms from $(G', \boldS')$ to $H'$. Let $Z_0$ be the number of non-full homomorphisms from $(G', \boldS')$ to $H'$. Then
	\begin{equation}\label{equ:PinNeighbourhood2-1}
	\hom{(G',\boldS')}{H'}= Z^* + Z_0.
	\end{equation}
	
	For $h \in \calH((G, \boldS), H')$, let $Z(h)$ be the number of homomorphisms $h'\in \calH((G', \boldS'), H')$ for which $h=h'\vert_{V(G)}$. By the construction of $G'$, every vertex $v\in V(G)$ with $h(v)\in \{b,r_1\}$ contributes a factor of $2^s$ to $Z(h)$, whereas a vertex $v\in V(G)$ with $h(v)\notin \{b,r_1\}$ contributes a factor of $1$ to $Z(h)$. It follows that 
	\begin{equation}\label{equ:PinNeighbourhood2-2}
	Z^*= \sum_{h\in \calH((G, \boldS), H'),\ h \text{ full}} Z(h) = 2^{sn} \cdot \hom{(G,\boldS)}{H'[\{b,r_1\}]},
	\end{equation}
	and
	\begin{equation*}
	Z_0 = \sum_{h\in \calH((G, \boldS), H'),\ h \text{ non-full}} Z(h) \le 2^{s(n-1)} \cdot \hom{(G,\boldS)}{H'}
	\le 2^{s(n-1)} \cdot q^n.
	\end{equation*}
	Therefore,
	\begin{equation}\label{equ:PinNeighbourhood2-3}
	Z_0/2^{sn} \le 2^{-s}\cdot q^n\le 1/4,
	\end{equation}
	where the last inequality holds for sufficiently large $n$ by the choice $s=n^2$.
	Summarising, by~\eqref{equ:PinNeighbourhood2-1} and~\eqref{equ:PinNeighbourhood2-2}, we have 
	\[
	\hom{(G,\boldS)}{H'[\{b,r_1\}]}=\frac{Z^*}{2^{sn}}\le\frac{\hom{(G',\boldS')}{H'}}{2^{sn}}
	\]
	and, using~\eqref{equ:PinNeighbourhood2-1},~\eqref{equ:PinNeighbourhood2-2} as well as~\eqref{equ:PinNeighbourhood2-3}, we obtain 
	\[
	\frac{\hom{(G',\boldS')}{H'}}{2^{sn}} = \frac{Z^*}{2^{sn}} + \frac{Z_0}{2^{sn}} \le \hom{(G,\boldS)}{H'[\{b,r_1\}]} + 1/4.
	\]
	Hence $\hom{(G',\boldS')}{H'}/2^{sn}\in \bigl[\hom{(G,\boldS)}{H'[\{b,r_1\}]},\hom{(G,\boldS)}{H'[\{b,r_1\}]} + 1/4\bigr]$. Let $Q$ be the solution returned by an oracle call to $\Ret{H'}$ with input $((G',\boldS'),\eps/21)$, i.e.~an approximation of $\hom{(G',\boldS')}{H'}$. Then the output $\floor{Q/2^{sn}}$ approximates $\hom{(G,\boldS)}{H'[\{b,r_1\}]}$ with the desired precision as was shown in~\cite[Proof of Theorem 3]{DGGJApprox}.
	{\it (End of the proof of the claim.)}
\end{proof}

\subsection{Square-Free Neighbourhoods of a Looped Vertex}\label{sec:NeighbourhoodOfLoopedVertex}

Now we consider graphs of the form $\X{k_1}{k_2}{k_3}$ (see Figure~\ref{fig:NeighbourhoodOfLoopedVertex}). Why are we interested in these graphs? Let $H$ be a square-free graph with a looped vertex $b$ and let $H$ not contain any mixed triangle as a subgraph. Then consider $H[\NH(b)]$, the graph induced by the neighbourhood of $b$. Since $H$ does not contain mixed triangles, the unlooped neighbours of $b$ do not have any neighbours in $H[\NH(b)]$ apart from $b$. Since $H$ is square-free, $H[\NH(b)]$ is square-free as well and therefore a looped neighbour $u\neq b$ of $b$ can have at most one additional neighbour apart from $b$ and $u$ itself (within $H[\NH(b)]$). It follows that $H[\NH(b)]$ is of the form  $\X{k_1}{k_2}{k_3}$. Note that $\X{k_1}{k_2}{0}$ does not contain any cycles. Therefore the hardness results for graphs of this form come from the classification for graphs of girth at least $5$ (Theorem~\ref{thm:RetGirth5}). The remaining cases ($k_3\ge 1$) are covered in this work. As an overview in advance, we will obtain $\sat \leap \Ret{\X{k_1}{k_2}{k_3}}$ in the following cases (where in most cases we actually show the stronger result $\sat \leap \Hom{\X{k_1}{k_2}{k_3}}$ --- $\sat$-hardness for $\Ret{\X{k_1}{k_2}{k_3}}$ then follows from Observation~\ref{obs:HomToRetToLHom}.):
\begin{myitemize}
	\item $k_3=0$ and
	\begin{myitemize}
		\item $k_2=0$ and $k_1\ge 1$ (Theorem~\ref{thm:RetGirth5})
		\item $k_2=1$ and $k_1\ge 1$ (Theorem~\ref{thm:RetGirth5})
		\item $k_2=2$ and $k_1\ge 2$ (Theorem~\ref{thm:RetGirth5})
	\end{myitemize}
	\item $k_3=1$ and
	\begin{myitemize}
		\item $k_2=0$ and $k_1\ge 1$ (Lemma~\ref{lem:hardNeighbourhood2})
		\item $k_2=1$ and $k_1\ge 3$ (Lemma~\ref{lem:hardNeighbourhood3})
	\end{myitemize}
	\item $k_3=2$, $k_2=0$ and $k_1\ge 5$ (Lemma~\ref{lem:hardNeighbourhood4})
	\item $k_2 + k_3 \ge 3$ (Lemma~\ref{lem:inducedWR3General})
\end{myitemize}
Following the classification for graphs of the form $\X{k_1}{k_2}{k_3}$ we give a hardness result (Lemma~\ref{lem:degree2bristle}) which uses properties of the distance-$2$ neighbourhood of a looped vertex $b$ in $H$.

\begin{figure}[h!]
	\centering
	\begin{tikzpicture}[scale=1, baseline=0.36cm, every loop/.style={min distance=10mm,looseness=10}]
	
	\filldraw (0,0) node(b){} circle[radius=3pt] -- (-1,1.5) node(u2){} circle[radius=3pt];
	\filldraw (b.center) -- (-3,1.5) node(u1){} circle[radius=3pt];
	\filldraw (b.center) -- (1,1.5) node(l1){} circle[radius=3pt];
	\filldraw (b.center) -- (3,1.5) node(l2){} circle[radius=3pt];
	\filldraw (b.center) -- (-1.5,-1.5) node(t2){} circle[radius=3pt] -- (-2.5,-1.5) node(t1){} circle[radius=3pt] -- (b.center);
	\filldraw (b.center) -- (1.5,-1.5) node(t3){} circle[radius=3pt] -- (2.5,-1.5) node(t4){} circle[radius=3pt] --  (b.center);
	
	
	\path[-] (b.center) edge  [in=-240,out=-300,loop] node {} ();
	\path[-] (l1.center) edge  [in=-240,out=-300,loop] node {} ();
	\path[-] (l2.center) edge  [in=-240,out=-300,loop] node {} ();
	\path[-] (t1.center) edge  [in=240,out=300,loop] node {} ();
	\path[-] (t2.center) edge  [in=240,out=300,loop] node {} ();
	\path[-] (t3.center) edge  [in=240,out=300,loop] node {} ();
	\path[-] (t4.center) edge  [in=240,out=300,loop] node {} ();
	
	\node at ($(u1)+(0cm,1cm)$) {$1$};
	\node at ($(u2)+(0cm,1cm)$) {$k_1$};
	\node at ($(l1)+(0cm,1cm)$) {$1$};
	\node at ($(l2)+(0cm,1cm)$) {$k_2$};
	\node at ($(b)+(0cm,-.5cm)$) {$b$};
	\node at ($(t1)+(0.5cm,-1cm)$) {$1$};
	\node at ($(t3)+(0.5cm,-1cm)$) {$k_3$};
	
	\coordinate (pmid) at (-2,1.5);
	\coordinate (cdot1) at ($(pmid)+(-.3cm,0cm)$);
	\coordinate (cdot2) at ($(pmid)+(0cm,0cm)$);
	\coordinate (cdot3) at ($(pmid)+(.3cm,0cm)$);
	\fill (cdot1) circle[radius=1.5pt];
	\fill (cdot2) circle[radius=1.5pt];
	\fill (cdot3) circle[radius=1.5pt];
	
	\coordinate (pmid2) at (2,1.5);
	\coordinate (cdot1) at ($(pmid2)+(-.3cm,0cm)$);
	\coordinate (cdot2) at ($(pmid2)+(0cm,0cm)$);
	\coordinate (cdot3) at ($(pmid2)+(.3cm,0cm)$);
	\fill (cdot1) circle[radius=1.5pt];
	\fill (cdot2) circle[radius=1.5pt];
	\fill (cdot3) circle[radius=1.5pt];
	
	\coordinate (pmid3) at (-90:1.6cm);
	\coordinate (cdot1) at ($(pmid3)+(-.3cm,0cm)$);
	\coordinate (cdot2) at ($(pmid3)+(0cm,0cm)$);
	\coordinate (cdot3) at ($(pmid3)+(.3cm,0cm)$);
	\fill (cdot1) circle[radius=1.5pt];
	\fill (cdot2) circle[radius=1.5pt];
	\fill (cdot3) circle[radius=1.5pt];
	
	\end{tikzpicture}
	\caption{The graph $\X{k_1}{k_2}{k_3}$.}
	\label{fig:NeighbourhoodOfLoopedVertex}
\end{figure}
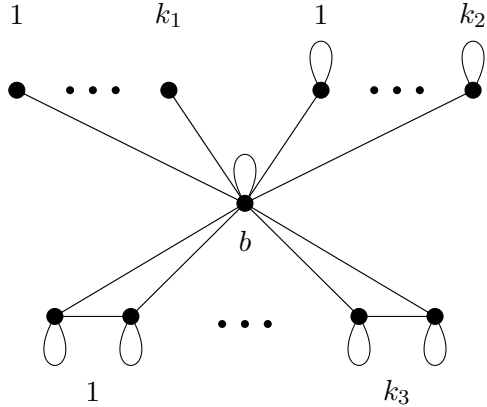

A useful and well-known tool for proving hardness results for approximate counting problems are gadgets based on complete bipartite graphs where two states dominate (see, e.g.,~\cite[Lemma 25]{DGGJApprox},~\cite[Section 5]{GKP2004},~\cite[Lemma 5.1]{KelkThesis} and~\cite[Lemma 2.30]{FGZRet}). Let $F(H)=\{u \in V(H) \mid \NH(u) = V(H)\}$. One can use the described tool to show that, under certain conditions, a homomorphism from a complete bipartite graph to $H$ will typically map one side to $F(H)$ and the other to $V(H)$. In this case it is then easy to reduce from counting independent sets to obtain $\sat$-hardness. Formally, we use the version stated by Kelk~\cite{KelkThesis}: 

\begin{lem}[{\cite[Lemma 5.1]{KelkThesis}}]\label{lem:Kelk5.1}
	Let $H$ be a graph with $\emptyset \subsetneq F(H) \subsetneq V(H)$. Suppose that, for every pair $(S,T)$ with $\emptyset \subseteq S,T \subseteq V(H)$ satisfying $S\subseteq \NH(T)$ and $T\subseteq \NH(S)$, at least one of the following holds:
	\begin{enumerate}[(1)]
		\item $S=F(H)$.\label{eq:Kelk1}
		\item $T=F(H)$.\label{eq:Kelk2}
		\item $\abs{S}\cdot \abs{T} < \abs{F(H)}\cdot \abs{V(H)}$.\label{eq:Kelk3}
	\end{enumerate}
	Then $\sat \leap \Hom{H}$.
\end{lem}

In order to prove Lemmas~\ref{lem:hardNeighbourhood2} and~\ref{lem:hardNeighbourhood3} we will use Lemma~\ref{lem:Kelk5.1} and, in addition, a reduction from the problem of counting large cuts, which is formally defined as follows: A \emph{cut} of a graph $G$ is a partition of $V(G)$ into two subsets (the order of this pair is ignored) and the size of a cut is the number of edges that have exactly one endpoint in each of these two subsets. 

\prob
{
	$\largecut$.
}
{
	An integer $K\ge 1$ and a connected graph $G$ in which every cut has size at most $K$.
}
{
	The number of size-$K$ cuts in $G$.
}

The full details of the proof involve analysing different types of homomorphisms. The most important part of the results leading up to the proof of Lemmas~\ref{lem:hardNeighbourhood2} and~\ref{lem:hardNeighbourhood3} are the Tables~\ref{tab:configs1} and~\ref{tab:configs2}, respectively. These tables show which types of homomorphisms represent a significant share of the overall number of homomorphisms that we are interested in. The crucial question is whether we can ensure that the right types of homomorphisms dominate this number. We desire two properties. First, the number of homomorphisms should be dominated by homomorphisms of two distinct types. Second, these two types should interact in an ``anti-ferromagnetic'' way.

The following definitions and preliminary technical results resemble the ones introduced in~\cite[Section 2.2.2]{FGZRet}.
We are going to use the graph $\Jpqt$ (see Figure~\ref{fig:Jpqt}) as a vertex gadget. This gadget was originally introduced in~\cite{DGGJApprox}. In general it is a good candidate when looking for gadgets to prove reductions from $\largecut$. Here is the formal definition.

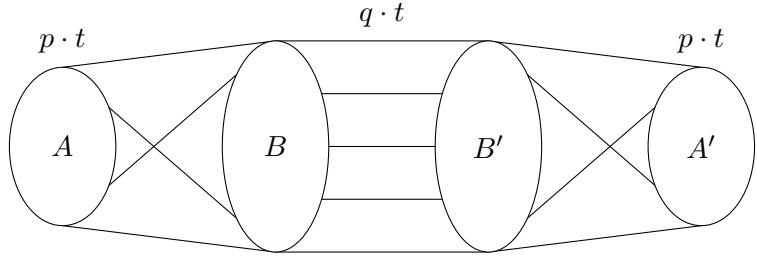
\begin{figure}[h!]\centering
	{\def\scaleFactor{.7}
		\begin{tikzpicture}[scale=\scaleFactor, every loop/.style={min distance=10mm,looseness=10}]
		
		\foreach \i in {-2,...,2}
		\draw (0,\i)--(4,\i);
		
		\node at ($(2,2.5)$) {$q\cdot t$};
		\node at ($(-4,2)$) {$p\cdot t$};
		\node at ($(8,2)$) {$p\cdot t$};
		
		\draw (-4,1.5) -- (0,2);
		\draw (-4,1.5) -- (0,-2);
		\draw (-4,-1.5) -- (0,-2);
		\draw (-4,-1.5) -- (0,2);
		
		\draw (8,1.5) -- (4,2);
		\draw (8,1.5) -- (4,-2);
		\draw (8,-1.5) -- (4,-2);
		\draw (8,-1.5) -- (4,2);

		\draw (-4,0) ellipse (1cm and 1.5cm) [fill=white];
		\draw (0,0) ellipse (1cm and 2cm) [fill=white];
		\draw (4,0) ellipse (1cm and 2cm) [fill=white];
		\draw (8,0) ellipse (1cm and 1.5cm) [fill=white];

		\node at (-4,0) {$A$};
		\node at (0,0) {$B$};
		\node at (4,0) {$B'$};
		\node at (8,0) {$A'$};
		
		\end{tikzpicture}
	}
	\caption{The graph $\Jpqt$.}
	\label{fig:Jpqt}
\end{figure}

Let $p$, $q$ and $t$ be positive integers. Let $A$ and $A'$ be independent sets of size $p\cdot t$ and let $B$ and $B'$ be independent sets of size $q\cdot t$. The set of edges $M$ between $B$ and $B'$ forms a perfect matching. Then $\Jpqt$ is the graph for which the vertex set is the union of $A$, $B$, $B'$ and $A'$. The edges are $(\ucp{A}{B})\cup M \cup (\ucp{B'}{A'})$.

Let $H$ be a graph and let $h$ be a homomorphism from $\Jpqt$ to $H$. Let $h(B,B')=\{(h(u),h(v)) \mid u\in B, v\in B', \{u,v\} \in E(H)\}$. We say that $h$ has \emph{type} $\left(h(A), h(B,B'), h(A')\right)$. In general, a tuple $T=(T_1, T_2, T_3)$ is an \emph{$H$-type} if $T_1,T_3\subseteq V(H)$ and $T_2 \subseteq \{(x,y) \mid \{x,y\}\in E(H)\}$. Let $A(T)=T_1$, $B(T)=\{x \mid \exists y\,(x,y)\in T_2\}$, $B'(T)=\{y \mid \exists x\,(x,y)\in T_2\}$ and $A'(T)=T_3$. 

An  $H$-type $T$ is \emph{non-empty} (with respect to $\Jpqt$) if there exists a homomorphism from $\Jpqt$ to $H$ that has type $T$. Otherwise, $T$ is called an \emph{empty $H$-type}. From the definition of $\Jpqt$ we observe the following.

\begin{obs}\label{obs:nonemptyType}
	Let $H$ be a graph.
	An $H$-type $T$ is non-empty if and only if
	\begin{enumerate}[(1)]
		\item $T_1$, $T_2$ and $T_3$ are non-empty,
		\item $\ucp{T_1}{B(T)} \subseteq E(H)$,
		\item $\ucp{B'(T)}{T_3} \subseteq E(H)$.
	\end{enumerate}
\end{obs}

Let $T$ and $T'$ be $H$-types. We write $T\subseteq T'$ if, for $i\in [3]$, we have $T_i\subseteq T_i'$.
An $H$-type $T$ is \emph{maximal} if it is non-empty and every $H$-type $T'$ with $T'\neq T$, $T\subseteq T'$ is empty. The following analysis (Lemma~\ref{lem:maximalType}) is contained in the proof of~\cite[Lemma 2.25]{FGZRet}.

\begin{lem}\label{lem:maximalType}
	Let $H$ be a graph and let $T=(T_1, T_2, T_3)$ 
	be a maximal $H$-type. Then $T$ is completely defined by $B(T)$ and $B'(T)$ since 
	\begin{enumerate}[(1)]
		\item\label{item:maxType1} $T_1=\NH(B(T))$, $T_2=E(B(T), B'(T))$ and $T_3=\NH(B'(T))$.
		\item\label{item:maxType2} $B(T)=\NH(\NH(B(T)))$ and $B'(T)=\NH(\NH(B'(T)))$.
	\end{enumerate}
\end{lem}

Given an $H$-type $T=(T_1, T_2, T_3)$ we define $N(T)$ to be the number of homomorphisms in $\calH(\Jpqt,H)$ that have type $T$. We also set $\displaystyle\hatN(T)=\abs{T_1}^{pt} \abs{T_2}^{qt} \abs{T_3}^{pt}$. For non-empty~$T$,
$\hatN(T)$  is a close approximation to $N(T)$:

\begin{lem}[{\cite[Lemma 2.20]{FGZRet}}]\label{lem:hatN}
	Let $H$ be a graph.
	Let $p$ and $q$ be positive integers. There exists a positive integer $t_0$ such that for all $t\ge t_0$ and all $H$-types $T$ that are non-empty with respect to $\Jpqt$, it holds that
	\[
	\frac{\hatN(T)}{2} \le N(T)\le \hatN(T).
	\]
\end{lem}

\begin{lem}[{\cite[Lemma 2.23]{FGZRet}}]\label{lem:maximalTypeDominates}
	Let $H$ be a connected graph with at least $2$ vertices. Let $T$ be a non-empty $H$-type that is not maximal. Then there exists a non-empty $H$-type $T^*$ such that $\hatN(T) \le \left(\frac{2\abs{E(H)}-1}{2\abs{E(H)}}\right)^t \hatN(T^*)$.
\end{lem}

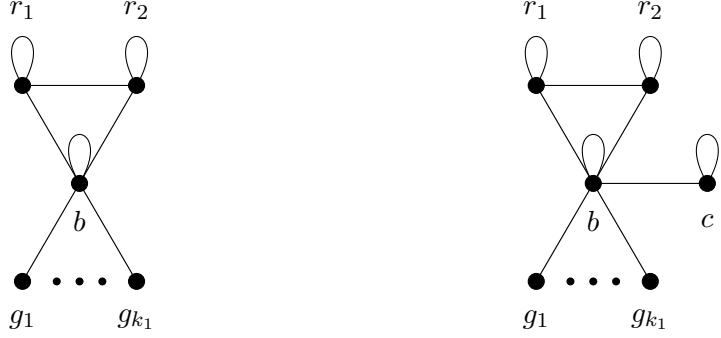
\begin{figure}[ht]
	\centering
	\begin{minipage}{.45 \textwidth}
		\centering
		\begin{tikzpicture}[scale=1, baseline=0.36cm, every loop/.style={min distance=10mm,looseness=10}]
		
		\filldraw (0,0) node(b){} circle[radius=3pt] -- (120:1.5cm) node(r1){} circle[radius=3pt] -- (60:1.5cm) node(r2){} circle[radius=3pt] -- (b.center);
		\filldraw (0,0) node(b){}  -- (-120:1.5cm) node(g1){} circle[radius=3pt];
		\filldraw (0,0) node(b){}  -- (-60:1.5cm) node(gk){} circle[radius=3pt];
		
		
		\path[-] (b.center) edge  [in=-240,out=-300,loop] node {} ();
		\path[-] (r1.center) edge  [in=-240,out=-300,loop] node {} ();
		\path[-] (r2.center) edge  [in=-240,out=-300,loop] node {} ();

		\node at ($(r1)+(0cm,1cm)$) {$r_1$};
		\node at ($(r2)+(0cm,1cm)$) {$r_2$};
		\node at ($(b)+(0cm,-.5cm)$) {$b$};
		\node at ($(g1)+(0cm,-.5cm)$) {$g_1$};
		\node at ($(gk)+(0cm,-.5cm)$) {$\displaystyle g_{k_1}$};
		
		\coordinate (pmid3) at (-90:1.3cm);
		\coordinate (cdot1) at ($(pmid3)+(-.3cm,0cm)$);
		\coordinate (cdot2) at ($(pmid3)+(0cm,0cm)$);
		\coordinate (cdot3) at ($(pmid3)+(.3cm,0cm)$);
		\fill (cdot1) circle[radius=1.5pt];
		\fill (cdot2) circle[radius=1.5pt];
		\fill (cdot3) circle[radius=1.5pt];
		
		\end{tikzpicture}
	\end{minipage}
	\begin{minipage}{.45 \textwidth}
		\centering
		\begin{tikzpicture}[scale=1, baseline=0.36cm, every loop/.style={min distance=10mm,looseness=10}]
		
		\filldraw (0,0) node(b){} circle[radius=3pt] -- (120:1.5cm) node(r1){} circle[radius=3pt] -- (60:1.5cm) node(r2){} circle[radius=3pt] -- (b.center);
		\filldraw (0,0) node(b){}  -- (-120:1.5cm) node(g1){} circle[radius=3pt];
		\filldraw (0,0) node(b){}  -- (-60:1.5cm) node(gk){} circle[radius=3pt];
		\filldraw (0,0) node(b){}  -- (0:1.5cm) node(c){} circle[radius=3pt];
		
		
		\path[-] (b.center) edge  [in=-240,out=-300,loop] node {} ();
		\path[-] (r1.center) edge  [in=-240,out=-300,loop] node {} ();
		\path[-] (r2.center) edge  [in=-240,out=-300,loop] node {} ();
		\path[-] (c.center) edge  [in=-240,out=-300,loop] node {} ();

		\node at ($(r1)+(0cm,1cm)$) {$r_1$};
		\node at ($(r2)+(0cm,1cm)$) {$r_2$};
		\node at ($(b)+(0cm,-.5cm)$) {$b$};
		\node at ($(g1)+(0cm,-.5cm)$) {$g_1$};
		\node at ($(gk)+(0cm,-.5cm)$) {$\displaystyle g_{k_1}$};
		\node at ($(c)+(0cm,-.5cm)$) {$c$};
		
		\coordinate (pmid3) at (-90:1.3cm);
		\coordinate (cdot1) at ($(pmid3)+(-.3cm,0cm)$);
		\coordinate (cdot2) at ($(pmid3)+(0cm,0cm)$);
		\coordinate (cdot3) at ($(pmid3)+(.3cm,0cm)$);
		\fill (cdot1) circle[radius=1.5pt];
		\fill (cdot2) circle[radius=1.5pt];
		\fill (cdot3) circle[radius=1.5pt];
		
		\end{tikzpicture}
	\end{minipage}
	\caption{The graphs $\X{k_1}{0}{1}$ (on the left) and $\X{k_1}{1}{1}$ (on the right).}
	\label{fig:2specialNeighbourhoods}
\end{figure}

\begin{table}[h!]
	\centering
	\caption{Maximal types of the homomorphisms from $\Jpqt$ to $\X{k_1}{0}{1}$, where the vertices of $\X{k_1}{0}{1}$ are labelled as in Figure~\ref{fig:2specialNeighbourhoods} (on the left). Each line $i$ corresponds to a type 
		$T_i=\Bigl(A(T_i), E\bigl(B(T_i),B'(T_i)\bigr), A'(T_i)\Bigr)$. To shorten the notation we set $\calG=\{g_j \mid j\in [k_1]\}$.}
	\label{tab:configs1}
	\setlength{\tabcolsep}{3pt}
	\newcommand\Tstrut{\rule{0pt}{2.6ex}} 
	\setcounter{rowno}{0}
	\begin{tabular}{@{}>{$T_{\stepcounter{rowno}\therowno}$}lcccc<{\ }|>{\ }r@{}}
		\multicolumn{1}{l}{} & $A(T)$ & $B(T)$ & $B'(T)$ & $A'(T)$ & $\hatN(T)$\\
		\hline
		\Tstrut
		
		& $\{r_1, r_2,b\}\cup \calG$ 
		& $\{b\}$ 
		& $\{b\}$ 
		& $\{r_1, r_2,b\}\cup \calG$  
		& $(3+k_1)^{pt}\cdot 1^{qt}\cdot (3+k_1)^{pt}$\\
		& $\{r_1, r_2,b\}\cup \calG$  
		& $\{b\}$ 
		& $\{r_1, r_2,b\}$ 
		& $\{r_1, r_2,b\}$ 
		& $(3+k_1)^{pt}\cdot 3^{qt}\cdot 3^{pt}$\\
		& $\{r_1, r_2,b\}\cup \calG$ 
		& $\{b\}$
		& $\{r_1, r_2,b\}\cup \calG$ 
		& $\{b\}$  
		& $(3+k_1)^{pt}\cdot (3+k_1)^{qt}\cdot 1^{pt}$\\
		& $\{r_1, r_2,b\}$ 
		& $\{r_1, r_2,b\}$ 
		& $\{r_1, r_2,b\}$  
		& $\{r_1, r_2,b\}$   
		& $3^{pt}\cdot 9^{qt}\cdot 3^{pt}$\\
		& $\{r_1, r_2,b\}$  
		& $\{r_1, r_2,b\}$  
		& $\{r_1, r_2,b\}\cup \calG$ 
		& $\{b\}$
		& $3^{pt}\cdot (9+k_1)^{qt}\cdot 1^{pt}$\\
		& $\{b\}$
		& $\{r_1, r_2,b\}\cup \calG$ 
		& $\{r_1, r_2,b\}\cup \calG$ 
		& $\{b\}$  
		& $1^{pt}\cdot (9+2k_1)^{qt}\cdot 1^{pt}$
	\end{tabular}
\end{table}

Let $T=\bigl(A(T), E(B(T),B'(T)), A'(T)\bigr)$ be an $H$-type. Then we call $T$ \emph{symmetric} to the $H$-type $T'=\bigl(A'(T), E(B'(T),B(T)), A(T)\bigr)$. Clearly, $\hatN(T)=\hatN(T')$.

\begin{lem}\label{lem:tableOfMaxTypes1}
	Let $H= \X{k_1}{0}{1}$. Then all maximal $H$-types are listed in Table~\ref{tab:configs1} (apart from  those that are symmetric to a listed $H$-type). For each listed $H$-type $T$ the last column of the table gives the corresponding value $\hatN(T)$.
\end{lem}
\begin{proof}
	Let $H= \X{k_1}{0}{1}$ and let $T$ be a maximal $H$-type. We claim that 
	\[
	B(T), B'(T) \in \{\{b\}, \{r_1, r_2, b\}, \{r_1, r_2, b, g_1, \dots, g_{k_1}\}\}
	\]
	for the following reasons (we give the arguments for $B(T)$, they are identical for $B'(T)$):
	\begin{itemize}
		\item Since $b$ is a neighbour of every vertex in $H$, from item~\eqref{item:maxType2} of Lemma~\ref{lem:maximalType} we obtain $b\in B(T)$. 
		\item If, for some $i\in [k_1]$, we have $g_i\in B(T)$ then $\NH(B(T))=\{b\}$ as $b$ is the only neighbour of $g_i$. By item~\eqref{item:maxType2} of Lemma~\ref{lem:maximalType} it follows that $B(T)=\{r_1, r_2, b, g_1, \dots, g_{k_1}\}$.
		\item If $B(T)=\{r_1,b\}$, then $\NH(B(T))=\{r_1, r_2, b\}$. By item~\eqref{item:maxType2} of Lemma~\ref{lem:maximalType} this gives $B(T)=\{r_1, r_2,b\}$, a contradiction.
		\item $B(T)=\{r_2,b\}$ gives a contradiction in the same way.
	\end{itemize}
	
	Table~\ref{tab:configs1} then lists all possible combinations of sets $B(T)$ and $B'(T)$. From Lemma~\ref{lem:maximalType} it follows that these sets determine $T$ completely (and $A(T), A'(T)$ are given accordingly). By definition, $T_1=A(T)$ and $T_3=A'(T)$. From item~\eqref{item:maxType1} of Lemma~\ref{lem:maximalType} it has to hold that $T_2=E(B(T),B'(T))$. Then $\displaystyle\hatN(T)=\abs{T_1}^{pt} \abs{T_2}^{qt} \abs{T_3}^{pt}$ can be computed from the given sets in each row.
\end{proof}

\begin{table}[h!]
	\centering
	\caption{Maximal types of the homomorphisms from $\Jpqt$ to $\X{k_1}{1}{1}$, where the vertices of $\X{k_1}{1}{1}$ are labelled as in Figure~\ref{fig:2specialNeighbourhoods} (on the right). Each line $i$ corresponds to a type 
		$T_i=\Bigl(A(T_i), E\bigl(B(T_i),B'(T_i)\bigr), A'(T_i)\Bigr)$. To shorten the notation we set $\calG=\{g_j \mid j\in [k_1]\}$.}
	\label{tab:configs2}
	\setlength{\tabcolsep}{3pt}
	\newcommand\Tstrut{\rule{0pt}{2.6ex}} 
	\setcounter{rowno}{0}
	\begin{tabular}{@{}>{$T_{\stepcounter{rowno}\therowno}$}lcccc<{\ }|>{\ }r@{}}
		\multicolumn{1}{l}{} & $A(T)$ & $B(T)$ & $B'(T)$ & $A'(T)$ & $\hatN(T)$\\
		\hline
		\Tstrut
		
		& $\{r_1, r_2,b,c\}\cup \calG$ 
		& $\{b\}$ 
		& $\{b\}$ 
		& $\{r_1, r_2,b,c\}\cup \calG$  
		& $(4+k_1)^{pt}\cdot 1^{qt}\cdot (4+k_1)^{pt}$\\
		& $\{r_1, r_2,b,c\}\cup \calG$  
		& $\{b\}$ 
		& $\{b,c\}$ 
		& $\{b,c\}$ 
		& $(4+k_1)^{pt}\cdot 2^{qt}\cdot 2^{pt}$\\
		& $\{r_1, r_2,b,c\}\cup \calG$  
		& $\{b\}$ 
		& $\{r_1, r_2,b\}$ 
		& $\{r_1, r_2,b\}$ 
		& $(4+k_1)^{pt}\cdot 3^{qt}\cdot 3^{pt}$\\
		& $\{r_1, r_2,b,c\}\cup \calG$ 
		& $\{b\}$
		& $\{r_1, r_2,b,c\}\cup \calG$ 
		& $\{b\}$  
		& $(4+k_1)^{pt}\cdot (4+k_1)^{qt}\cdot 1^{pt}$\\
		& $\{b,c\}$ 
		& $\{b,c\}$  
		& $\{b,c\}$ 
		& $\{b,c\}$   
		& $2^{pt}\cdot 4^{qt}\cdot 2^{pt}$\\
		& $\{b,c\}$  
		& $\{b,c\}$  
		& $\{r_1, r_2,b\}$ 
		& $\{r_1, r_2,b\}$ 
		& $2^{pt}\cdot 4^{qt}\cdot 3^{pt}$\\
		& $\{b,c\}$  
		& $\{b,c\}$  
		& $\{r_1, r_2,b,c\}\cup \calG$ 
		& $\{b\}$
		& $2^{pt}\cdot (6+k_1)^{qt}\cdot 1^{pt}$\\
		& $\{r_1, r_2,b\}$ 
		& $\{r_1, r_2,b\}$ 
		& $\{r_1, r_2,b\}$  
		& $\{r_1, r_2,b\}$   
		& $3^{pt}\cdot 9^{qt}\cdot 3^{pt}$\\
		& $\{r_1, r_2,b\}$  
		& $\{r_1, r_2,b\}$  
		& $\{r_1, r_2,b,c\}\cup \calG$ 
		& $\{b\}$
		& $3^{pt}\cdot (10+k_1)^{qt}\cdot 1^{pt}$\\
		& $\{b\}$
		& $\{r_1, r_2,b,c\}\cup \calG$ 
		& $\{r_1, r_2,b,c\}\cup \calG$ 
		& $\{b\}$  
		& $1^{pt}\cdot (12+2k_1)^{qt}\cdot 1^{pt}$
	\end{tabular}
\end{table}

\begin{lem}\label{lem:tableOfMaxTypes2}
	Let $H= \X{k_1}{1}{1}$. Then all maximal $H$-types are listed in Table~\ref{tab:configs2} (apart from  those that are symmetric to a listed $H$-type). For each listed $H$-type $T$ the last column of the table gives the corresponding value $\hatN(T)$.
\end{lem}
\begin{proof}
	Let $H= \X{k_1}{1}{1}$ and let $T$ be a maximal $H$-type. We claim that 
	\[
	B(T), B'(T) \in \{\{b\}, \{b,c\}, \{r_1, r_2, b\}, \{r_1, r_2, b, c, g_1, \dots, g_{k_1}\}\}.
	\]
	The remainder of the proof is analogous to that of Lemma~\ref{lem:tableOfMaxTypes1} with only one additional argument:
	\begin{itemize}
		\item If $r_1\in B(T)$ and $c\in B(T)$, then $\NH(B(T))=\{b\}$. By item~\eqref{item:maxType2} of Lemma~\ref{lem:maximalType} it follows that $B(T)=\{r_1, r_2, b, c, g_1, \dots, g_{k_1}\}$. The same is true if both $r_2\in B(T)$ and $c\in B(T)$.
	\end{itemize}
\end{proof}

\begin{lem}\label{lem:T5dominance}
	Let $k_1\in [7]$. Consider the types $T_i$, $i\in[6]$ given by Table~\ref{tab:configs1}. Then there is a $\gamma\in (0,1)$ and positive integers $p$ and $q$ such that, for all $i\in[6],\ i\neq 5$ and all positive integers $t$, we have $\hatN(T_i)\le \gamma^t \hatN(T_5)$.
\end{lem}
\begin{proof}
	Let
	\[
	\calL=	\left\{ 
	\frac{\log\left(\frac{(3+k_1)^2}{3}\right)}{\log\left(9+k_1\right)},
	\frac{\log \left(3+k_1\right)}{\log\left(\frac{9+k_1}{3}\right)},
	\frac{\log\left(\frac{3+k_1}{3}\right)}{\log\left(\frac{9+k_1}{3+k_1}\right)},
	\frac{\log 3}{\log\left(\frac{9+k_1}{9}\right)}
	\right\}
	\]
	and 
	\[
	R=\frac{\log 3}{\log\left(\frac{9+2k_1}{9+k_1}\right)}.
	\]
	For each of the seven possible values of $k_1$ we can check (for example by computer) that every member of $\calL$ is less than $R$. Thus, we can choose $p$ and $q$ so that
	\begin{equation}\label{equ:T5fracqp1}
	\forall L\in \calL,\ L< \frac{q}{p} < R.
	\end{equation}
	
	We check the sought-for bound for each $i\in[6]$, $i\neq 5$:
	\begin{list}{}%
		{\renewcommand\makelabel[1]{#1:\hfill}%
			\settowidth\labelwidth{\makelabel{$T_2$ \quad}}%
			\setlength\leftmargin{\labelwidth}
			\addtolength\leftmargin{\labelsep}}
		\item[$T_1$]  $\frac{\hatN(T_1)}{\hatN(T_5)}=((3+k_1)^2/3)^{pt}(1/(9+k_1))^{qt} < \gamma^t$ is fulfilled for some sufficiently large $\gamma<1$ if and only if $((3+k_1)^2/3)^{p}<(9+k_1)^{q}$ which is equivalent to $\displaystyle\log\left((3+k_1)^2/3\right)/\log\left(9+k_1\right) <q/p$. This is true by \eqref{equ:T5fracqp1}.
		\item[$T_2$]  $\frac{\hatN(T_2)}{\hatN(T_5)}=(3+k_1)^{pt}(3/(9+k_1))^{qt} < \gamma^t$ is fulfilled for some sufficiently large $\gamma<1$ if and only if $(3+k_1)^{p}<((9+k_1)/3)^{q}$. This is true by \eqref{equ:T5fracqp1}.
		\item[$T_3$]  $\frac{\hatN(T_3)}{\hatN(T_5)}=((3+k_1)/3)^{pt}((3+k_1)/(9+k_1))^{qt} < \gamma^t$ is fulfilled for some sufficiently large $\gamma<1$ if and only if $((3+k_1)/3)^{p}<((9+k_1)/(3+k_1))^{q}$. This is true by \eqref{equ:T5fracqp1}.
		\item[$T_4$]  $\frac{\hatN(T_4)}{\hatN(T_5)}=3^{pt}(9/(9+k_1))^{qt} < \gamma^t$ is fulfilled for some sufficiently large $\gamma<1$ if and only if $3^{p}<((9+k_1)/9)^{q}$. This is true by \eqref{equ:T5fracqp1}.
		\item[$T_6$]  $\frac{\hatN(T_6)}{\hatN(T_5)}=(1/3)^{pt}((9+2k_1)/(9+k_1))^{qt} < \gamma^t$ is fulfilled for some sufficiently large $\gamma<1$ if and only if $((9+2k_1)/(9+k_1))^{q}<3^{p}$. This is true by \eqref{equ:T5fracqp1}.
	\end{list}
\end{proof}

\begin{lem}\label{lem:T9dominance}
	Let $k_1\in \{3,4,5,6\}$. Consider the types $T_i$, $i\in[10]$ given by Table~\ref{tab:configs2}. Then there is a $\gamma\in (0,1)$ and positive integers $p$ and $q$ such that, for all $i\in[10],\ i\neq 9$ and all positive integers $t$, we have $\hatN(T_i)\le \gamma^t \hatN(T_9)$.
\end{lem}
\begin{proof}
	
	Let
	\[
	\calL=	\left\{ 
	\frac{\log\left(\frac{(4+k_1)^2}{3}\right)}{\log\left(10+k_1\right)},
	\frac{\log \left(4+k_1\right)}{\log\left(\frac{10+k_1}{3}\right)},
	\frac{\log\left(\frac{4+k_1}{3}\right)}{\log\left(\frac{10+k_1}{4+k_1}\right)},
	\frac{\log 3}{\log\left(\frac{10+k_1}{9}\right)}
	\right\}
	\]
	and 
	\[
	R=\frac{\log 3}{\log\left(\frac{12+2k_1}{10+k_1}\right)}.
	\]
	For each of the four possible values of $k_1$ we can check (for example by computer) that every member of $\calL$ is less than $R$. Thus, we can choose $p$ and $q$ so that
	\begin{equation}\label{equ:T9fracqp1}
	\forall L\in \calL,\ L< \frac{q}{p} < R.
	\end{equation}
	
	Suppose that $T$ and $T'$ are types listed in Table~\ref{tab:configs2} which are distinct from $T_9$ and have the property that $\hatN(T')\le\hatN(T)$ for all $k_1\in \{3,4,5,6\}$. Then the sought-for bound automatically holds for $T'$ if it holds for $T$. 
	
	We check the sought-for bound for each $i\in[10]$, $i\neq 9$:

	\begin{list}{}%
		{\renewcommand\makelabel[1]{#1:\hfill}%
			\settowidth\labelwidth{\makelabel{$T_2$ \quad}}%
			\setlength\leftmargin{\labelwidth}
			\addtolength\leftmargin{\labelsep}}
		\item[$T_1$]  $\frac{\hatN(T_1)}{\hatN(T_9)}=((4+k_1)^2/3)^{pt}(1/(10+k_1))^{qt} < \gamma^t$ is fulfilled for some sufficiently large $\gamma<1$ if and only if $((4+k_1)^2/3)^{p}<(10+k_1)^{q}$. This is true by \eqref{equ:T9fracqp1}.
		\item[$T_2$]  $\hatN(T_2)\le \hatN(T_3)$.
		\item[$T_3$]  $\frac{\hatN(T_3)}{\hatN(T_9)}=(4+k_1)^{pt}(3/(10+k_1))^{qt} < \gamma^t$ is fulfilled for some sufficiently large $\gamma<1$ if and only if $(4+k_1)^{p}<((10+k_1)/3)^{q}$. This is true by \eqref{equ:T9fracqp1}.
		\item[$T_4$]  $\frac{\hatN(T_4)}{\hatN(T_9)}=((4+k_1)/3)^{pt}((4+k_1)/(10+k_1))^{qt} < \gamma^t$ is fulfilled for some sufficiently large $\gamma<1$ if and only if $((4+k_1)/3)^{p}<((10+k_1)/(4+k_1))^{q}$. This is true by \eqref{equ:T9fracqp1}.
		\item[$T_5$]  $\hatN(T_5)\le \hatN(T_8)$.
		\item[$T_6$]  $\hatN(T_6)\le \hatN(T_8)$, for all $k\in \{3,4,5,6\}$.
		\item[$T_7$]  $\frac{\hatN(T_7)}{\hatN(T_9)}=(2/3)^{pt}((6+k_1)/(10+k_1))^{qt} < \gamma^t$ is fulfilled for $2/3 \cdot (6+k_1)/(10+k_1)<\gamma <1$.
		\item[$T_8$]  $\frac{\hatN(T_8)}{\hatN(T_9)}=3^{pt}(9/(10+k_1))^{qt} < \gamma^t$ is fulfilled for some sufficiently large $\gamma<1$ if and only if $3^{p}<((10+k_1)/9)^{q}$. This is true by \eqref{equ:T9fracqp1}.
		\item[$T_{10}$]  $\frac{\hatN(T_{10})}{\hatN(T_9)}=(1/3)^{pt}((12+2k_1)/(10+k_1))^{qt} < \gamma^t$ is fulfilled for some sufficiently large $\gamma<1$ if and only if $((12+2k_1)/(10+k_1))^{q}<3^{p}$. This is true by \eqref{equ:T9fracqp1}.
	\end{list}
\end{proof}

\begin{rem}
	We point out that the proofs of Lemmas~\ref{lem:T5dominance} and~\ref{lem:T9dominance} break for larger $k_1$. (For larger $k_1$ there exists some lower bound on $p/q$ which exceeds some upper bound on that ratio.)
	Lemma~\ref{lem:T9dominance} also breaks for $k_1=1$ and $k_1=2$. This matches the results from Section~\ref{sec:bis-easiness} which show that approximately counting retractions to $\X{k_1}{1}{1}$ is actually $\bis$-easy for these values of $k_1$.
\end{rem}

\begin{lem}\label{lem:hardNeighbourhood2}
	If $k_1\ge 1$, then $\sat \leap \Hom{\X{k_1}{0}{1}}$.
\end{lem}
\begin{proof}
	We make a case distinction depending on $k_1$. The first case is the main work of the proof and we use the dominance of the type $T_5$ from Table~\ref{tab:configs1} for $k_1 \le 6$ as shown in Lemma~\ref{lem:T5dominance}. The second case ($k_1\ge 7$) then follows from Lemma~\ref{lem:Kelk5.1}.
	
	\bigskip
	\noindent {\bf\boldmath Case 1: $k_1 \in [6]$.}
	We use a reduction from $\largecut$, which is known to be $\sat$-hard (see~\cite{DGGJApprox}). Let $G$ and $K$ be an input to $\largecut$, $n$ be the number of vertices of $G$ and $\eps\in (0,1)$ be the parameter of the desired precision. To shorten notation let $H=\X{k_1}{0}{1}$. From $G$ we construct an input $G'$ to $\Hom{H}$ by introducing vertex and edge gadgets. We assume that the vertices of $H$ are labelled as in Figure~\ref{fig:2specialNeighbourhoods} (on the left).
	
	Let $p$, $q$ be positive integers that fulfil \eqref{equ:T5fracqp1}. Note that $p$ and $q$ only depend on $k_1$ which is a parameter of $H$ and therefore does not depend on the input $G$. 
	We define the parameter~$t$ of the gadget graph~$\Jpqt$ to be $t=n^4$. We also define a new parameter $s=n+2$. 
	
	For each vertex $v\in V(G)$ we introduce a vertex gadget $G'_v$ which is a graph $\Jpqt$ as given in Figure~\ref{fig:Jpqt}. We denote the corresponding sets $A,B,B',A'$ by $A_v,B_v,B'_v$ and $A'_v$, respectively. We say that two gadgets $G'_u$ and $G'_v$ are adjacent if $u$ and $v$ are adjacent in $G$.
	
	For every edge $e=\{u,v\}\in E(G)$ we introduce an edge gadget as follows. We introduce two 
	size-$s$ independent sets,
	denoted by $S_e$ and $S'_e$. As shown in Figure~\ref{fig:EdgeGadget1} we construct the set of edges
	\[
	E'_e = (\ucp{B_u}{S_e}) \cup (\ucp{B'_u}{S'_e})\cup (\ucp{B_v}{S'_e}) \cup (\ucp{B'_v}{S_e}).
	\]
	
	\begin{figure}[h!]\centering
		{\def\scaleFactor{1}
			\begin{tikzpicture}[scale=.7, every loop/.style={min distance=10mm,looseness=10}]
			
			\draw (-6,2) -- (0,1.5);
			\draw (-6,2) -- (0,-1.5);
			\draw (-6,-2) -- (0,-1.5);
			\draw (-6,-2) -- (0,1.5);
			
			\draw (6,2) -- (0,1.5);
			\draw (6,2) -- (0,-1.5);
			\draw (6,-2) -- (0,-1.5);
			\draw (6,-2) -- (0,1.5);
			
			\draw (-6,-4) -- (0,-4.5);
			\draw (-6,-4) -- (0,-7.5);
			\draw (-6,-8) -- (0,-7.5);
			\draw (-6,-8) -- (0,-4.5);
			
			\draw (6,-4) -- (0,-4.5);
			\draw (6,-4) -- (0,-7.5);
			\draw (6,-8) -- (0,-7.5);
			\draw (6,-8) -- (0,-4.5);
			
			\draw (-6,0) ellipse (1cm and 2cm) [fill=white];
			\draw (0,0) ellipse (1cm and 1.5cm) [fill=white];
			\draw (6,0) ellipse (1cm and 2cm) [fill=white];
			
			\node at (-6,0) {$B_u$};
			\node at (0,0) {$S_e$};
			\node at (6,0) {$B'_v$};
			
			\draw (-6,-6) ellipse (1cm and 2cm) [fill=white];
			\draw (0,-6) ellipse (1cm and 1.5cm) [fill=white];
			\draw (6,-6) ellipse (1cm and 2cm) [fill=white];
			
			\node at (-6,-6) {$B'_u$};
			\node at (0,-6) {$S'_e$};
			\node at (6,-6) {$B_v$};
			\end{tikzpicture}
		}
		\caption{The edge gadget for the edge $e=\{u,v\}$.}
		\label{fig:EdgeGadget1}
	\end{figure}
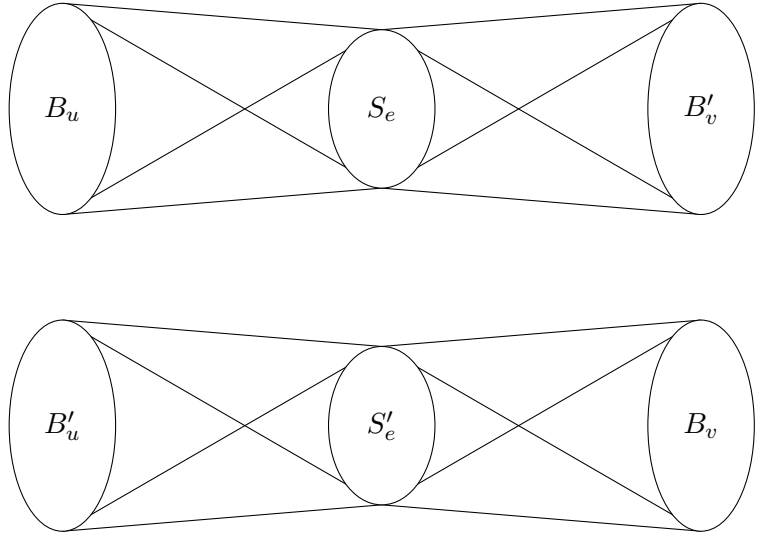
	
	Putting the pieces together, $G'$ is the graph with 
	\[
	V(G')= \bigcup_{v\in V(G)} V(G'_v) \cup \bigcup_{e\in E(G)} \left(S_e\cup S'_e\right)
	\qquad\text{ and }\qquad
	E(G')= \bigcup_{v\in V(G)} E(G'_v) \cup \bigcup_{e\in E(G)} E'_e.
	\]
	
	Let $h$ be a homomorphism from $G$ to $H$, $v$ be some vertex of $G$ and $G'_v$ be the corresponding vertex gadget. Then $h\vert_{V(G'_v)}$ corresponds to a homomorphism from $\Jpqt$ to $H$ and therefore has an $H$-type. 
	
	We say that a homomorphism from $G'$ to $H$ is \emph{full} if its restriction to each vertex gadget is either of type $T_5$ (from Table~\ref{tab:configs1}) or of its symmetric type (let us call it $T'_5$). The cut corresponding to a full homomorphism $h$ partitions $V(G)$ into those vertices $v$ for which $h\vert_{G'_v}$ has type $T_5$ and those for which $h\vert_{G'_v}$ has type $T'_5$. We say that a full homomorphism is \emph{$K$-large} if the size of the corresponding cut is equal to $K$, otherwise we say that the homomorphism is \emph{$K$-small}.
	Consider a full homomorphism $h$ from $G'$ to $H$.
	\begin{itemize}
		\item
		For an edge $e=\{u,v\}$ of $G$ suppose $h\vert_{G'_u}$ has type $T_5$ and $h\vert_{G'_v}$ has type $T'_5$. Note that by the definition of the edge gadget, we have $h(S_e)\subseteq \NH(h(B_u))\cap \NH(h(B'_v))$. Then the vertices in $S_e$ can be mapped to any of $\{r_1,r_2,b\}$, whereas all vertices in $S_e'$ have to be mapped to $b$ (the sole common neighbour of $r_1$, $r_2$, $b$ and the vertices in $\calG$).
		\item
		Suppose instead that $h\vert_{G'_u}$ and $h\vert_{G'_v}$ have the same type $T_5$ or $T_5'$. Then the homomorphism $h$ has to map the vertices in both $S_e$ and $S_e'$ to $b$. 
	\end{itemize}
	Thus, every pair of adjacent gadgets of different types contributes a factor of $3^s$ to the number of full homomorphisms, whereas every pair of adjacent gadgets of the same type only contributes a factor of $1$. Recall the definition of $N(T)$ as the number of homomorphisms from $\Jpqt$ to $H$ that have type $T$. Then for $\ell\ge 1$ every size-$\ell$ cut of $G$ arises in $2\cdot N(T_5)^n\cdot 3^{s\ell}$ ways as a full homomorphism from $G'$ to $H$.
	
	Let $L$ be the number of solutions to $\largecut$ with input $G$ and $K$ (our goal is to approximate this number). We partition the homomorphisms from  $G'$ to $H$ into three different sets. $Z^*$ is the number of $K$-large (full) homomorphisms, $Z_1$ is the number of homomorphisms that are full but $K$-small and $Z_2$ is the number of non-full homomorphisms.
	Then we have $L= Z^*/(2 N(T_5)^n 3^{sK})$ and $\hom{G'}{H} = Z^* +Z_1 +Z_2$. Thus it remains to show that $(Z_1 +Z_2)/(2 N(T_5)^n 3^{sK})\le 1/4$  for our choice of $p$, $q$, $t$ and $s$. Under this assumption we then have $\displaystyle\hom{G'}{H}/(2 N(T_5)^n 3^{sK}) \in [L,L+1/4]$ and a single oracle call to determine $\hom{G'}{H}$ with precision $\delta=\eps/21$ suffices to determine $L$ with the sought-for precision as demonstrated in~\cite{DGGJApprox}.
	
	Now we prove $(Z_1 +Z_2)/(2 N(T_5)^n 3^{sK})\le 1/4$.
	As there are at most $2^n$ ways to assign a type $T_5$ or $T'_5$ to the $n$ vertex gadgets in $G'$ we have $Z_1 \le 2^n \cdot N(T_5)^n\cdot 3^{s(K-1)}$.
	Then we obtain the following bound since $s=n+2$:
	\[
	\frac{Z_1}{2 N(T_5)^n 3^{sK}}\le \frac{2^n N(T_5)^n 3^{s(K-1)}}{2 N(T_5)^n 3^{sK}} = \frac{2^n}{2\cdot 3^s}\le \frac{1}{8}.
	\]
	We obtain a similar bound for $Z_2$: From Lemmas~\ref{lem:maximalTypeDominates},~\ref{lem:tableOfMaxTypes1} and~\ref{lem:T5dominance} we know that for our choice of $p$ and $q$ there exists $\gamma\in(0,1)$ such that for every $H$-type $T$ that is not $T_5$ or $T_5'$ we have $\hatN(T)\le \gamma^t\hatN(T_5)$. Using Lemma~\ref{lem:hatN} this gives $N(T) \le 2\gamma^t N(T_5)$ for sufficiently large $t$ with respect to $p$, $q$ and $k_1$ (which do not depend on the input $G$). Since $t=n^4$ we can assume that $t$ is sufficiently large with respect to $p$ and $q$ as otherwise the input size is bounded by a constant (in which case we can solve $\largecut$ in constant time).
	
	For each $H$-type $T=(T_1,T_2,T_3)$, the cardinality of each set~$T_i$
	is bounded above by $\max\{\abs{V(H)},2\abs{E(H)}\}=12+2k_1$ and hence there are at most $\left(2^{12+2k_1}\right)^3$ different types. Furthermore, as $H$ has $3+k_1$ vertices, there are at most $(3+k_1)^{2sn^2}$ possible functions from the at most $2sn^2$ vertices in $\bigcup_{e\in E(G)} (S_e\cup S'_e)$ to vertices in $H$. Since $t=n^4$ and $s=n+2$ we obtain
	\begin{align*}
	\frac{Z_2}{2 N(T_5)^n 3^{sK}}
	&\le \frac{\left(2^{12+2k_1}\right)^{3n}\cdot N(T_5)^{n-1}\cdot 2\gamma^tN(T_5)\cdot(3+k_1)^{2sn^2}}{2 N(T_5)^n 3^{sK}} \\
	&= \gamma^t\cdot\frac{\left(2^{12+2k_1}\right)^{3n}(3+k_1)^{2sn^2}}{3^{sK}}\le \frac{1}{8}.
	\end{align*}
	The last inequality holds 
	for sufficiently large~$n$
	as \[\frac{\left(2^{12+2k_1}\right)^{3n}(3+k_1)^{2sn^2}}{3^{sK}}\le C^{n^3}\] for some constant $C$ that only depends on $H$,
	but not on the input $G$, whereas $t=n^4$.
	{\bf (End of Case 1)}
	
	\bigskip
	\noindent {\bf\boldmath Case 2: $k_1 \ge 7$.}
	We will show that in this case we can apply Lemma~\ref{lem:Kelk5.1} to obtain $\sat \leap \Hom{H}$.
	We have $F(H)=\{b\}$. Therefore, $\abs{F(H)}\cdot \abs{V(H)}=\abs{V(H)} = 3+k_1 \ge 10$.
	Let $(S,T)$ be a pair with $\emptyset \subseteq S,T \subseteq V(H)$ satisfying $S\subseteq \Nb{H}(T)$, $T\subseteq \Nb{H}(S)$ and both $S\neq\{b\}$ and $T\neq \{b\}$ to meet the requirements of Lemma~\ref{lem:Kelk5.1}. We have to show that
	$\abs{S}\cdot \abs{T} < \abs{F(H)}\cdot \abs{V(H)}=3+k_1$. Note that for every vertex $u\neq b$ of $H=\X{k_1}{0}{1}$ it holds that $\abs{\NH(u)}\le 3$. Therefore, $\abs{S} \le\abs{\NH(T)} \le 3$ and analogously $\abs{T} \le\abs{\NH(S)} \le 3$. Hence
	\[
	\abs{S}\cdot \abs{T} \le 9 < 10\le 3+k_1.
	\] 
	{\bf (End of Case 2)}
\end{proof}

\begin{lem}\label{lem:hardNeighbourhood3}
	If $k_1\ge 3$, then $\sat \leap \Hom{\X{k_1}{1}{1}}$.
\end{lem}
\begin{proof}
	This proof is very similar to the proof of Lemma~\ref{lem:hardNeighbourhood2}. We give the details for the sake of completeness.
	As before, we make a case distinction depending on $k_1$. The first case is about the dominance of the type $T_9$ from Table~\ref{tab:configs2} for $k_1 \in \{3,4,5,6\}$ as shown in Lemma~\ref{lem:T9dominance}. Otherwise, we use Lemma~\ref{lem:Kelk5.1}. 
	
	\bigskip
	\noindent {\bf\boldmath Case 1: $k_1 \in \{3,4,5,6\}$.}
	We use a reduction from $\largecut$ ($\sat$-hard by~\cite{DGGJApprox}). Let $G$ and $K$ be an input to $\largecut$, $n$ be the number of vertices of $G$ and $\eps\in (0,1)$ be the parameter of the desired precision. To shorten notation let $H=\X{k_1}{1}{1}$. We assume that the vertices of $H$ are labelled as in Figure~\ref{fig:2specialNeighbourhoods} (on the right).
	
	Let $p$, $q$ be positive integers that fulfil \eqref{equ:T9fracqp1}. Note that $p$ and $q$ only depend on $k_1$ which is a parameter of $H$ and therefore does not depend on the input $G$. 
	We will define the parameter~$t$ of the gadget graph~$J$ to be $t=n^4$. We also define a new parameter $s=n+2$. 
	
	For each vertex $v\in V(G)$ we introduce a vertex gadget $G'_v$ which is a graph $\Jpqt$ as given in Figure~\ref{fig:Jpqt}. We denote the corresponding sets $A,B,B',A'$ by $A_v,B_v,B'_v$ and $A'_v$, respectively. We say that two gadgets $G'_u$ and $G'_v$ are adjacent if $u$ and $v$ are adjacent in $G$. For every edge $e=\{u,v\}\in E(G)$ we use exactly the same edge gadget as in the proof of Lemma~\ref{lem:hardNeighbourhood2} (see Figure~\ref{fig:EdgeGadget1}). 
	$G'$ is the graph with 
	\[
	V(G')= \bigcup_{v\in V(G)} V(G'_v) \cup \bigcup_{e\in E(G)} \left(S_e\cup S'_e\right)
	\qquad\text{ and }\qquad
	E(G')= \bigcup_{v\in V(G)} E(G'_v) \cup \bigcup_{e\in E(G)} E'_e.
	\]
	
	Let $h$ be a homomorphism from $G$ to $H$, $v$ be some vertex of $G$ and $G'_v$ be the corresponding vertex gadget. Then $h\vert_{V(G'_v)}$ corresponds to a homomorphism from $\Jpqt$ to $H$ and therefore has an $H$-type. 
	
	We say that a homomorphism from $G'$ to $H$ is \emph{full} if its restriction to each vertex gadget is either of type $T_9$ (from Table~\ref{tab:configs2}) or of its symmetric type (let us call it $T'_9$). The cut corresponding to a full homomorphism $h$ partitions $V(G)$ into those vertices $v$ for which $h\vert_{G'_v}$ has type $T_9$ and those for which $h\vert_{G'_v}$ has type $T'_9$. We say that a full homomorphism is \emph{$K$-large} if the size of the corresponding cut is equal to $K$, otherwise we say that the homomorphism is \emph{$K$-small}.
	Consider a full homomorphism $h$ from $G'$ to $H$.
	\begin{itemize}
		\item
		For an edge $e=\{u,v\}$ of $G$ suppose $h\vert_{G'_u}$ has type $T_9$ and $h\vert_{G'_v}$ has type $T'_9$. Note that by the definition of the edge gadget, we have $h(S_e)\subseteq \NH(h(B_u))\cap \NH(h(B'_v))$. Then the vertices in $S_e$ can be mapped to any of $\{r_1,r_2,b\}$, whereas all vertices in $S_e'$ have to be mapped to $b$ (the sole common neighbour of $r_1$, $r_2$, $b$, $c$ and the vertices in $\calG$).
		\item
		Suppose instead that $h\vert_{G'_u}$ and $h\vert_{G'_v}$ have the same type $T_9$ or $T_9'$. Then the homomorphism $h$ has to map the vertices in both $S_e$ and $S_e'$ to $b$. 
	\end{itemize}
	Thus, every pair of adjacent gadgets of different types contributes a factor of $3^s$ to the number of full homomorphisms, whereas every pair of adjacent gadgets of the same type only contributes a factor of $1$. Recall the definition of $N(T)$ as the number of homomorphisms from $\Jpqt$ to $H$ that have type $T$. Then for $\ell\ge 1$ every size-$\ell$ cut of $G$ arises in $2\cdot N(T_9)^n\cdot 3^{s\ell}$ ways as a full homomorphism from $G'$ to $H$.
	
	Let $L$ be the number of solutions to $\largecut$ with input $G$ and $K$ (our goal is to approximate this number). We partition the homomorphisms from  $G'$ to $H$ into three different sets. $Z^*$ is the number of $K$-large (full) homomorphisms, $Z_1$ is the number of homomorphisms that are full but $K$-small and $Z_2$ is the number of non-full homomorphisms.
	Then we have $L= Z^*/(2 N(T_9)^n 3^{sK})$ and $\hom{G'}{H} = Z^* +Z_1 +Z_2$. Thus it remains to show that $(Z_1 +Z_2)/(2 N(T_9)^n 3^{sK})\le 1/4$  for our choice of $p$, $q$, $t$ and $s$. Under this assumption we then have $\displaystyle\hom{G'}{H}/(2 N(T_9)^n 3^{sK}) \in [L,L+1/4]$ and a single oracle call to determine $\hom{G'}{H}$ with precision $\delta=\eps/21$ suffices to determine $L$ with the sought-for precision as demonstrated in~\cite{DGGJApprox}.
	
	Now we prove $(Z_1 +Z_2)/(2 N(T_9)^n 3^{sK})\le 1/4$.
	As there are at most $2^n$ ways to assign a type $T_9$ or $T'_9$ to the $n$ vertex gadgets in $G'$ we have $Z_1 \le 2^n \cdot N(T_9)^n\cdot 3^{s(K-1)}$.
	Then we obtain the following bound since $s=n+2$:
	\[
	\frac{Z_1}{2 N(T_9)^n 3^{sK}}\le \frac{2^n N(T_9)^n 3^{s(K-1)}}{2 N(T_9)^n 3^{sK}} = \frac{2^n}{2\cdot 3^s}\le \frac{1}{8}.
	\]
	We obtain a similar bound for $Z_2$: From Lemmas~\ref{lem:maximalTypeDominates},~\ref{lem:tableOfMaxTypes2} and~\ref{lem:T9dominance} we know that for our choice of $p$ and $q$ there exists $\gamma\in(0,1)$ such that for every $H$-type $T$ that is not $T_9$ or $T_9'$ we have $\hatN(T)\le \gamma^t\hatN(T_9)$. Using Lemma~\ref{lem:hatN} this gives $N(T) \le 2\gamma^t N(T_9)$ for sufficiently large $t$ with respect to $p$, $q$ and $k_1$. Since $t=n^4$ we can assume that $t$ is sufficiently large with respect to $p$ and $q$ as otherwise the input size is bounded by a constant (in which case we can solve $\largecut$ in constant time).
	
	For each $H$-type $T=(T_1,T_2,T_3)$, the cardinality of each set~$T_i$
	is bounded above by $\max\{\abs{V(H)},2\abs{E(H)}\}=16+2k_1$ and hence there are at most $\left(2^{16+2k_1}\right)^3$ different types. Furthermore, as $H$ has $4+k_1$ vertices, there are at most $(4+k_1)^{2sn^2}$ possible functions from the at most $2sn^2$ vertices in $\bigcup_{e\in E(G)} (S_e\cup S'_e)$ to vertices in $H$. Since $t=n^4$ and $s=n+2$ we obtain
	\begin{align*}
	\frac{Z_2}{2 N(T_9)^n 3^{sK}}
	&\le \frac{\left(2^{16+2k_1}\right)^{3n}\cdot N(T_9)^{n-1}\cdot 2\gamma^tN(T_9)\cdot(4+k_1)^{2sn^2}}{2 N(T_9)^n 3^{sK}} \\
	&= \gamma^t\cdot\frac{\left(2^{16+2k_1}\right)^{3n}(4+k_1)^{2sn^2}}{3^{sK}}\le \frac{1}{8}.
	\end{align*}
	The last inequality holds 
	for sufficiently large~$n$
	as \[\frac{\left(2^{16+2k_1}\right)^{3n}(4+k_1)^{2sn^2}}{3^{sK}}\le C^{n^3}\] for some constant $C$ that only depends on $H$,
	but not on the input $G$, whereas $t=n^4$.
	{\bf (End of Case 1)}
	
	\bigskip
	\noindent {\bf\boldmath Case 2: $k_1 \ge 6$.}
	We will show that in this case we can apply Lemma~\ref{lem:Kelk5.1} to obtain $\sat \leap \Hom{H}$.
	We have $F(H)=\{b\}$. Therefore, $\abs{F(H)}\cdot \abs{V(H)}=\abs{V(H)} = 4+k_1 \ge 10$.
	Let $(S,T)$ be a pair with $\emptyset \subseteq S,T \subseteq V(H)$ satisfying $S\subseteq \NH(T)$, $T\subseteq \NH(S)$ and both $S\neq\{b\}$ and $T\neq \{b\}$ to meet the requirements of Lemma~\ref{lem:Kelk5.1}. We have to show that
	$\abs{S}\cdot \abs{T} < \abs{F(H)}\cdot \abs{V(H)}=4+k_1$. Note that for every vertex $u\neq b$ of $H=\X{k_1}{1}{1}$ it holds that $\abs{\NH(u)}\le 3$. Therefore, $\abs{S} \le\abs{\NH(T)} \le 3$ and analogously $\abs{T} \le\abs{\NH(S)} \le 3$. Hence
	\[
	\abs{S}\cdot \abs{T} \le 9 < 10\le 4+k_1.
	\] 
	{\bf (End of Case 2)}
\end{proof}

\begin{lem}\label{lem:hardNeighbourhood4}
	If $k_1\ge 5$, then $\sat \leap \Hom{\X{k_1}{0}{2}}$.
\end{lem}
\begin{proof}
	Let $k_1\ge 5$. To shorten notation let $H=\X{k_1}{0}{2}$. Again we use Lemma~\ref{lem:Kelk5.1} to obtain $\sat \leap \Hom{H}$. We have $F(H)=\{b\}$ and $\abs{F(H)}\cdot \abs{V(H)}=\abs{V(H)} = 5+k_1 \ge 10$.
	The remainder of the proof is identical to Case 2 in the proof of Lemma~\ref{lem:hardNeighbourhood3}.
\end{proof}

\begin{lem} \label{lem:degree2bristle}
	Let $H$ be a graph and $b\in V(H)$ be a looped vertex with an unlooped neighbour $g\in V(H)$. If $\abs{\NH(g)}\ge 2$ ($g$ has at least $2$ neighbours in $H$) and for all $u\in \NH(b)\setminus\{g\}$ we have $\abs{\NH(u)\cap \NH(g)} = 1$ ($b$ is the only common neighbour of $g$ and $u$), then
	$\sat \leap \Ret{H}$.
\end{lem}
\begin{proof}
	Let $H'$ be the graph obtained by replacing the vertex $g$ in $H[\NH(b)]$ by an independent set $I$ of size $\abs{\NH(g)}^{s}$, where $s=2\ceil{\log_{\abs{\NH(g)}}(\abs{\NH(b)})}$. This is well-defined as $\abs{\NH(g)}>1$. The choice of $s$ will become clear in a moment. Within the graph $H[\NH(b)]$, $g$ is adjacent only to $b$ by the assumption that $\abs{\NH(b)\cap \NH(g)}=1$. Therefore, each vertex in $I$ shares an edge only with $b$. The transformation is depicted in Figure~\ref{fig:TrafoWeightedToUnweighted}. 
	
	\begin{figure}[ht]
		\centering
		\begin{minipage}{.45 \textwidth}
			\centering
			\begin{tikzpicture}[scale=1, baseline=0.36cm, every loop/.style={min distance=10mm,looseness=10}]
			
			\draw (0,0) circle[x radius=2,y radius=1];
			
			\filldraw (0,-1) node(b){} circle[radius=3pt]  -- (-90:2.6cm) node(g){} circle[radius=3pt];
			
			
			\path[-] (b.center) edge  [in=-240,out=-300,loop] node {} ();

			\node at ($(0,0)$) {$\NH(b)\setminus \{g\}$};
			\node at ($(b)+(.25cm,-.25cm)$) {$b$};
			\node at ($(g)+(0cm,-.5cm)$) {$g$};
			
			\end{tikzpicture}
		\end{minipage}
		\begin{minipage}{.45 \textwidth}
			\centering
			\begin{tikzpicture}[scale=1, baseline=0.36cm, every loop/.style={min distance=10mm,looseness=10}]
			
			\draw (0,0) circle[x radius=2,y radius=1];
			
			\filldraw (0,-1) node(b){} circle[radius=3pt]  -- (-120:3cm) node(g1){} circle[radius=3pt];
			\filldraw (b.center)  -- (-60:3cm) node(gk){} circle[radius=3pt];
			
			
			\path[-] (b.center) edge  [in=-240,out=-300,loop] node {} ();

			\node at ($(0,0)$) {$\NH(b)\setminus \{g\}$};
			\node at ($(b)+(0cm,-.5cm)$) {$b$};
			\node at ($(g1)+(-1cm,-.5cm)$) {$I:$};
			\node at ($(g1)+(0cm,-.5cm)$) {$1$};
			\node at ($(gk)+(0cm,-.5cm)$) {$\displaystyle \abs{\NH(g)}^{s}$};
			
			\coordinate (pmid) at ($(b)+(-90:1.6cm)$);
			\coordinate (cdot1) at ($(pmid)+(-.3cm,0cm)$);
			\coordinate (cdot2) at ($(pmid)+(0cm,0cm)$);
			\coordinate (cdot3) at ($(pmid)+(.3cm,0cm)$);
			\fill (cdot1) circle[radius=1.5pt];
			\fill (cdot2) circle[radius=1.5pt];
			\fill (cdot3) circle[radius=1.5pt];
			
			\end{tikzpicture}
		\end{minipage}
		\caption{$H[\NH(b)]$ on the left and $H'$ on the right.}
		\label{fig:TrafoWeightedToUnweighted}
	\end{figure}
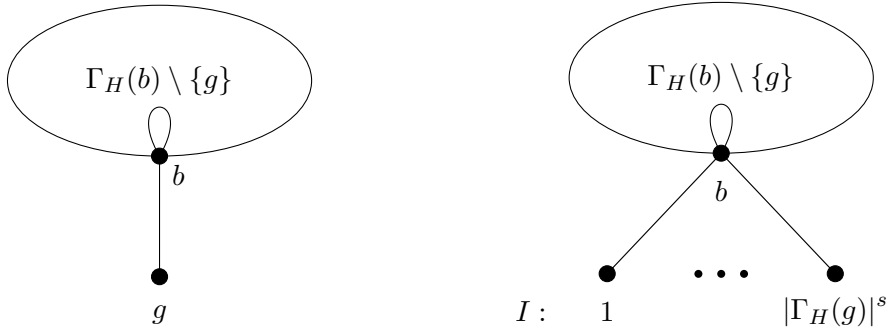
	
	First, we will show $\sat$-hardness for $\Ret{H'}$ and we will apply Lemma~\ref{lem:Kelk5.1} to achieve this. Let us check that the requirements are met. First note that $F(H')=\{b\}$. Now we have to show that, for every pair $(S,T)$ with $\emptyset \subseteq S,T \subseteq V(H')$ satisfying $S\subseteq \Nb{H'}(T)$, $T\subseteq \Nb{H'}(S)$ and both $S\neq\{b\}$ and $T\neq \{b\}$ it holds that
	$\abs{S}\cdot \abs{T} < \abs{F(H')}\cdot \abs{V(H')}=\abs{V(H')}$. First note that if there exists a vertex $u\in I\cap S$ then $T=\{b\}$ as $b$ is the only neighbour of $u$ in $H'$. Hence $I\cap S=\emptyset$. By the same reasoning it holds that $I\cap T=\emptyset$. Then $\abs{S},\abs{T} \le \abs{\Nb{H'}(b)\setminus I}\le \abs{\NH(b)}$ and, by our choice $s=2\ceil{\log_{\abs{\NH(g)}}(\abs{\NH(b)})}$, we can conclude that
	\[
	\abs{S}\cdot \abs{T} \le \abs{\NH(b)}^2 \le \abs{\NH(g)}^{s} < \abs{V(H')} = \abs{F(H')}\cdot  \abs{V(H')}.
	\]
	This proves $\sat \leap \Ret{H'}$.
	
	To complete the proof of the lemma we show the following claim.
	
	\bigskip
	\noindent{\bf\boldmath Claim: $\Ret{H'}\leap \Ret{H}$.}
	\smallskip
	
	\noindent {\it Proof of the claim:}\quad
	
	Let $(G,\boldS)$ be an input to $\Ret{H'}$.
	Let $w$ be a weight function on the vertices in $\NH(b)$ with $w(g)=\abs{\NH(g)}^{s}$ and $w(u)=1$ for all $u\in \NH(b)\setminus \{g\}$. By the construction of $H'$ it is standard that
	\begin{equation}\label{equ:degree2bristle1}
	\hom{(G,\boldS)}{H'}=\sum_{h\in \calH((G,\boldS),H[\NH(b)])}\prod_{v\in V(G)} w(h(v)).
	\end{equation}
	
	Now consider the vertices in $\NH(b)$. For $u\in \NH(b)$, let $w'(u)$ be the number of common neighbours of $u$ and $g$ in $H$. (It is essential that we regard all neighbours in $H$, not just the neighbours in $\NH(b)$.)
	By definition $w'(u)=\abs{\NH(g)}$ if $u=g$ and, by the assumptions of this lemma, $w(u)=1$ if $u\in \NH(b)\setminus\{g\}$. Hence, for all $u\in \NH(b)$ we have
	\begin{equation}\label{equ:degree2bristle2}
	w(u)=w'(u)^s.
	\end{equation}
	
	Intuitively, we will use the fact that $g$ has larger ``weight'' $w'$ compared to the other vertices in $\NH(b)$ to ``boost'' the vertex $g$ and make it more likely (by a factor of $\abs{\NH(g)}^s$) to be used in a homomorphism to $H[\NH(b)]$. 
	
	Here are the details: Let $(G,\boldS)$ be an input to $\Ret{H'}$. We construct a graph $G'$ from $G$ in the following way. 
	Let $\beta$ and $\gamma$ be vertices that are distinct from the vertices in $G$. Intuitively, $\beta$ and $\gamma$ will serve as ``pins'' to $b$ and $g$, respectively. In addition, for each $v\in V(G)$, we introduce a independent set $I_v$ of size $s$, see Figure~\ref{fig:degree2bristleGadget}. Then $G'$ is the graph with vertices $V(G')=V(G) \cup \{\beta, \gamma\} \cup \bigcup_{v \in V(G)} I_v$ and edges $E(G')=E(G) \cup \bigl(\ucp{V(G)}{\{\beta\}}\bigr)\cup \bigcup_{v\in V(G)} \bigl(\ucp{I_v}{\{v,\gamma\}}\bigr)$.
	
	\begin{figure}[ht]
		\centering
		\begin{tikzpicture}[scale=1, baseline=0.36cm, every loop/.style={min distance=10mm,looseness=10}]

		\draw[fill=white] (0,0) circle[x radius=2,y radius=1];
		
		\coordinate (mid) at (2,1.5);
		\draw (.75,0) node(v){} -- (1,1.5) node(u2){};
		\draw (v.center) -- (2.83,1.22) node(u1){};
		
		\draw ($(mid)+(-.9,.22)$) -- (2,2.5) node(gamma){} -- ($(mid)+(.9,.22)$);
		\draw[fill=white] (mid) circle[x radius=1,y radius=.5];
		\fill (v) circle[radius=3pt];
		\fill (gamma) circle[radius=3pt];

		\filldraw (v.center) -- (0,1.5) node(beta){} circle[radius=3pt];
		
		\node at (-.5,0) {$G$};
		\node at (mid) {$I_v$};
		\node at ($(mid)+(1.25cm,0)$) {$s$};
		\node at ($(v)+(0,-.3cm)$) {$v$};
		\node at ($(beta)+(0,.5cm)$) {$\beta \rightarrow b$};
		\node at ($(gamma)+(0,.5cm)$) {$\gamma \rightarrow g$};
		
		\end{tikzpicture}
		\caption{The vertex gadget used in the construction of $G'$.}
		\label{fig:degree2bristleGadget}
	\end{figure}
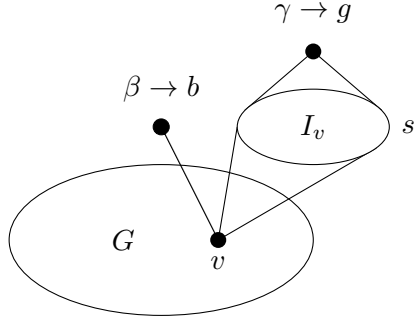

	Consider a homomorphism $h$ from $(G',\boldS')$ to $H$. Then, since every vertex in $V(G)$ is a neighbour of $\beta$ and $h(\beta)=b$, $h\vert_{V(G)}$ is a homomorphism from $(G, \boldS)$ to $H[\NH(b)]$. 
	Furthermore, by the construction of $G'$, for each $v\in V(G)$ and each $u\in I_v$, it holds that $h(u)$ is a common neighbour of $h(v)$ and $g$ in $H$. Recall that there are $w'(h(v))$ such common neighbours. Thus, using~\eqref{equ:degree2bristle1} and~\eqref{equ:degree2bristle2}, we conclude
	\begin{align*}
	\hom{(G',\boldS')}{H} 
	&= \sum_{h\in \calH((G,\boldS),H[\NH(b)])}\prod_{v\in V(G)} w'(h(v))^s\\
	&= \sum_{h\in \calH(G,H[\NH(b)])}\prod_{v\in V(G)} w(h(v))\\
	&=\hom{(G,\boldS)}{H'}.
	\end{align*}
\end{proof}

\subsection{Square-Free Graphs with an Induced {\boldmath$\WR{3}$}}\label{sec:inducedWR3} 
This section can be seen as an extension of Section~\ref{sec:NeighbourhoodOfLoopedVertex} as it essentially shows $\sat$-hardness for graphs of the form $\X{k_1}{k_2}{k_3}$ where $k_2+k_3\ge 3$. Consider a square-free graph $H$ with an induced $\WR{3}$. Suppose that there is an induced $\WR{3}$ such that the neighbourhood of its center $b$ does not contain any triangles. Then $H[\NH(b)]$ is subject to Theorem~\ref{thm:RetGirth5} which shows $\sat \leap \Ret{H[\NH(b)]}$. Then, $\sat \leap \Ret{H}$ by Observation~\ref{obs:PinNeighbourhood}.

However, when considering square-free graphs as opposed to graphs of girth at least $5$, $H[\NH(b)]$ might contain triangles. The smallest open case is displayed in Figure~\ref{fig:weightedWR}.

\begin{figure}[ht]
	\centering
	\begin{tikzpicture}[scale=1, baseline=0.36cm, every loop/.style={min distance=10mm,looseness=10}]

			\filldraw (0,0) node(a){} circle[radius=3pt] --++ (60:1cm) node(b){} circle[radius=3pt] --++ (-60:1cm) node(c){} circle[radius=3pt];
			\filldraw (b.center) --++ (130:1cm) node(d){} circle[radius=3pt];
			\filldraw (b.center) --++ (50:1cm) node(f){} circle[radius=3pt];
			\draw (a.center) -- (c.center);	
			
			\path[-] (a.center) edge  [in=240,out=300,loop] node {} ();
			\path[-] (b.center) edge  [in=30,out=-30,loop] node {} ();
			\path[-] (c.center) edge  [in=240,out=300,loop] node {} ();
			\path[-] (d.center) edge  [in=-240,out=-300,loop] node {} ();
			\path[-] (f.center) edge  [in=-240,out=-300,loop] node {} ();
			
			\node at ($(b)+(-.35cm,0cm)$) {$b$};

	\end{tikzpicture}
	\caption{Smallest square-free graph with induced $\WR{3}$ for which it remains to prove hardness.}
\label{fig:weightedWR}
\end{figure}
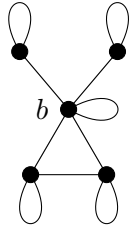

The goal of this section is to prove Lemma~\ref{lem:inducedWR3General}.
\newcommand{\stateleminducedWRGeneral}{Let $H$ be a square-free graph. If $H$ contains a $\WR{3}$ as an induced subgraph then $\sat \leap\Ret{H}$.}
\begin{lem}\label{lem:inducedWR3General}
	\stateleminducedWRGeneral
\end{lem}

First we will introduce some preliminary results.

\begin{defn}
Let $\stirling{a}{b}$ be the \emph{Stirling number of the second kind}, i.e.\ the number of surjective functions from a set of $a$ elements to a set of $b$ elements.
\end{defn}

\begin{lem}[{\cite[Lemma 18]{DGGJApprox}}]\label{lem:DGGJ18}
If $a$ and $b$ are positive integers and $a \ge 2b \ln b$, then
\[
b^a \left( 1- \exp{\left(-\frac{a}{2b}\right)}\right) \le \stirling{a}{b}\le b^a.
\]
\end{lem}
\begin{cor}\label{cor:DGGJ18}
If $a$ and $b$ are positive integers with $a \ge 2b \ln 2b$ and $b\ge 1$, then
\[
\frac{b^a}{2} \le \stirling{a}{b}\le b^a.
\]
\end{cor}

The proof of Lemma~\ref{lem:inducedWR3General} uses the same general idea as the proof of~\cite[Theorem 2]{GJIsing} --- namely that approximating the partition function of the $q$-state ferromagnetic Potts model is $\sat$-equivalent if $q\ge3$ and, in addition, we are allowed to specify that certain vertices have to have a specific spin. Crucially, $\sat$-hardness is known only if this single-vertex ``pinning'' is allowed. In general, the complexity of approximating the partition function of the Potts model is still unresolved and an important open problem. The approach of simulating ferromagnetic Potts with ``pinning'' to obtain hardness results has been used before, for instance in the proofs of~\cite[Lemma 3.6]{GJTreeHoms} and~\cite[Lemma 2.2]{FGZRet}. The gadgets we use here to accomplish the reduction are tailored to the specific problem and different from the gadgets used in similar reductions.

As in the proof of~\cite[Theorem 2]{GJIsing} we use a reduction from the problem of counting so-called multiterminal cuts. We introduce the corresponding definitions from~\cite{GJIsing}. A \emph{multiterminal cut} of a graph $G$ with distinguished vertices $\tau_1, \dots, \tau_q$ (called \emph{terminals}) is a set of edges $E'\subseteq E(G)$ that disconnects the terminals (i.e. ensures that there is no path in $(V(G),E(G)\setminus E')$ that connects any two distinct terminals). The \emph{size} of a multiterminal cut is its cardinality. 
We consider the following computational problem.

\prob
{
$\TCut{q}$.
}
{
A connected irreflexive graph $G$ with $q$ distinct terminals $\tau_1, \dots, \tau_q \in V(G)$ and a positive integer $K$. The input has the property that every multiterminal cut has size at least $K$.
}
{
The number of size-$K$ multiterminal cuts of $G$ with terminals $\tau_1, \dots, \tau_q$.
}
\begin{lem}[{\cite[Section 4]{GJIsing}}]\label{lem:TCutSAT}
Let $q\ge 3$. Then $\TCut{q}\eqap\sat$.
\end{lem}

\begin{defn}\label{def:Phi}
Let $I=(G, \tau_1, \dots, \tau_q, K)$ be an instance of $\TCut{q}$.
$\Phi(I)=\{\phi\from V(G) \to [q] \mid \phi(\tau_i)=i,\ i\in [q]\}$ is the set of \emph{separating functions} from $V(G)$ to $[q]$. For $\phi \in \Phi(I)$ let $\Ecut(\phi)=\{\{u,v\}\in E(G) \mid \phi(u)\neq \phi(v)\}$ and, for $i\in [q]$, let $\mathrm{Mon}_i(\phi)=\{\{u,v\}\in E(G) \mid \phi(u)=\phi(v)=i\}$. Finally, let $\Phi^*(I)=\{\phi \in \Phi(I) \mid \abs{\Ecut(\phi)}=K\}$.
\end{defn}

\begin{obs}\label{obs:Phi}
Let $I=(G, \tau_1, \dots, \tau_q, K)$ be an instance of $\TCut{q}$.
For each $\phi\in\Phi(I)$, $\Ecut(\phi)$ is a multiterminal cut of $I$. On the other hand, each size-$K$ multiterminal cut splits the graph $G$ into exactly $q$ connected components (as otherwise there would exist a multiterminal cut of size less than $K$). Hence each size-$K$ multiterminal cut corresponds exactly to the function $\phi\in \Phi(I)$ for which $\phi(v)=i$ if $v$ is in the same connected component as $\tau_i$. Thus, $\Phi^*(I)$ is the subset of functions in $\Phi(I)$ that correspond to size-$K$ multiterminal cuts. Let $T(I)$ be the number of size-$K$ multiterminal cuts of the instance $I$. Then $T(I)=\abs{\Phi^*(I)}$.
\end{obs}
Now we have all the tools at hand to prove the main lemma of this section.

\begin{leminducedWR3General}
	\stateleminducedWRGeneral
\end{leminducedWR3General}
\begin{proof}
Suppose that $H$ contains a mixed triangle as an induced subgraph, then the statement of this lemma follows from Lemmas~\ref{lem:hardtriangles1} and~\ref{lem:hardtriangles2}.
Hence, for the remainder of this proof let $H$ be a square-free graph with an induced $\WR{3}$ without any induced mixed triangle subgraphs. We choose a vertex $b$ such that $b$ is the center of an induced $\WR{3}$. We consider the graph $H[\NH(b)]$ which is the subgraph of $H$ that is induced by the neighbourhood of $b$. For ease of notation we set $H_b=H[\NH(b)]$. Let $U$ be the set of unlooped neighbours of $b$. Since $H$ does not contain any mixed triangles, for each $u\in U$, $b$ is the only neighbour of $u$ in $H_b$. Since $H$ is square-free, so is $H_b$. Therefore every looped neighbour $w\neq b$ of $b$ has degree $\deg_{H_b}(w)\in \{2,3\}$. By the choice of $b$, $b$ has at least $4$ neighbours including itself, i.e.~we have $\deg_{H_b}(b)=\deg_{H}(b)\ge 4$. 
Let $x_1, \dots, x_k$ be the looped neighbours of degree $2$ and $x_{k+1}, y_{k+1}, \dots, x_q, y_q$ be the looped neighbours of degree $3$, where for each $i\in \{k+1, \dots, q\}$ we have $\{x_i,y_i\}\in E(H_b)$ (where we use that looped vertices can only have looped neighbours since $b$ is the sole neighbour of vertices in $U$). The graph $H_b$ is depicted in Figure~\ref{fig:Hb}.

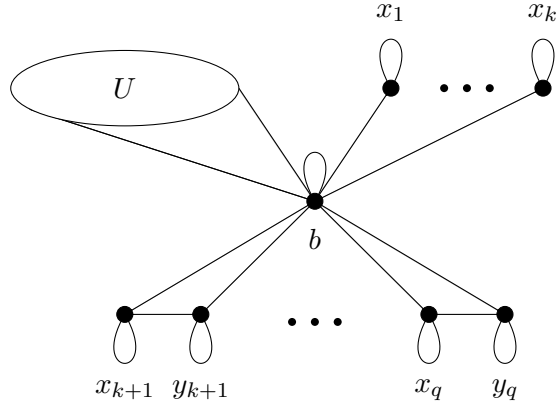
\begin{figure}[ht]
	\centering
	\begin{tikzpicture}[scale=1, baseline=0.36cm, every loop/.style={min distance=10mm,looseness=10}]

	\coordinate (mid) at (-2.5,1.5);
	\filldraw (0,0) node(b){} circle[radius=3pt] -- (-1,1.5) node(u2){};
	\filldraw (b.center) -- (-3.75,1.22) node(u1){};
	\draw[fill=white] (mid) circle[x radius=1.5,y radius=.5];

	\filldraw (b.center) -- (1,1.5) node(l1){} circle[radius=3pt];
	\filldraw (b.center) -- (3,1.5) node(l2){} circle[radius=3pt];
	\filldraw (b.center) -- (-1.5,-1.5) node(t2){} circle[radius=3pt] -- (-2.5,-1.5) node(t1){} circle[radius=3pt] -- (b.center);
	\filldraw (b.center) -- (1.5,-1.5) node(t3){} circle[radius=3pt] -- (2.5,-1.5) node(t4){} circle[radius=3pt] --  (b.center);
	
	
	\path[-] (b.center) edge  [in=-240,out=-300,loop] node {} ();
	\path[-] (l1.center) edge  [in=-240,out=-300,loop] node {} ();
	\path[-] (l2.center) edge  [in=-240,out=-300,loop] node {} ();
	\path[-] (t1.center) edge  [in=240,out=300,loop] node {} ();
	\path[-] (t2.center) edge  [in=240,out=300,loop] node {} ();
	\path[-] (t3.center) edge  [in=240,out=300,loop] node {} ();
	\path[-] (t4.center) edge  [in=240,out=300,loop] node {} ();
	
	\node at (mid.center) {$U$};
	\node at ($(l1)+(0cm,1cm)$) {$x_1$};
	\node at ($(l2)+(0cm,1cm)$) {$x_k$};
	\node at ($(b)+(0cm,-.5cm)$) {$b$};
	\node at ($(t1)+(0cm,-1cm)$) {$x_{k+1}$};
	\node at ($(t2)+(0cm,-1cm)$) {$y_{k+1}$};
	\node at ($(t3)+(0cm,-1cm)$) {$x_{q}$};
	\node at ($(t4)+(0cm,-1cm)$) {$y_{q}$};
	
	\coordinate (pmid2) at (2,1.5);
	\coordinate (cdot1) at ($(pmid2)+(-.3cm,0cm)$);
	\coordinate (cdot2) at ($(pmid2)+(0cm,0cm)$);
	\coordinate (cdot3) at ($(pmid2)+(.3cm,0cm)$);
	\fill (cdot1) circle[radius=1.5pt];
	\fill (cdot2) circle[radius=1.5pt];
	\fill (cdot3) circle[radius=1.5pt];
	
	\coordinate (pmid3) at (-90:1.6cm);
	\coordinate (cdot1) at ($(pmid3)+(-.3cm,0cm)$);
	\coordinate (cdot2) at ($(pmid3)+(0cm,0cm)$);
	\coordinate (cdot3) at ($(pmid3)+(.3cm,0cm)$);
	\fill (cdot1) circle[radius=1.5pt];
	\fill (cdot2) circle[radius=1.5pt];
	\fill (cdot3) circle[radius=1.5pt];
	
	\end{tikzpicture}
	\caption{The graph $H_b$.}
	\label{fig:Hb}
\end{figure}

We will give a reduction from $\TCut{q}$ to $\Ret{H_b}$. By the choice of $b$ we have $q\ge 3$. This gives the desired reduction since
\[
\sat \leap \TCut{q} \leap \Ret{H_b} \leap \Ret{H},
\]
where the first reduction is from Lemma~\ref{lem:TCutSAT} and the last reduction is from Observation~\ref{obs:PinNeighbourhood}.

Let $I=(G, \tau_1, \dots, \tau_q, K)$ be an instance of $\TCut{q}$ and let $\eps\in(0,1)$ be the desired precision bound. Let $n=\abs{V(G)}$ and $m=\abs{E(G)}$. From the instance $I$ we construct an instance $(J,\boldS)$ of $\Ret{H_b}$. We will need some parameters whose relevance will become clear later in the proof. Let $s=n^5$ and $t=n^2$.

The intuition behind the gadgets that will be used in this proof is the following. For every vertex $v$ in $G$ we introduce a huge clique $C_v$. The image of such a clique under a homomorphism to $H_b$ tends to be a reflexive clique, i.e.~tends to be of the form $\NHb(x_i)$. There are $q$ such neighbourhoods. These will correspond to the $q$ different states that a vertex $v\in V(G)$ can be in. We will have to add some attachments to the clique $C_v$ to balance out the fact that $\NHb(x_i)$ is a clique on $2$ vertices if $i\le k$, whereas it is a clique on $3$ vertices if $i>k$. For each edge $\{u,v\}\in E(G)$ we introduce a gadget that favours the case where $u$ and $v$ have identical states (i.e.~the corresponding cliques have the same image under homomorphisms to $H_b$). 

Here are the details.
First we define the graph $J$. We introduce $q$ distinct vertices $p_1, \dots, p_q$ which will serve as ``pins'' to the vertices $x_1, \dots, x_q$. In the first part of the construction we will only use the vertices $p_1, \dots, p_k$. For every vertex $v\in V(G)$ we introduce a graph $J_v$ (the ``vertex gadget'') as follows. Let $C_v$ be a clique on $s$ vertices. For each vertex $w$ in this clique we introduce $k$ distinct vertices $\{w_1, \dots, w_k\}$. Then $J_v$ is the graph with vertices 
\[
V(J_v)= \{p_1, \dots, p_k\} \cup V(C_v) \cup \bigcup_{w\in V(C_v)} \{w_1, \dots, w_k\} 
\]
and edges 
\[
E(J_v)=E(C_v) \cup \bigcup_{w\in V(C_v)}\bigcup_{i\in [k]} \{\{w,w_i\},\{w_i,p_i\}\}.
\]
The graph $J_v$ is depicted in Figure~\ref{fig:WR3VertexGadget}. Note that the vertices $p_1, \dots p_k$ are identical over all vertex gadgets whereas the remaining vertices are distinct for each $v$.

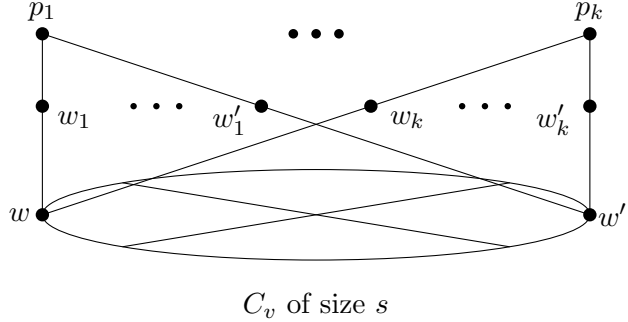
\begin{figure}[tbh]\centering
{\def\scaleFactor{.6}\centering
\begin{tikzpicture}[scale=\scaleFactor, every loop/.style={min distance=10mm,looseness=10}]

\def\xx{6} \def\yy{1} 

\coordinate (Cc) at (0,0);
\coordinate (Cleft) at ($(Cc) - (\xx, 0)$);
\coordinate (Cright) at ($(Cc) + (\xx, 0)$);
\coordinate (Cul) at ($(Cc) + (-\xx, \yy)$);
\coordinate (Cll) at ($(Cc) + (-\xx, -\yy)$);
\coordinate (Cur) at ($(Cc) + (\xx, \yy)$);
\coordinate (Clr) at ($(Cc) + (\xx, -\yy)$);

\path[name path=ellipse, fill=white](0,0)circle[x radius=\xx,y radius=\yy];
\path[name path=construct1] (Cul) -- (Clr);
\path[name path=construct2] (Cll) -- (Cur);
\path [name intersections={of = ellipse and construct1}];
\coordinate (X1) at (intersection-1);
\coordinate (Y1) at (intersection-2);
\path [name intersections={of = ellipse and construct2}];
\coordinate (X2) at (intersection-1);
\coordinate (Y2) at (intersection-2);

\coordinate (pleft) at ($(Cc) + (-\xx, 4)$);
\coordinate (pright) at ($(Cc) + (\xx, 4)$);
\coordinate (pmid) at ($(Cc) + (0, 4)$); 

\draw[fill=white](0,0)circle[x radius=\xx,y radius=\yy];
\draw (X1) -- (Y1);
\draw (X2) -- (Y2);

\draw (Cleft) -- (pleft) node(w1)[circle,fill=black,inner sep=0pt,minimum size=5pt,pos=.6]{} node[circle,fill=black,inner sep=0pt,minimum size=5pt]{} -- (Cright) node(z1)[circle,fill=black,inner sep=0pt,minimum size=5pt,pos=.4]{} ;

\draw (Cright)node[circle,fill=black,inner sep=0pt,minimum size=5pt]{} -- (pright) node(z2)[circle,fill=black,inner sep=0pt,minimum size=5pt,pos=.6]{} node[circle,fill=black,inner sep=0pt,minimum size=5pt]{} -- (Cleft) node(w2)[circle,fill=black,inner sep=0pt,minimum size=5pt,pos=.4]{} node[circle,fill=black,inner sep=0pt,minimum size=5pt]{};

\node at ($(Cc)+(0cm,-2cm)$) {$C_v$ of size $s$};
\node at ($(pleft)+(0cm,.5cm)$) {$p_1$};
\node at ($(pright)+(0cm,.5cm)$) {$p_k$};
\node at ($(w1)+(.7cm,-.3cm)$) {$w_1$};
\node at ($(w2)+(.8cm,-.3cm)$) {$w_k$};
\node at ($(z1)+(-.7cm,-.3cm)$) {$w'_1$};
\node at ($(z2)+(-.8cm,-.3cm)$) {$w'_k$};
\node at ($(Cleft)+(-.5cm,0cm)$) {$w$};
\node at ($(Cright)+(.5cm,0cm)$) {$w'$};

\coordinate (cdot1) at ($(pmid)+(-.5cm,0cm)$);
\coordinate (cdot2) at ($(pmid)+(0cm,0cm)$);
\coordinate (cdot3) at ($(pmid)+(.5cm,0cm)$);
\fill (cdot1) circle[radius=3pt];
\fill (cdot2) circle[radius=3pt];
\fill (cdot3) circle[radius=3pt];

\coordinate (cdot1) at ($(w1)+(2cm,0cm)$);
\coordinate (cdot2) at ($(w1)+(2.5cm,0cm)$);
\coordinate (cdot3) at ($(w1)+(3cm,0cm)$);
\fill (cdot1) circle[radius=2pt];
\fill (cdot2) circle[radius=2pt];
\fill (cdot3) circle[radius=2pt];

\coordinate (cdot1) at ($(w2)+(2cm,0cm)$);
\coordinate (cdot2) at ($(w2)+(2.5cm,0cm)$);
\coordinate (cdot3) at ($(w2)+(3cm,0cm)$);
\fill (cdot1) circle[radius=2pt];
\fill (cdot2) circle[radius=2pt];
\fill (cdot3) circle[radius=2pt];
\end{tikzpicture}
}
\caption{The graph $J_v$ for a vertex $v$.}
\label{fig:WR3VertexGadget}
\end{figure}

For every edge $e=\{u,v\}\in E(G)$ we introduce a graph $J_e$ together with a set of edges $E_e$ (the ``edge gadget''). The graph $J_e$ is defined in precisely the same way as $J_v$ but uses the parameter $t$ instead of $s$. We denote the corresponding clique by $C_e$. Further, we set $E_e=\ucp{\left(V(C_u)\cup V(C_v)\right)}{V(C_e)}$. The edge gadget is depicted in Figure~\ref{fig:WR3EdgeGadget}.

\begin{figure}[tbh]\centering
{\def\scaleFactor{.6}\centering
\begin{tikzpicture}[scale=\scaleFactor, every loop/.style={min distance=10mm,looseness=10}]

\def\xx{1} \def\yy{2} 

\coordinate (Cc) at (0,0);
\coordinate (Ctop) at ($(Cc) + (0, \yy)$);
\coordinate (Cbot) at ($(Cc) - (0, \yy)$);
\coordinate (Cul) at ($(Cc) + (-\xx, \yy)$);
\coordinate (Cll) at ($(Cc) + (-\xx, -\yy)$);
\coordinate (Cur) at ($(Cc) + (\xx, \yy)$);
\coordinate (Clr) at ($(Cc) + (\xx, -\yy)$);

\def\x{1.5} \def\y{3} 

\def\ax{-5.5} \def\ay{0}  

\coordinate (Ac) at (\ax,\ay);
\coordinate (Atop) at ($(Ac) + (0, \y)$);
\coordinate (Abot) at ($(Ac) + (0, -\y)$);
\coordinate (Aul) at ($(Ac) + (-\x, \y)$);
\coordinate (All) at ($(Ac) + (-\x, -\y)$);
\coordinate (Aur) at ($(Ac) + (\x, \y)$);
\coordinate (Alr) at ($(Ac) + (\x, -\y)$);

\draw (Atop) -- (Ctop);
\draw (Abot) -- (Cbot);
\draw (Atop) -- (Cbot);
\draw (Abot) -- (Ctop);

\draw[name path=ellipsel, fill=white](Ac)circle[x radius=\x,y radius=\y];
\path[name path=construct1] (Aul) -- (Alr);
\path[name path=construct2] (All) -- (Aur);
\path [name intersections={of = ellipsel and construct1}];
\coordinate (X) at (intersection-1);
\coordinate (Y) at (intersection-2);
\draw (X) -- (Y);
\path [name intersections={of = ellipsel and construct2}];
\coordinate (X) at (intersection-1);
\coordinate (Y) at (intersection-2);
\draw (X) -- (Y);

\node at ($(Ac)+(-2.5cm,0cm)$) {$C_u$}; 
\node at ($(Ac)+(0cm,3.5cm)$) {$s$};

\def\bx{5.5} \def\by{0}  

\coordinate (Bc) at (\bx,\by);
\coordinate (Btop) at ($(Bc) + (0, \y)$);
\coordinate (Bbot) at ($(Bc) + (0, -\y)$);
\coordinate (Bul) at ($(Bc) + (-\x, \y)$);
\coordinate (Bll) at ($(Bc) + (-\x, -\y)$);
\coordinate (Bur) at ($(Bc) + (\x, \y)$);
\coordinate (Blr) at ($(Bc) + (\x, -\y)$);

\draw (Btop) -- (Ctop);
\draw (Bbot) -- (Cbot);
\draw (Btop) -- (Cbot);
\draw (Bbot) -- (Ctop);

\draw[name path=ellipser, fill=white](Bc)circle[x radius=\x,y radius=\y];
\path[name path=construct1] (Bul) -- (Blr);
\path[name path=construct2] (Bll) -- (Bur);
\path [name intersections={of = ellipser and construct1}];
\coordinate (X) at (intersection-1);
\coordinate (Y) at (intersection-2);
\draw (X) -- (Y);
\path [name intersections={of = ellipser and construct2}];
\coordinate (X) at (intersection-1);
\coordinate (Y) at (intersection-2);
\draw (X) -- (Y);

\node at ($(Bc)+(2.5cm,0cm)$) {$C_v$};
\node at ($(Bc)+(0cm,3.5cm)$) {$s$};

\draw[name path=ellipse, fill=white](0,0)circle[x radius=\xx,y radius=\yy];
\path[name path=construct1] (Cul) -- (Clr);
\path[name path=construct2] (Cll) -- (Cur);
\path [name intersections={of = ellipse and construct1}];
\coordinate (X) at (intersection-1);
\coordinate (Y) at (intersection-2);
\draw (X) -- (Y);
\path [name intersections={of = ellipse and construct2}];
\coordinate (X) at (intersection-1);
\coordinate (Y) at (intersection-2);
\draw (X) -- (Y);

\node at ($(Cc)+(0cm,-3cm)$) {$C_e$};
\node at ($(Cc)+(0cm,2.5cm)$) {$t$};
\end{tikzpicture}
}
\caption{The edge gadget for an edge $e=\{u,v\}$. The edges to the vertices $\{w_1, \dots, w_k \mid w\in C_u \cup C_e \cup C_v\}$ and $\{p_1, \dots, p_k\}$ are omitted.}
\label{fig:WR3EdgeGadget}
\end{figure}
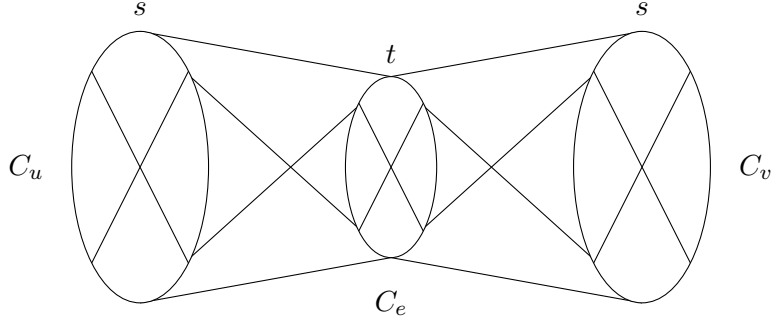

Finally, $J$ is the graph with vertices
\[
V(J)=\{p_{k+1}, \dots, p_q\} \cup \bigcup_{v\in V(G)} V(J_v) \cup \bigcup_{e\in E(G)} V(J_e) 
\]
and edges
\[
E(J)=\bigcup_{i\in [q]} \bigl(\ucp{V(C_{\tau_i})}{\{p_i\}}\bigr) \cup \bigcup_{v\in V(G)} E(J_v) \cup \bigcup_{e\in E(G)} \bigl(E(J_e)\cup E_e\bigr).
\]
Note that the first set in the union is a set of edges for each terminal $\tau_i$. The purpose of these edges will be to ensure that the corresponding graph $J_{\tau_i}$ is in the right ``state'' (the one corresponding to $i$). Here we now use all of the $p_i$, not just the first $k$ as we did in the construction of $J_v$ and $J_e$.

Next we define the lists $\boldS=\{S_v\subseteq V(H_b) \mid v\in V(J)\}$. We set
\[
S_v=
\begin{cases}
\{x_i\}, &\text{if } v=p_i,\ i\in[q]\\
V(H_b), &\text{otherwise.}
\end{cases}
\]

For a homomorphism $h\in \calH((J,\boldS),H_b)$ and a vertex $v\in V(G)$ we say that the image $h(V(C_v))$ is the \emph{state} of $v$ (under $h$). 

A \emph{pinned configuration} is a tuple $(z, z_1, \dots, z_k)$ of vertices of $H_b$ such that, for each $i\in [k]$, $\{z,z_i\}$ and $\{z_i,x_i\}$ are edges of $H_b$. Note that the vertices $(w, w_1, \dots, w_k)$ of $C_v$ (see Figure~\ref{fig:WR3VertexGadget}) have to map to a pinned configuration under a homomorphism from $(J,\boldS)$ to $H_b$. 
For $z\in V(H_b)$ let $f(z)$ be the number of pinned configurations $(z, z_1, \dots, z_k)$. We have
\begin{align}
f(b)&= 2^k &&\text{(All $z_j$ can be either $x_j$ or $b$.)}\label{equ:b-configs}\\
f(x_i)&= 2\ (\forall i\in [k]) &&\text{($z_i$ can be either $x_i$ or $b$, all other $z_j$ have to be $b$.)}\label{equ:x-configs}\\
f(x_i)=f(y_i)&= 1\ (\forall i\in\{k+1, \dots, q\}) &&\text{(All $z_j$ have to be $b$.)}\label{equ:y-configs}\\
f(u)&=1\ (\forall u\in U) &&\text{(All $z_j$ have to be $b$.)}\label{equ:u-configs}
\end{align}
We say that a vertex $v\in V(G)$ is \emph{full} (under $h$) if the following conditions are met:
\begin{itemize}
	\item There exists $i\in [q]$ such that $h(V(C_v))=\NHb(x_i)$.
	\item For every element $z\in h(V(C_v))$ and every pinned configuration $(z, z_1, \dots, z_k)$ there exists a vertex $w$ in the clique $C_v$ such that $h(w, w_1, \dots, w_k)= (z, z_1, \dots, z_k)$ (elementwise).
\end{itemize}
We call an edge $e=\{u,v\}\in E(G)$ \emph{monochromatic} (under $h$) if $u$ and $v$ have the same state. Otherwise, we say that $e$ is \emph{dichromatic}. 
We say that the homomorphism $h$ is \emph{full} if every vertex $v\in V(G)$ is full under $h$. We say that a full homomorphism $h$ is \emph{$K$-small} if there are at most $K$ dichromatic edges under $h$, otherwise we say that it is \emph{$K$-large}.

Let $Z^*$ be the number of full homomorphisms that are $K$-small. Further, let $Z_1$ be the number of full homomorphisms that are $K$-large and let $Z_2$ be the number of non-full homomorphisms. Then
\[
\hom{(J, \boldS)}{H_b} = Z^* + Z_1 + Z_2.
\]

Let $T$ be the sought-for number of size-$K$ multiterminal cuts of the instance $I=(G, \tau_1, \dots, \tau_q, K)$. We will now investigate the way in which the number $Z^*$ relates to $T$. Recall the definitions about separating functions from Definition~\ref{def:Phi}. In particular, we will use the sets $\Phi(I)$, $\Phi^*(I)$ and $\Ecut(\phi)$. To shorten notation within the scope of this proof, we write $\Phi$ when we mean $\Phi(I)$ and $\Phi^*$ when we mean $\Phi^*(I)$. From Observation~\ref{obs:Phi} we know that $T=\abs{\Phi^*}$.

For a function $\phi\in \Phi$ we say that a homomorphism $h\in  \calH((J,\boldS),H_b)$ \emph{agrees} with $\phi$ if, for each vertex $v$ of $G$, the state of $v$ under $h$ is $h(V(C_v))=\NHb(x_{\phi(v)})$. Note that, by the construction of $J$, under a full homomorphism a terminal $\tau_i$ has state $\NHb(x_{i})$. Therefore, each full homomorphism $h$ agrees with exactly one $\phi \in \Phi$.

If $h$ agrees with $\phi$, then it follows that $\Ecut(\phi)$ is exactly the set of dichromatic edges under $h$. Hence, each $K$-small full homomorphism agrees with exactly one function $\phi \in \Phi^*$ and each $K$-large full homomorphism agrees with exactly one function $\phi \in \Phi\setminus\Phi^*$. Let $Z_\phi$ be the number of full homomorphisms that agree with $\phi\in \Phi$. Then
\begin{equation}\label{equ:WR3_Z*Z1}
Z^*=\sum_{\phi\in\Phi^*} Z_\phi \quad \text{ and } \quad Z_1=\sum_{\phi\in\Phi\setminus\Phi^*} Z_\phi.
\end{equation}

Let $\phi \in \Phi$. We are interested in the number $Z_\phi$. What are the possible full homomorphisms $h$ that agree with $\phi$? 

\begin{description}
	\item[Observation A] Let $v\in V(G)$. We consider possible images of the vertices of $J_v$. For $h$ to agree with $\phi$, the state of $v$ is fixed to be $\NHb(x_{\phi(v)})$, where this set can be either of the form $\{b,x_i\}$ or of the form $\{b,x_i, y_i\}$. From Equations~\eqref{equ:b-configs} and~\eqref{equ:x-configs} it follows that there are a total of $2^k+2$ pinned configurations $(z, z_1, \dots, z_k)$ with $z\in\{b, x_i\}$. Similarly, from Equations~\eqref{equ:b-configs} and~\eqref{equ:y-configs} it follows that there are a total of $2^k+1+1$ pinned configurations with $z\in\{b, x_i, y_i\}$. As $h$ is full, each possible pinned configuration has to be used at least once by the $s$ vertices in $C_v$. As a consequence each vertex $v\in V(G)$ contributes a factor of $\stirling{s}{2^k+2}$ to $Z_\phi$. 
	\item[Observation B] Let $e=\{u,v\}$. What are the possible images of the vertices of $J_e$? We make a case distinction depending on $e$.
	\begin{itemize}
		\item Let $e=\{u,v\}\in \Ecut(\phi)$. Then, as $h$ is full, $h(V(C_u))$ and $h(V(C_v))$ are different states from the set $\{\NHb(x_i)\mid i\in [q]\}$. By the definition of $J$ we have that $h(V(C_e))\subseteq h(V(C_u))\cap h(V(C_v))$. It follows that $h(V(C_e))=\{b\}$. There are $2^k$ pinned configurations with $z=b$. Each of the $t$ vertices of $C_e$ can have any of these $2^k$ pinned configurations. Therefore, $e$ contributes a factor of $2^{kt}$ to $Z_\phi$.
		\item Let $e=\{u,v\}\notin \Ecut(\phi)$. Then, as $h$ is full, $h(V(C_u))=h(V(C_v))=\NHb(x_i)$ for some $i\in[q]$. Then $h(C_e)\subseteq \NHb(x_i)$, where $\NHb(x_i)$ is of the form $\{b,x_i\}$ or of the form $\{b,x_i,y_i\}$. As before there are $2^k + 2$ corresponding pinned configurations. Therefore, $e$ contributes a factor of $(2^k+2)^t$ to $Z_\phi$. 
	\end{itemize}
\end{description}
Summarising, we obtain
\begin{equation}\label{equ:WR3_Zphi}
	Z_\phi = \stirling{s}{2^k+2}^n (2^k)^{t\abs{\Ecut(\phi)}}(2^k+2)^{t(m- \abs{\Ecut(\phi)})}.
\end{equation}
For each $\phi\in\Phi^*$ we have $\abs{\Ecut(\phi)}=K$. Then, using the fact that $\abs{\Phi^*}=T$, we obtain
\begin{align*}
	Z^* &= \sum_{\phi \in \Phi^*} Z_\phi\\
	&= \sum_{\phi \in \Phi^*} \stirling{s}{2^k+2}^n (2^k)^{tK}(2^k+2)^{t(m- K)}\\
	&= \stirling{s}{2^k+2}^n (2^k)^{tK}(2^k+2)^{t(m- K)} \cdot T.
\end{align*}
To shorten the notation let 
\[
	L=\stirling{s}{2^k+2}^n (2^k)^{tK}(2^k+2)^{t(m- K)}.
\] 
We want to approximately compute the value $T$, where we have shown that $T=Z^*/L$. Assume for now that we have 
\[
	\hom{(J, \boldS)}{H_b}/L \in [Z^*/L, Z^*/L+1/4] = [T, T+1/4].
\]
Then consider the algorithm which makes a $\Ret{H_b}$ oracle call with input $((J, \boldS),\eps/21)$ to obtain a value $Q$ and returns $\floor{Q/L}$. This algorithm approximates $T$ with the desired error bound $\eps$ as is shown in~\cite[Proof of Theorem 3]{DGGJApprox}.
It remains to show the following claim.

\medskip
\noindent{\bf\boldmath Claim: $\hom{(J, \boldS)}{H_b}/L \in [Z^*/L, Z^*/L+1/4]$.}
\medskip

\noindent {\it Proof of the claim:}\quad
Recall that $\hom{(J, \boldS)}{H_b}=Z^* + Z_1+Z_2$. Clearly, we have $\hom{(J, \boldS)}{H_b}/L \ge Z^*/L$.
We will show $Z_1/L\le 1/8$ and $Z_2/L\le 1/8$ to prove $\hom{(J, \boldS)}{H_b}/L\le Z^*/L+1/4$.

Recall that $Z_1$ is the number of $K$-large full homomorphisms. Using~\eqref{equ:WR3_Z*Z1} and~\eqref{equ:WR3_Zphi} and the fact that for each $\phi\in \Phi\setminus\Phi^*$ we have $\abs{\Ecut(\phi)}\ge K+1$ we obtain
\begin{align*}
Z_1 &\le  \sum_{\phi\in \Phi\setminus\Phi^*}\stirling{s}{2^k+2}^n (2^k)^{t(K+1)}(2^k+2)^{t(m-K-1)}\\
&\le  q^n\stirling{s}{2^k+2}^n (2^k)^{t(K+1)}(2^k+2)^{t(m-K-1)}
\end{align*}
where the last inequality follows from the fact that there are $q^n$ functions in $\Phi$. Then

\begin{equation*}
\frac{Z_1}{L}
\le \left(\frac{2^{k}}{(2^k+2)}\right)^tq^n \le 1/8,
\end{equation*}
where the last inequality holds for sufficiently large $n$ by our choice of $t=n^2$ and the fact that $2^k/(2^k+2)<1$.

Recall that $Z_2$ is the number of homomorphisms that are not full. How many non-full homomorphisms $h$ are there? In general, there are at most $2^{\abs{V(H_b)}}$ possible states $h(V(C_v))$ for any vertex $v\in V(G)$. By the same arguments as given in Observation A, each full vertex under $h$ contributes a factor of $\stirling{s}{2^k+2}$ to $Z_2$. Since $s=n^5$ the requirements of Corollary~\ref{cor:DGGJ18} are met for sufficiently large $n$ and we obtain 
\begin{equation}\label{equ:WR3_fullv}
(2^k+2)^s/2\le \stirling{s}{2^k+2}.
\end{equation}
We will use this bound shortly. If $h$ is not full, there has to exist at least one vertex $v\in V(G)$ which is not full under $h$, and consequently there are at most $n-1$ full vertices under $h$.

Now assume that $v\in V(G)$ is not full under $h$. Then either $v$ has state $h(V(C_v))=\NHb(x_i)$ for some $i\in [q]$ but not all of the $2^k+2$ possible pinned configurations are used, or $h(V(C_v))\notin\{\NHb(x_i)\mid i\in [q]\}$ (which means that either $h(V(C_v))\subsetneq\NHb(x_i)$ or $h(V(C_v))=\{b,u\}$ for some $u\in U$, by the fact that $C_v$ is a large clique). In both cases the $s$ vertices in $C_v$ can each use at most $2^k+1$ different pinned configurations. Hence, each non-full vertex contributes a factor of at most $(2^k+1)^s$ to $Z_2$. In particular this factor is smaller (for sufficiently large $n$) than the factor contributed by full vertices (see~\eqref{equ:WR3_fullv}). Finally, for each edge $e$ there are at most $\abs{V(H_b)}^{(k+1)t}$ mappings from the $(k+1)\cdot t$  vertices in $V(J_e)\setminus \{p_1, \dots, p_k\}$ to $V(H_b)$. Therefore,
\[
Z_2 \le  2^{\abs{V(H_b)}n} \cdot \stirling{s}{2^k+2}^{n-1}\cdot (2^k+1)^s \cdot \abs{V(H_b)}^{(k+1)tm}.
\]
Recall that $L=\stirling{s}{2^k+2}^n (2^k)^{tK}(2^k+2)^{t(m- K)}\ge \stirling{s}{2^k+2}^n$. Then
\begin{equation*}
\frac{Z_2}{L}
\le \frac{2^{\abs{V(H_b)}n} \cdot (2^k+1)^s\cdot \abs{V(H_b)}^{(k+1)tm}}{\stirling{s}{2^k+2}}
\le \left(\frac{2^k+1}{2^k+2}\right)^s\cdot2\cdot 2^{\abs{V(H_b)}n}\cdot \abs{V(H_b)}^{(k+1)n^4}\le 1/8,
\end{equation*}
where the second inequality follows from~\eqref{equ:WR3_fullv} and the last inequality holds for sufficiently large $n$ by our choice of $s=n^5$. This proves the claim and completes the proof.
{\it (End of the proof of the claim.)}
\end{proof}

\subsection{Square-Free Graphs with an Induced Net} \label{sec:inducedNet}
The goal of this section is to prove $\sat$-hardness for square-free graphs with an induced net (see Figure~\ref{fig:net}). Note that the subgraphs of the net that are induced by a distance-$1$ neighbourhood of some vertex $u$ of the net are of two forms. Either the corresponding subgraph is a looped edge (if $u\notin\{w_i\mid i\in [3]\}$) or it is isomorphic to a looped triangle where one vertex in the triangle has a single additional looped neighbour (if $u\in\{w_i\mid i\in [3]\}$). Approximately counting retractions to either of these two graphs is $\bis$-easy (see Theorem~\ref{thm:bisEasyRet}). Therefore we cannot use these subgraphs in our hardness proof, so we need to work harder.

\begin{figure}[ht]
	\centering
	\begin{tikzpicture}[scale=1, baseline=0.36cm, every loop/.style={min distance=10mm,looseness=10}]

			\filldraw (0,0) node(d){} circle[radius=3pt] --++ (30:1cm) node(a){} circle[radius=3pt] --++ (60:1cm) node(b){} circle[radius=3pt] --++ (-60:1cm) node(c){} circle[radius=3pt]--++ (-30:1cm) node(f){} circle[radius=3pt];
			\filldraw (b.center) --++ (90:1cm) node(e){} circle[radius=3pt];
			\draw (a.center) -- (c.center);	
			
			\node at ($(a)+(.2cm,-.3cm)$) {$w_1$};	
			\node at ($(b)+(-.4cm,0cm)$) {$w_2$};	
			\node at ($(c)+(.2cm,.2cm)$) {$w_3$};
			
			\path[-] (a.center) edge  [in=155,out=85,loop] node {} ();
			\path[-] (d.center) edge  [in=155,out=85,loop] node {} ();
			\path[-] (b.center) edge  [in=30,out=-30,loop] node {} ();
			\path[-] (e.center) edge  [in=30,out=-30,loop] node {} ();
			\path[-] (c.center) edge  [in=-155,out=-85,loop] node {} ();
			\path[-] (f.center) edge  [in=-155,out=-85,loop] node {} ();

	\end{tikzpicture}
	\caption{The \emph{net}.}
\label{fig:net}
\end{figure}

In our proof (Lemma~\ref{lem:inducedNet}) we use the same general approach that we used to prove Lemma~\ref{lem:inducedWR3General}. To make the approach work we have to find new gadgets tailored to the net. In one part of the reduction we will need to approximate real values by integers. To achieve this we use Dirichlet's approximation lemma, which has been used frequently in this line of research (see for instance~\cite{GGJBIS}).

\begin{lem}[{see e.g.~\cite[p. 34]{Schmidt1991}}] \label{lem:Dirichlet}
Let $\lambda_1, \dots, \lambda_d > 0$ be real numbers and $N$ be a natural number. Then there exist positive integers $t_1, \dots, t_d, r$ with $r \le N$ such that $\abs{r\lambda_i - t_i}\le 1/N^{1/d}$ for every $i\in [d].$
\end{lem}

\begin{lem}\label{lem:inducedNet}
Let $H$ be a square-free graph that has the net (as displayed in Figure~\ref{fig:net}) as an induced subgraph.
Then $\sat\leap\Ret{H}$.
\end{lem}
\begin{proof}
Let $H$ be a square-free graph with an induced net that is labelled as in Figure~\ref{fig:net}. Note that each of the vertices $w_1$, $w_2$ and $w_3$ might have additional neighbours in $H$. However, they cannot have further common neighbours as $H$ is square-free.
We use a reduction from $\TCut{3}$ which is $\sat$-hard under AP-reductions by Lemma~\ref{lem:TCutSAT}. Let $I=(G, \tau_1, \tau_2, \tau_3, K)$ be an instance of $\TCut{3}$ and $\eps\in(0,1)$ the desired precision bound. Let $n=\abs{V(G)}$ and $m=\abs{E(G)}$. We will construct an instance $(J, \boldS)$ of $\Ret{H}$. 

To construct this instance we will use a number of parameters which we introduce at this point. Let $s=n^2$. For $i\in[3]$ let $s_i=s\cdot\log_{\abs{\NH(w_i)}/3}2$.
For $i\in [3]$, the value $s_i$ is chosen such that 
\begin{equation}\label{equ:srsbsg}
\left(\frac{\abs{\NH(w_i)}}{3}\right)^{s_i}=
2^s.
\end{equation}
It will become clear later in the proof why this is useful. (The important part is that the values $\left(\frac{\abs{\NH(w_i)}}{3}\right)^{s_i}$, for $i\in [3]$, are identical and that the base of the exponent on the right-hand side is greater than $1$.) We will now determine integers that approximate $s_1$, $s_2$ and $s_3$ using Dirichlet's approximation lemma. Let $\delta=\eps/2$ and $\delta'=\log_{\abs{V(H)}} e^\delta$ (the choices will become clear later in the proof). By Lemma~\ref{lem:Dirichlet} we obtain integers $r$, $t_1$, $t_2$ and $t_3$ of value at most $(m/\delta')^3\in \text{poly}(n,\eps^{-1})$ such that $\abs{rs_i - t_i}\le \delta'/m$ for $i\in[3]$.

We go on to define the graph $J$. Intuitively, for $i\in[3]$, the terminal $\tau_i$ will serve as a ``pin'' to the vertex $w_i$. For each vertex $v\in V(G)$ we introduce a graph $J_v$ which is simply a star with center $v$ and leaves $\{\tau_1, \tau_2, \tau_3\}$. Formally, the vertices of $J_v$ are $V(J_v)=\{v, \tau_1, \tau_2, \tau_3\}$. and the edges are $E(J_v)=\ucp{\{v\}}{\{\tau_1, \tau_2, \tau_3\}}$. For each edge $e=\{u,v\}\in E(G)$ we introduce a graph $J_e$ which is defined as follows. For $i\in[3]$, let $I_e^i$ be disjoint independent sets of size $t_i$. Let $t=\sum_{i\in [3]} t_i$ and let $I_e=\bigcup_{i\in [3]}I_e^i$ ($I_e$ is an independent set of size $t$). Then $J_e$ is depicted in Figure~\ref{fig:NetEdgeGadget} and formally defined as the graph with vertices
\[
V(J_e)=I_e \cup \{u,v,\tau_1,\tau_2,\tau_3\}
\]
and edges
\[
E(J_e)=\left(\ucp{\{u,v\}}{I_e}\right) \cup \bigcup_{i=1}^3\left(\ucp{I_e^i}{\{\tau_i\}}\right).
\]

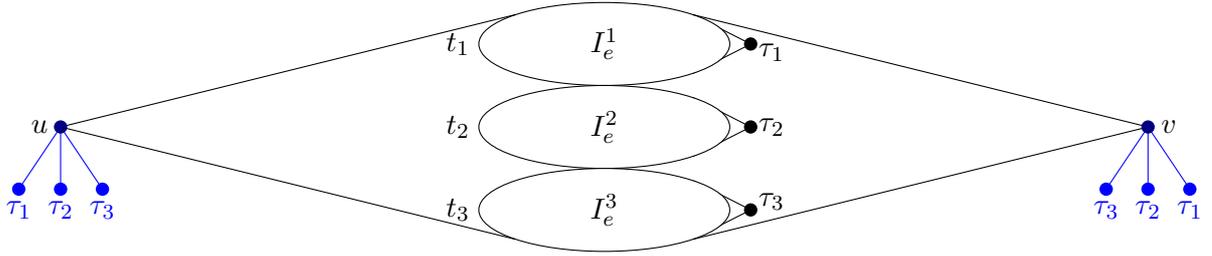
\begin{figure}[tbh]\centering
{\def\scaleFactor{.55}\centering
\begin{tikzpicture}[scale=\scaleFactor]
	
	\coordinate (u) at (-13,0);
	\coordinate (u1) at (-14,-1.5);
	\coordinate (u2) at (-13,-1.5);
	\coordinate (u3) at (-12,-1.5);

	\coordinate (v) at (13,0);
	\coordinate (v1) at (14,-1.5);
	\coordinate (v2) at (13,-1.5);
	\coordinate (v3) at (12,-1.5);

\coordinate (tau1) at (3.5,2);
\coordinate (tau2) at (3.5,0);
\coordinate (tau3) at (3.5,-2);

\def\xx{3} \def\yy{1} 

\coordinate (Cc1) at (0,2);
\coordinate (Ctop1) at ($(Cc1) + (0, \yy)$);
\coordinate (Cbot1) at ($(Cc1) - (0, \yy)$);
\coordinate (Cul1) at ($(Cc1) + (-\xx, \yy)$);
\coordinate (Cll1) at ($(Cc1) + (-\xx, -\yy)$);
\coordinate (Cur1) at ($(Cc1) + (\xx, \yy)$);
\coordinate (Clr1) at ($(Cc1) + (\xx, -\yy)$);

\path[name path=ellipse1](Cc1)circle[x radius=\xx,y radius=\yy];
\path[name path=construct1] (Cul1) -- (Clr1);
\path[name path=construct2] (Cll1) -- (Cur1);
\path [name intersections={of = ellipse1 and construct1}];
\coordinate (X11) at (intersection-1);
\coordinate (Y11) at (intersection-2);
\path [name intersections={of = ellipse1 and construct2}];
\coordinate (X12) at (intersection-1);
\coordinate (Y12) at (intersection-2);
\draw (Y11) -- (tau1) node[circle,fill=black,inner sep=0pt,minimum size=5pt]{} -- (X12);

\coordinate (Cc2) at (0,0);
\coordinate (Ctop2) at ($(Cc2) + (0, \yy)$);
\coordinate (Cbot2) at ($(Cc2) - (0, \yy)$);
\coordinate (Cul2) at ($(Cc2) + (-\xx, \yy)$);
\coordinate (Cll2) at ($(Cc2) + (-\xx, -\yy)$);
\coordinate (Cur2) at ($(Cc2) + (\xx, \yy)$);
\coordinate (Clr2) at ($(Cc2) + (\xx, -\yy)$);

\path[name path=ellipse2](Cc2)circle[x radius=\xx,y radius=\yy];
\path[name path=construct1] (Cul2) -- (Clr2);
\path[name path=construct2] (Cll2) -- (Cur2);
\path [name intersections={of = ellipse2 and construct1}];
\coordinate (X21) at (intersection-1);
\coordinate (Y21) at (intersection-2);
\path [name intersections={of = ellipse2 and construct2}];
\coordinate (X22) at (intersection-1);
\coordinate (Y22) at (intersection-2);
\draw (Y21) -- (tau2) node[circle,fill=black,inner sep=0pt,minimum size=5pt]{} -- (X22);

\coordinate (Cc3) at (0,-2);
\coordinate (Ctop3) at ($(Cc3) + (0, \yy)$);
\coordinate (Cbot3) at ($(Cc3) - (0, \yy)$);
\coordinate (Cul3) at ($(Cc3) + (-\xx, \yy)$);
\coordinate (Cll3) at ($(Cc3) + (-\xx, -\yy)$);
\coordinate (Cur3) at ($(Cc3) + (\xx, \yy)$);
\coordinate (Clr3) at ($(Cc3) + (\xx, -\yy)$);

\path[name path=ellipse3](Cc3)circle[x radius=\xx,y radius=\yy];
\path[name path=construct1] (Cul3) -- (Clr3);
\path[name path=construct2] (Cll3) -- (Cur3);
\path [name intersections={of = ellipse3 and construct1}];
\coordinate (X31) at (intersection-1);
\coordinate (Y31) at (intersection-2);
\path [name intersections={of = ellipse3 and construct2}];
\coordinate (X32) at (intersection-1);
\coordinate (Y32) at (intersection-2);
\draw (Y31) -- (tau3) node[circle,fill=black,inner sep=0pt,minimum size=5pt]{} -- (X32);

\draw[fill=white](Cc1)circle[x radius=\xx,y radius=\yy];
\draw[fill=white](Cc2)circle[x radius=\xx,y radius=\yy];
\draw[fill=white](Cc3)circle[x radius=\xx,y radius=\yy];

\draw (u) -- (X11) (X12) -- (v);
\draw (u)  -- (Y32) (Y31)  -- (v);

\draw[draw=blue] (u) -- (u3) node[circle,fill=blue,inner sep=0pt,minimum size=5pt]{};
\draw[draw=blue] (u1) node[circle,fill=blue,inner sep=0pt,minimum size=5pt]{} -- (u) node[circle,fill=blue!50!black,inner sep=0pt,minimum size=5pt]{} -- (u2) node[circle,fill=blue,inner sep=0pt,minimum size=5pt]{};

\draw[draw=blue] (v) -- (v3) node[circle,fill=blue,inner sep=0pt,minimum size=5pt]{};
\draw[draw=blue] (v1) node[circle,fill=blue,inner sep=0pt,minimum size=5pt]{} -- (v) node[circle,fill=blue!50!black,inner sep=0pt,minimum size=5pt]{} -- (v2) node[circle,fill=blue,inner sep=0pt,minimum size=5pt]{};


\node at ($(u)+(-.5cm,0cm)$) {$u$}; 
\node at ($(u1)+(0cm,-.5cm)$) {\color{blue}$\tau_1$};
\node at ($(u2)+(0cm,-.5cm)$) {\color{blue}$\tau_2$};
\node at ($(u3)+(0cm,-.5cm)$) {\color{blue}$\tau_3$};

\node at ($(v)+(0.5cm,0cm)$) {$v$};
\node at ($(v1)+(0cm,-.5cm)$) {\color{blue}$\tau_1$};
\node at ($(v2)+(0cm,-.5cm)$) {\color{blue}$\tau_2$};
\node at ($(v3)+(0cm,-.5cm)$) {\color{blue}$\tau_3$};

\node at ($(Cc1)+(0cm,0cm)$) {$I^1_e$};
\node at ($(Cc2)+(0cm,0cm)$) {$I^2_e$};
\node at ($(Cc3)+(0cm,0cm)$) {$I^3_e$};
\node at ($(Cc1)+(-3.5cm,0cm)$) {$t_1$};
\node at ($(Cc2)+(-3.5cm,0cm)$) {$t_2$};
\node at ($(Cc3)+(-3.5cm,0cm)$) {$t_3$};

\node at ($(tau1)+(0.5cm,-.2cm)$) {$\tau_1$}; 
\node at ($(tau2)+(0.5cm,0cm)$) {$\tau_2$};
\node at ($(tau3)+(0.5cm,.2cm)$) {$\tau_3$};

\end{tikzpicture}
}
\caption{The graph $J_e$ for an edge $e=\{u,v\}$ is depicted in black. $J_u$ and $J_v$ are depicted in blue.
Note that $\tau_1$, $\tau_2$ and $\tau_3$ are global vertices. That is, $\tau_i$ is the same vertex in all of the different gadgets. }
\label{fig:NetEdgeGadget}
\end{figure}

Then $J$ is the graph with vertices 
\[
V(J)=\bigcup_{v\in V(G)} V(J_v) \cup \bigcup_{e\in E(G)} V(J_e)
\]
and edges 
\[
E(J)=\bigcup_{v\in V(G)} E(J_v) \cup \bigcup_{e\in E(G)} E(J_e).
\]
The corresponding set of lists $\boldS=\{S_v\subseteq V(H) \mid v\in V(J)\}$ is defined by
\[
S_v=
\begin{cases}
\{w_i\}, &\text{if }v=\tau_i,\ i\in [3]\\
V(H), &\text{otherwise.}
\end{cases}
\]

Let $h$ be a homomorphism from $(J, \boldS)$ to $H$. By the definition of $J_v$ every vertex $v\in V(G)$ is also a vertex of $J$. An edge $e=\{u,v\}\in E(G)$ is called \emph{monochromatic} under $h$ if $h(u)=h(v)$. Otherwise, it is called \emph{dichromatic} under $h$. We say that $h$ is \emph{$K$-small} if the number of dichromatic edges under $h$ is at most $K$. Otherwise, $h$ is called \emph{$K$-large}. Let $Z^*$ be the number of $K$-small homomorphisms from $(J,\boldS)$ to $H$ and let $Z_1$ be the number of $K$-large homomorphisms. Clearly,
\[
\hom{(J, \boldS)}{H} = Z^* + Z_1.
\]

Recall the definitions of separating functions from Definition~\ref{def:Phi} and, to shorten notation, define the sets $\Phi=\Phi(I)$ and $\Phi^*=\Phi^*(I)$. $T$ is the number of size-$K$ multiterminal cuts of the instance $I=(G, \tau_1, \tau_2, \tau_3, K)$. Our goal is to approximate $T$. From Observation~\ref{obs:Phi} we know that $T=\abs{\Phi^*}$.

We say that a homomorphism $h\in\calH((J,\boldS),H)$ \emph{agrees} with $\phi\in\Phi$ if, for each $v\in V(G)$, we have $h(v)=w_{\phi(v)}$. By definition of the lists $\boldS$, for each $i\in [3]$, a homomorphism $h$ from $(J, \boldS)$ to $H$ has to map $\tau_i$ to $w_i$. Furthermore, as $v$ is adjacent to all three terminals and $H$ is square-free we have $h(v)\in\{w_1,w_2,w_3\}$. 

At this point we have introduced the gadget $J_e$ and the graph $J$ and have defined what it means for a homomorphism from $J$ to $H$ to agree with a function in $\Phi$ (which in turn corresponds to a multiterminal cut). All these definitions are heavily tailored to the graph $H$. The following steps, however, are very similar to those in the proof of Lemma~\ref{lem:inducedWR3General}. What complicates this proof in comparison to that of Lemma~\ref{lem:inducedWR3General} is the fact that we need to use Dirichlet's approximation lemma to balance out the edge interactions. Here are the details.

Every homomorphism $h\in\calH((J,\boldS),H)$ agrees with exactly one function $\phi \in \Phi$. In particular, if $h$ agrees with $\phi$, then the dichromatic edges of $h$ are exactly the multiterminal cut $\Ecut(\phi)$. It follows that every $K$-small homomorphism agrees with exactly one function $\phi \in \Phi^*$ and every $K$-large homomorphism agrees with exactly one function $\phi \in \Phi\setminus\Phi^*$. For $\phi\in \Phi$ let $Z_\phi$ be the number of homomorphisms from $(J,\boldS)$ to $H$ that agree with $\phi$. Then
\begin{equation}\label{equ:Net_Z*Z1}
Z^*=\sum_{\phi\in\Phi^*} Z_\phi \quad \text{ and } \quad Z_1=\sum_{\phi\in\Phi\setminus\Phi^*} Z_\phi.
\end{equation}

Let $\phi\in \Phi$. We are interested in the number $Z_\phi$ and investigate which homomorphisms $h$ agree with $\phi$. For each $v\in V(G)$ the image of $J_v$ under $h$ is fixed by the lists in $\boldS$ and the fact that $h(v)=\phi(v)$. Therefore, we only need to consider possible images of the graphs $J_e$. We make a case distinction depending on $e$.
\begin{itemize}
\item Let $e=\{u,v\}\in \Ecut(\phi)$. (This means that $e$ is dichromatic under $h$.) By the definition of $J_e$ it follows that the image $h(I_e)$ is a subset of $\NH(h(u)) \cap \NH(h(v))$. The vertices $h(u)$ and $h(v)$ are distinct and are from $\{w_1, w_2, w_3\}$. As $H$ is square-free it follows that $\NH(h(u)) \cap \NH(h(v))=\{w_1, w_2, w_3\}$. In addition, each vertex of $I_e$ is adjacent to one of the terminals. Since the images of these vertices are also in $\{w_1, w_2, w_3\}$ this does not bring any additional constraints. Summarising, since $I_e$ has size $t$ the edge $e$ contributes a factor of $3^t$ to $Z_\phi$.
\item Let $e=\{u,v\}\in \mathrm{Mon}_i(\phi)$ for some $i\in[3]$. (This means that $e$ is a monochromatic edge under $h$ with $h(u)=h(v)= w_i$.) Then, for $j\in [3]$ with $j\neq i$, by the same arguments as before we have $h(I_e^j)\subseteq \{w_1, w_2, w_3\}$. However the vertices in $I_e^i$ can be mapped to any neighbour of $w_i$. Therefore, each edge in $\mathrm{Mon}_i(\phi)$ contributes a factor of $\abs{\NH(w_i)}^{t_i}\cdot 3^{t-t_i}$ to $Z_\phi$. 
\end{itemize}

Using this knowledge, for $\phi \in \Phi$, we have
\begin{equation*}
Z_\phi= 3^{t\abs{\Ecut(\phi)}} \cdot \prod_{i\in[3]} \bigl(\abs{\NH(w_i)}^{t_i}\cdot 3^{t-t_i}\bigr)^{\abs{\mathrm{Mon}_i(\phi)}}.
\end{equation*}
Since $m=\abs{\Ecut(\phi)}+\sum_{i\in[3]}\abs{\mathrm{Mon}_i(\phi)}$ we can simplify this expression to
\begin{equation*}
Z_\phi=3^{tm}\prod_{i\in[3]} \left(\abs{\NH(w_i)}/3\right)^{t_i\abs{\mathrm{Mon}_i(\phi)}}.
\end{equation*}

For $i\in[3]$, recall the fact that $\abs{rs_i - t_i}\le \delta'/m$, where $s_i=s\cdot \log_{\abs{\NH(w_i)}/3}2$. For an upper bound on $Z_\phi$ we use the fact that $t_i\le r s_i +\delta'/m$. Also note that $\abs{\NH(w_i)}$ is bounded above by $\abs{V(H)}$ (which we use in the second inequality of the following expression). Furthermore, $\sum_{i\in [3]}\abs{\mathrm{Mon}_i(\phi)}\le m$ and $\delta'= \log_{\abs{V(H)}} e^\delta$ (which we use in the third inequality of the following expression). Then
\begin{align*}
Z_\phi &\le 3^{tm}\prod_{i\in[3]} \left(\abs{\NH(w_i)}/3\right)^{(r s_i + \delta'/m)\cdot\abs{\mathrm{Mon}_i(\phi)}}\\
&\le 3^{tm}\cdot 2^{(\sum_{i\in [3]}\abs{\mathrm{Mon}_i(\phi)})\cdot rs}\cdot {\abs{V(H)}}^{(\sum_{i\in [3]}\abs{\mathrm{Mon}_i(\phi)})\cdot \delta'/m}\\
&\le 3^{tm}\cdot 2^{(\sum_{i\in [3]}\abs{\mathrm{Mon}_i(\phi)})\cdot rs}\cdot e^\delta.
\end{align*}
Analogously, for a lower bound on $Z_\phi$ we use the fact that $t_i\ge r s_i -\delta'/m$. We obtain $3^{tm}\cdot 2^{(\sum_{i\in [3]}\abs{\mathrm{Mon}_i(\phi)})\cdot rs}\cdot e^{-\delta} \le Z_\phi$.
Summarising, we have
\begin{equation}\label{equ:Zphi} 
3^{tm}\cdot 2^{(\sum_{i\in [3]}\abs{\mathrm{Mon}_i(\phi)})\cdot rs}\cdot e^{-\delta} \le Z_\phi \le 3^{tm}\cdot 2^{(\sum_{i\in [3]}\abs{\mathrm{Mon}_i(\phi)})\cdot rs}\cdot e^\delta.
\end{equation}
Putting these bounds on $Z_\phi$ into the expression for $Z^*$ in~\eqref{equ:Net_Z*Z1} gives
\begin{align*}
\sum_{\phi\in \Phi^*} 3^{tm}\cdot 2^{(\sum_{i\in [3]}\abs{\mathrm{Mon}_i(\phi)})\cdot rs}\cdot e^{-\delta} \le Z^*&\le \sum_{\phi\in \Phi^*} 3^{tm}\cdot 2^{(\sum_{i\in [3]}\abs{\mathrm{Mon}_i(\phi)})\cdot rs}\cdot e^\delta.
\end{align*}
Since $\sum_{i\in [3]}\abs{\mathrm{Mon}_i(\phi)}=m-K$ for each $\phi \in \Phi^*$ and $\abs{\Phi^*}=T$ we obtain
\begin{equation*}
T \cdot 3^{tm}\cdot 2^{(m-K)\cdot rs} \cdot e^{-\delta} \le Z^*\le T \cdot 3^{tm}\cdot 2^{(m-K)\cdot rs} \cdot e^\delta.
\end{equation*} 
We set $L=3^{tm}2^{(m-K)rs}$ to obtain
\begin{equation}\label{equ:Z*}
 T \cdot e^{-\delta}\le Z^*/L \le T \cdot e^{\delta}.
\end{equation}

Assume for now that $\hom{(J,\boldS)}{H}/L\in [Z^*/L, Z^*/L+1/4]$. Then the algorithm that makes a $\Ret{H}$ oracle call with input $((J,\boldS),\eps/42)$ returns a solution $Q$ such that $Z^*/L \cdot e^{-\eps/2}\le \floor{Q/L} \le Z^*/L \cdot e^{\eps/2}$ as was shown in~\cite[Proof of Theorem 3]{DGGJApprox}. Using~\eqref{equ:Z*} and our choice of $\delta=\eps/2$ this gives
\[
T \cdot e^{-\eps}\le \floor{Q/L} \le T \cdot e^{\eps}.
\]
Therefore the output $\floor{Q/L}$ approximates $T$ with the desired precision.
It remains to show the following claim.

\medskip
\noindent{\bf\boldmath Claim: $\hom{(J, \boldS)}{H}/L \in [Z^*/L, Z^*/L+1/4]$.}
\medskip

\noindent {\it Proof of the claim:}\quad
Recall that $\hom{(J, \boldS)}{H}=Z^* + Z_1$. It is immediate that $\hom{(J, \boldS)}{H}/L \ge Z^*/L$. It remains to show that $Z_1/L \le 1/4$.

To obtain the following expression we first use~\eqref{equ:Net_Z*Z1} and~\eqref{equ:Zphi}. The third inequality then uses the fact that, for every $\phi\in \Phi\setminus\Phi^*$, we have $\abs{\Ecut(\phi)}\ge K+1$ and hence $\sum_{i\in [3]}\abs{\mathrm{Mon}_i(\phi)}\le m-(K+1)$. Finally, in the fourth inequality we use the fact that $\abs{\Phi\setminus\Phi^*}\le \abs{V(H)}^n$
\begin{align*}
Z_1 &= \sum_{\phi\in\Phi\setminus \Phi^*} Z_\phi\\
&\le \sum_{\phi\in\Phi\setminus \Phi^*} 3^{tm}\cdot 2^{(\sum_{i\in [3]}\abs{\mathrm{Mon}_i(\phi)})\cdot rs}\cdot e^\delta\\
&\le \sum_{\phi\in\Phi\setminus \Phi^*} 3^{tm}\cdot 2^{(m-K-1)\cdot rs}\cdot e^\delta\\
&\le \abs{V(H)}^n \cdot3^{tm}\cdot 2^{(m-K-1)\cdot rs}\cdot e^\delta
\end{align*}
Recall the definition $L=3^{tm}\cdot 2^{(m-K)\cdot rs}$. It follows that
\begin{align*}
Z_1/L &\le \frac{\abs{V(H)}^n \cdot e^\delta}{2^{rs}}\le 1/4,
\end{align*}
where the last inequality holds for sufficiently large $n$ by our choice of $s=n^2$ and the fact that $r\ge 1$. This proves the claim and completes the proof.
{\it (End of the proof of the claim.)}
\end{proof}

\subsection{Square-Free Graphs with an Induced Reflexive Cycle of Length at least $5$} \label{sec:inducedCycle}

The goal of this section is to show the following lemma. We build up to its proof, which is given at the end of this section. 
\newcommand{\stateleminducedCycleGeneral}{Let $H$ be a square-free graph. If $H$ contains a reflexive cycle of length at least $5$ as an induced subgraph then $\sat\leap \Ret{H}$.}		
\begin{lem}\label{lem:inducedCycleGeneral}
	\stateleminducedCycleGeneral
\end{lem}

Let $H$ be a connected square-free graph with an induced reflexive cycle of length at least $5$. If all cycles in $H$ have length at least $5$, then $H$ has girth at least $5$ and the complexity of $\Ret{H}$ is classified by Theorem~\ref{thm:RetGirth5}. In the special case where $H$ is reflexive this classification is straightforward to see: Either there is just a single cycle in $H$, then $H$ is a pseudotree and $\sat$-hardness follows from $\NP$-hardness for the decision problem~\cite[Corollary 4.2, Theorem 5.1]{FederPseudoForest} together with~\cite[Theorem 1]{DGGJApprox} --- or there are multiple cycles (all of which have length at least $5$), then there exists an induced $\WR{3}$ (as $H$ is connected) and hardness follows from Lemma~\ref{lem:inducedWR3General}.

Thus, it remains to show hardness if $H$ contains both a cycle of length at least $5$ as well as a cycle of length at most $5$, i.e.~(since $H$ is square-free) it contains a triangle. The hardness proof we give in this section will handle the case where $H$ includes triangles but will not rely on this fact (i.e.~it will also cover the before-mentioned case where all cycles have length at least $5$ without relying on hardness results for the decision problem).

As mentioned before, it is known that approximately counting list homomorphisms to reflexive graphs with an induced cycle of length at least $4$ is $\sat$-hard~\cite[Lemma 3.4]{GGJList}. That proof makes use of a certain set of two-vertex lists. In the proof of Lemma~\ref{lem:inducedCycleGeneral} we will use single-vertex lists to simulate these two-vertex lists.

As we have already shown $\sat$-hardness results for square-free graphs with an induced $\WR{3}$ or an induced net in Sections~\ref{sec:inducedWR3} and~\ref{sec:inducedNet} we now focus on graphs that do not contain such subgraphs. When considering reflexive graphs it turns out that this leaves a class of graphs which we call reflexive triangle-extended cycles. We will also make use of these reflexive triangle-extended cycles when considering square-free graphs $H$ that are not necessarily reflexive. When we do this we will restrict to the (reflexive) subgraph induced by the looped vertices of $H$.

\begin{defn}\label{def:extendedCycle}
A \emph{reflexive triangle-extended cycle} of length $q$ consists of a reflexive cycle $c_0, \dots, c_{q-1}$ together with a set $\calI\subseteq \{0, \dots, q-1\}$, and a reflexive triangle $d_i, c_i, c_{i+1\bmod q}$ for each $i\in \calI$.
An example of a reflexive triangle-extended cycle is depicted in Figure~\ref{fig:extendedCycle}.
\end{defn}

\begin{figure}[ht]
	\centering
	\begin{tikzpicture}[scale=1, every loop/.style={min distance=10mm,looseness=10}]

		\filldraw (0,0) node(c0){} circle[radius=3pt] --++ (72:1cm)  node(c4){} circle[radius=3pt] --++ (144:1cm )node(c3){} circle[radius=3pt] --++ (216:1cm) node(c2){} circle[radius=3pt]--++ (288:1cm) node(c1){} circle[radius=3pt] --++ (360:1cm);
		\filldraw (c2.center) --++ (228:1cm) node(d1){} circle[radius=3pt] --++ (348:1cm);
		\filldraw (c4.center) --++ (84:1cm) node(d3){} circle[radius=3pt] --++ (204:1cm);
		\filldraw (c0.center) --++ (12:1cm) node(d4){} circle[radius=3pt] --++ (132:1cm);

		\node at ($(c0)+(-54:-.3cm)$) {$c_0$}; 
		\node at ($(c1)+(54:.3cm)$) {$c_1$};
		\node at ($(c2)+(162:-.3cm)$) {$c_2$};
		\node at ($(c3)+(270:.3cm)$) {$c_3$};
		\node at ($(c4)+(18:-.3cm)$) {$c_4$};
		\node at ($(d1)+(198:-.4cm)$) {$d_1$};
		\node at ($(d3)+(54:-.4cm)$) {$d_3$};
		\node at ($(d4)+(-20:-.4cm)$) {$d_4$};
		
		\path[-] (c0.center) edge  [in=-24,out=-84,loop] node {} ();
		\path[-] (c1.center) edge  [in=204,out=264,loop] node {} ();
		\path[-] (c2.center) edge  [in=132,out=192,loop] node {} ();
		\path[-] (c3.center) edge  [in=60,out=120,loop] node {} ();
		\path[-] (c4.center) edge  [in=348,out=48,loop] node {} ();
		\path[-] (d1.center) edge  [in=168,out=228,loop] node {} ();
		\path[-] (d3.center) edge  [in=24,out=84,loop] node {} ();
		\path[-] (d4.center) edge  [in=-42,out=18,loop] node {} ();
	\end{tikzpicture}
	\caption{Reflexive triangle-extended cycle with $q=5$ and $\calI=\{1,3,4\}$.}
\label{fig:extendedCycle}
\end{figure}
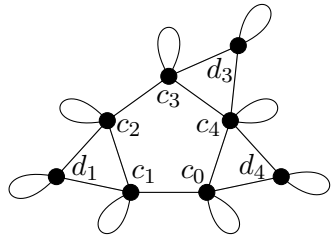
\begin{defn}\label{def:extendedPath}
	Analogously to Definition~\ref{def:extendedCycle}, a \emph{reflexive triangle-extended path} is a reflexive path $c_0, \dots, c_{q-1}$ together with a set $\calI\subseteq \{0, \dots, q-2\}$,  and a reflexive triangle $d_i, c_i, c_{i+1}$ for each $i\in \calI$.
\end{defn}

\begin{lem}\label{lem:extendedCycle}
Let $H$ be a connected reflexive square-free graph that does not contain an induced $\WR{3}$ and also does not contain an induced net. If $H$ contains an induced cycle of length at least $5$ then $H$ is a reflexive triangle-extended cycle of length at least $5$. Otherwise it is a reflexive triangle-extended path.
\end{lem}
\begin{proof}
\noindent{\bf\boldmath Case 1: $H$ contains an induced cycle $C=c_0, \dots, c_{q-1}$ with $q\ge 5$.}
If $H$ is just the cycle $C$, then the statement of the lemma is true ($\calI =\emptyset$). Otherwise, consider any $d\in V(H)\setminus V(C)$ with a neighbour $c\in V(C)$ (has to exist since $H$ is connected).

Since $H$ does not contain an induced $\WR{3}$ the vertex $d$ is adjacent to a neighbour $c'\in V(C)$ of $c$. Let $c^0$ and $c''$ be the other neighbours of $c$ and $c'$ in $C$, respectively, i.e.~ $\Nb{C}(c)=\{c^0, c'\}$ and $\Nb{C}(c')=\{c, c''\}$. The vertices $\{c^0, c, c', c''\}$ are all distinct as $C$ has length at least $5$. As $H$ is square-free we observe
\begin{equation}\label{equ:extendedCycle1}
\{d,c^0\} \notin E(H)\text{ and } \{d,c''\} \notin E(H).
\end{equation} 

The proof of the following claim directly proves that $H$ is a reflexive triangle-extended cycle.

\medskip
\noindent{\bf\boldmath Claim: $\NH(d)=\{c,c'\}$.}

\smallskip
\noindent{\it Proof of the claim:} Assume there exists a neighbour $d'\notin\{c, c'\}$ of $d$ in $H$. By~\eqref{equ:extendedCycle1} we have $d'\notin \{c^0, c''\}$. 
Furthermore, since $H$ is square-free, we obtain the following.
\begin{equation}\label{equ:extendedCycle2}
\text{There is no } u\neq d \text{ with } u\in \NH(c) \cap \NH(d') \text{ or } u\in \NH(c') \cap \NH(d').
\end{equation}
Let $H'$ be the subgraph of $H$ induced by the vertices $\{c^0, c, c', c'', d, d'\}$, see Figure~\ref{fig:extendedCycleHprime}. Then $H'$ is a net (cf.~Figure~\ref{fig:net} where $c=w_1$, $d=w_2$ and $c'=w_3$).

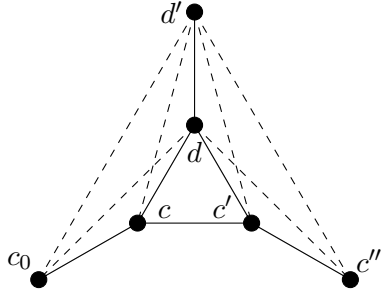
\begin{figure}[ht]
	\centering
	\begin{tikzpicture}[scale=1, baseline=0.36cm, every loop/.style={min distance=10mm,looseness=10}]
	
	\filldraw (0,0) node(d){} circle[radius=3pt] --++ (30:1.5cm) node(a){} circle[radius=3pt] --++ (60:1.5cm) node(b){} circle[radius=3pt] --++ (-60:1.5cm) node(c){} circle[radius=3pt]--++ (-30:1.5cm) node(f){} circle[radius=3pt];
	\filldraw (b.center) --++ (90:1.5cm) node(e){} circle[radius=3pt];
	\draw (a.center) -- (c.center);	
	
	\node at ($(d)+(-.25cm,.25cm)$) {$c_0$};	
	\node at ($(a)+(.35cm,.15cm)$) {$c$};	
	\node at ($(b)+(0cm,-.35cm)$) {$d$};	
	\node at ($(e)+(-.3cm,0cm)$) {$d'$};
	\node at ($(c)+(-.38cm,.22cm)$) {$c'$};
	\node at ($(f)+(.25cm,.25cm)$) {$c''$};
	
	\path[dashed] (e.center) edge (d.center);
	\path[dashed] (e.center) edge (a.center);
	\path[dashed] (e.center) edge (c.center);
	\path[dashed] (e.center) edge (f.center);
	\path[dashed] (b.center) edge (d.center);
	\path[dashed] (b.center) edge (f.center);

	\end{tikzpicture}
	\caption{The graph $H'$ induced by $\{c^0, c, c', c'', d, d'\}$. Loops are omitted. Dashed lines show edges that cannot exist by the fact that $H$ is square-free.}
	\label{fig:extendedCycleHprime}
\end{figure}

Because of~\eqref{equ:extendedCycle2} we have $\Nb{H'}(d')=\{d\}$. (The dashed edges incident to $d'$ in Figure~\ref{fig:extendedCycleHprime} cannot exist.) Because of~\eqref{equ:extendedCycle1} we have $\Nb{H'}(d)=\{d',c,c'\}$. (The dashed edges incident to $d$ in Figure~\ref{fig:extendedCycleHprime} cannot exist.) Finally, since $C$ is an induced cycle, there are no edges between the vertices $\{c^0, c, c', c''\}$ outside of $C$. Therefore, $H'$ is an \emph{induced} net in $H$, which gives a contradiction. This proves the claim in Case 1. {\bf (End of Case 1)}

\bigskip\noindent{\bf \boldmath Case 2: All induced cycles in $H$ are triangles.} This case is handled very similarly to the previous one: Let $P= c_0, \dots, c_{q-1}$ be a maximal induced path in $H$. If $H$ is just the path $P$, then the statement of the lemma is true ($\calI=\emptyset$). Otherwise, let $d\in V(H)\setminus V(P)$ be a neighbour of $c\in V(P)$. We show that $d$ is adjacent to a neighbour $c'\in V(P)$ of $c$:
\begin{itemize}
	\item If $c$ is an inner vertex of $P$ then, since $H$ does not contain an induced $\WR{3}$, $d$ is also adjacent to a neighbour of $c$ in $P$. 
	\item If $c$ is an endpoint of $P$, then $d$ has to be adjacent to a vertex $c'\in V(P)$ with $c'\neq c$ as $P$ is maximal induced. Without loss of generality assume that $c'$ is the neighbour of $d$ which is closest to $c$ in $P$. Then $c', c, d$ has to be a triangle ($c'$ has to be a neighbour of $c$) as $P$ is induced and all induced cycles in $H$ are triangles.
\end{itemize} 

Then the proof of the following claim shows that $H$ is a reflexive triangle-extended path.

\medskip
\noindent{\bf\boldmath Claim: $\NH(d)=\{c,c'\}$.}

\smallskip
\noindent{\it Proof of the claim:} Assume there exists a neighbour $d'\notin\{c, c'\}$ of $d$ in $H$. We observe the following properties:
\begin{itemize}
	\item The vertex $d'$ does not have a neighbour in $P$: Assume the opposite and let $u\in P$ be a neighbour of $d'$. Without loss of generality let $u$ be closer to $c'$ than $c$ in $P$. Furthermore, let $u$ be the neighbour of $d'$ in $P$ which is closest to $c'$. Then the edge $\{d',u\}$ and the path $d', d,c'$ close an induced cycle with $P$. If $u\neq c'$ this cycle has length greater than $3$, a contradiction. If $u=c'$ we obtain a contradiction to the fact that $H$ is square-free (see Figure~\ref{fig:extendedCycleHprime}). 
	\item Both $c$ and $c'$ are inner points of $P$: Suppose, for contradiction, that $c'$ is an end point of $P$. Since $d'$ does not have a neighbour in $P$, replacing $c'$ by $d,d'$ in $P$ gives an induced path $P'$. Moreover, $P'$ is longer than $P$ which is a contradiction to the maximality of $P$. This shows that $c'$ cannot be an end point of $P$. Analogously $c$ cannot be an end point of $P$.
\end{itemize}

Then let $c^0$ and $c''$ be the other neighbours of $c$ and $c'$ in $P$ (they have to exist since $c$ and $c'$ are inner points of $P$). The remainder of the argument is analogous to the proof of the claim in Case 1.
{\bf (End of Case 2)}
\end{proof}

The goal of the remainder of this section is to prove Lemma~\ref{lem:inducedCycleGeneral}. In order to prove Lemma~\ref{lem:inducedCycleGeneral} we work with the following parameterised version of the list homomorphism counting problem. Let $H$ be a graph and $\calL$ be a set of subsets of $V(H)$.

\prob
	{
	$\Hom{H,\calL}$.
	}
	{
	An irreflexive graph $G$ and a collection of lists $\boldS=\{S_v\in \calL\mid v\in V(G)\}$.
	}
	{
	$\hom{(G,\boldS)}{H}$.
	}
We also use the following lemma.
\begin{lem}[{\cite[Proof of Lemma 3.4]{GGJList}}]\label{lem:LHomReflexiveCycleHardness}
	Let $H$ be a graph that contains an induced reflexive cycle $C=c_0,\dots, c_{q-1}$ on $q\ge 4$ vertices. Let $\calL=\{\{c_0,c_1\},\{c_0,c_2\},\dots, \{c_0,c_{q-1}\}\}$. Then $\Hom{H,\calL}\eqap\sat$.
\end{lem}

The key to proving Lemma~\ref{lem:inducedCycleGeneral} is the following result. It states that for certain graphs that contain a reflexive triangle-extended cycle we can simulate each size-$2$ list of vertices in the corresponding cycle $C$ by gadgets using only single-vertex lists.

\begin{lem}\label{lem:extendedCycleAnalysis}
Let $H$ be a square-free graph that does not contain any mixed triangle as an induced subgraph and let $H^*$ be the graph induced by the looped vertices of $H$. Suppose that $H^*$ contains a connected component $H^{**}$ that is a reflexive triangle-extended cycle, where $C=c_0, \dots, c_{q-1}$ is the corresponding reflexive cycle as given by Definition~\ref{def:extendedCycle} and the length of $C$ is $q\ge 5$. Let $\calL$ and $\calL'$ be sets with 
\[
\calL' \subsetneq \calL \subseteq \{\{c_0, c_1\}, \{c_0, c_2\}, \dots, \{c_0, c_{q-1}\}\} \quad\text {such that}\quad\abs{\calL'}=\abs{\calL}-1.
\]
Let $\calL''= \calL'\cup \bigl\{S \subseteq V(H) \mid \abs{S}\in \{1,\abs{V(H)}\}\bigr\}$. Then
\[
\Hom{H,\calL} \leap \Hom{H, \calL''}.
\]
\end{lem}
\begin{proof}
For the reflexive triangle-extended cycle $H^{**}$ we use the notation ($C$, $\calI$ and $d_i$ for $i\in \calI$) as given by Definition~\ref{def:extendedCycle}. Let $\calL$, $\calL'$ and $\calL''$ be as given in the statement of the lemma.
We have $\calL'= \calL\setminus \{\{c_0, c_{\ell}\}\}$ for some $\ell\in [q-1]$.
Let $(G,\boldS^G)$ be an input to $\Hom{H,\calL}$. Let $U=\{u\in V(G) \mid S_u^G=\{c_0, c_{\ell}\}\}$. 
Since $\{c_0, c_{\ell}\}$ is not part of $\calL''$ the goal is to simulate $\{c_0, c_{\ell}\}$ using gadgetry and lists from $\calL''$.
 
From $(G,\boldS^G)$ we define an instance $(J, \boldS^J)$ of $\Hom{H, \calL''}$.
To this end we will define, for each $u\in U$, a vertex gadget $J_u$ and a corresponding set of lists $\boldS^u=\{S^{u}_v \in \calL'' \mid v\in V(J_u)\}$. 
There are two distinct paths in $C$ that connect $c_0$ and $c_{\ell}$: $P_1=c_0, c_1, \dots, c_{\ell}$ and $P_2=c_{\ell}, \dots, c_{q-1}, c_0$. The graph $J_u$ has two parts: a graph $J_{P_1}$ and a graph $J_{P_2}$, which depend on the paths $P_1$ and $P_2$. We first define $J_{P_1}$ and $J_{P_2}$ (and the corresponding sets of lists $\boldS^1=\{S^{1}_v \in \calL'' \mid v\in V(J_{P_1})\}$ and $\boldS^2=\{S^{2}_v \in \calL'' \mid v\in V(J_{P_2})\}$) and then we describe the way in which they are connected to form $J_u$. The definition of $J_{P_1}$ depends on $\ell$, the number of edges of $P_1$:

\begin{itemize}
\item If $\ell$ is even, think of a path on $\ell/2$ edges. Let $v^*$ be one of the end points of this path. We pin $v^*$ to $c_{\ell/2}$ (the vertex in the ``middle'' of $P_1$). This graph is depicted in Figure~\ref{fig:inducedCycleGadget1} on the left. The graph $J_{P_1}$ is then a modification of this graph where each vertex of the path is replaced by a clique of size $2$ (apart from $v^*$ which will be pinned to $c_{\ell/2}$ anyway). This modification will ensure that only looped vertices can be in the image of $J_{P_1}$. The graph $J_{P_1}$ is depicted in Figure~\ref{fig:inducedCycleGadget1} (on the right) and is formally defined as follows:
 $V(J_{P_1})=\{v_i, v_i' \mid i\in \{0, \dots, \ell/2-1\}\} \cup \{v^*\}$, where all these vertices are distinct from the vertices of $G$ (and distinct from the vertices used in other gadgets). $E(J_{P_1})= \bigl\{\{v_i, v_i'\} \mid i\in \{0, \dots, \ell/2-1\}\bigr\} \cup  \bigl\{\ucp{\{v_{i-1},v_{i-1}'\}}{\{v_{i},v_{i}'\}} \mid i\in [\ell/2-1]\bigr\} \cup \bigl\{\{v_{\ell/2-1},v^*\}, \{v_{\ell/2-1}',v^*\}\bigr\}$.
 We set $S^1_{v^*}=\{c_{\ell/2}\}$ and $S^1_v=V(H)$ for all $v\in V(J_{P_1})\setminus\{v^*\}$
\item If $\ell$ is odd, then $J_{P_1}$ is defined very similarly to the previous case, see Figure~\ref{fig:inducedCycleGadget2}. Formally, $V(J_{P_1})=\{v_i, v_i' \mid i\in \{0, \dots, \floor{\ell/2}\}\} \cup \{v_1^*, v_2^*\}$, where all these vertices are distinct from the vertices of $G$ (and distinct from the vertices used in other gadgets). $E(J_{P_1})= \bigl\{\{v_i, v_i'\} \mid i\in \{0, \dots, \floor{\ell/2}\}\bigr\} \cup \bigl\{\ucp{\{v_{i-1},v_{i-1}'\}}{\{v_{i},v_{i}'\}} \mid i\in [\floor{\ell/2}]\bigr\} \cup \bigl\{\ucp{\{v_{\floor{\ell/2}},v_{\floor{\ell/2}}'\}}{\{v^*_{1},v^*_{2}\}}\bigr\}$.
We set $S^{1}_{v_1^*}=\{c_{\ceil{\ell/2}}\}$, $S^{1}_{v_2^*}=\{c_{\floor{\ell/2}}\}$ and $S^{1}_v=V(H)$ for all $v\in V(J_{P_1})\setminus\{v_1^*, v_2^*\}$.
\end{itemize}

\begin{figure}[ht]
	\centering
	\begin{minipage}{.49 \textwidth}
		\centering
		\begin{tikzpicture}[scale=1, baseline=0.36cm, every loop/.style={min distance=10mm,looseness=10}]
		
		\node[circle,fill=black,inner sep=0pt,minimum size=5pt] (u) at (0,0){};
		\node[circle,fill=black,inner sep=0pt,minimum size=5pt] (wfirst) at (1,0){};
		\node[circle,fill=black,inner sep=0pt,minimum size=5pt] (wsecond) at (2,0){};
		\node[circle,fill=black,inner sep=0pt,minimum size=5pt] (wlast) at (4,0){};
		\node[circle,fill=black,inner sep=0pt,minimum size=5pt] (p) at (5,0){};
		\draw (u) -- (wfirst);
		\draw (wfirst) --  (wsecond);
		\draw (wlast) -- (p);
		
		\coordinate (pmid) at (3,0);
		\node[circle,fill=black,inner sep=0pt,minimum size=3pt] (cdot1) at ($(pmid)+(-.3cm,0cm)$){};
		\node[circle,fill=black,inner sep=0pt,minimum size=3pt] (cdot1) at (pmid){};
		\node[circle,fill=black,inner sep=0pt,minimum size=3pt] (cdot1) at ($(pmid)+(+.3cm,0cm)$){};
		
		\draw[<->,line width = 2pt] ($(u.east)+(0,.5cm)$) -- node[yshift=.5cm](){$\ell/2$ edges} ($(p.west)+(0,.5cm)$);
		
		\node at ($(p.center)+(1.1cm,-.05cm)$) {$v^* \rightarrow c_{\ell/2}$};
		
		\end{tikzpicture}
	\end{minipage}
	\begin{minipage}{.49 \textwidth}
		\centering
		\begin{tikzpicture}[scale=1, baseline=0.36cm, every loop/.style={min distance=10mm,looseness=10}]
		
		\node[circle,fill=black,inner sep=0pt,minimum size=5pt] (u1) at (0,.5){};
		\node[circle,fill=black,inner sep=0pt,minimum size=5pt] (u2) at (0,-.5){};
		\node[circle,fill=black,inner sep=0pt,minimum size=5pt] (wfirst1) at (1,.5){};
		\node[circle,fill=black,inner sep=0pt,minimum size=5pt] (wfirst2) at (1,-.5){};
		\node[circle,fill=black,inner sep=0pt,minimum size=5pt] (wsecond1) at (2,.5){};
		\node[circle,fill=black,inner sep=0pt,minimum size=5pt] (wsecond2) at (2,-.5){};
		\node[circle,fill=black,inner sep=0pt,minimum size=5pt] (wlast1) at (4,.5){};
		\node[circle,fill=black,inner sep=0pt,minimum size=5pt] (wlast2) at (4,-.5){};
		\node[circle,fill=black,inner sep=0pt,minimum size=5pt] (p) at (5,0){};
		\draw (u1) -- (wfirst1) -- (u2) -- (wfirst2);
		\draw (u2) -- (u1) -- (wfirst2) -- (wfirst1);
		\draw (wfirst1) -- (wsecond1) -- (wfirst2) -- (wsecond2);
		\draw (wfirst2) -- (wfirst1) -- (wsecond2) -- (wsecond1);
		\draw (wlast1) -- (p) -- (wlast2) -- (wlast1);
		
		\coordinate (pmid) at (3,0);
		\node[circle,fill=black,inner sep=0pt,minimum size=3pt] (cdot1) at ($(pmid)+(-.3cm,0cm)$){};
		\node[circle,fill=black,inner sep=0pt,minimum size=3pt] (cdot1) at (pmid){};
		\node[circle,fill=black,inner sep=0pt,minimum size=3pt] (cdot1) at ($(pmid)+(+.3cm,0cm)$){};
		
		\draw[<->,line width = 2pt] ($(u1.east)+(0,.5cm)$) -- node[yshift=.5cm](){$\ell/2$ edges} ($(p.west)+(0,1cm)$);
		
		\node at ($(u1)+(-.5cm,0)$) {$v_0$};
		\node at ($(u2)+(-.5cm,0)$) {$v_0'$};
		\node at ($(p.center)+(1.1cm,-.05cm)$) {$v^* \rightarrow c_{\ell/2}$};
		\end{tikzpicture}
	\end{minipage}
	\caption{Construction of the graph $J_{P_1}$ for even $\ell$. The label of the form $v^* \rightarrow c_{\ell/2}$ means  that the vertex $v^*\in V(J_{P_1})$ is ``pinned'' to $c_{\ell/2}\in V(H)$ since $S^1_{v^*}=\{c_{\ell/2}\}$.}
	\label{fig:inducedCycleGadget1}
\end{figure}
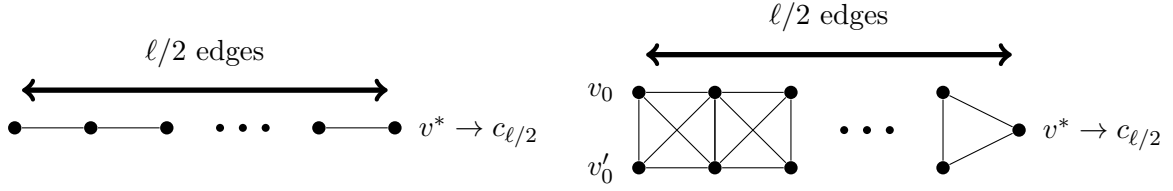

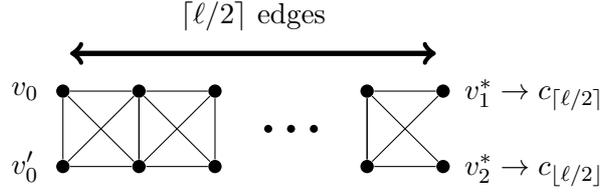
\begin{figure}[ht]
	\centering
	\begin{tikzpicture}[scale=1, baseline=0.36cm, every loop/.style={min distance=10mm,looseness=10}]
	
	\node[circle,fill=black,inner sep=0pt,minimum size=5pt] (u1) at (0,.5){};
	\node[circle,fill=black,inner sep=0pt,minimum size=5pt] (u2) at (0,-.5){};
	\node[circle,fill=black,inner sep=0pt,minimum size=5pt] (wfirst1) at (1,.5){};
	\node[circle,fill=black,inner sep=0pt,minimum size=5pt] (wfirst2) at (1,-.5){};
	\node[circle,fill=black,inner sep=0pt,minimum size=5pt] (wsecond1) at (2,.5){};
	\node[circle,fill=black,inner sep=0pt,minimum size=5pt] (wsecond2) at (2,-.5){};
	\node[circle,fill=black,inner sep=0pt,minimum size=5pt] (wlast1) at (4,.5){};
	\node[circle,fill=black,inner sep=0pt,minimum size=5pt] (wlast2) at (4,-.5){};
	\node[circle,fill=black,inner sep=0pt,minimum size=5pt] (p1) at (5cm,.5cm){};
	\node[circle,fill=black,inner sep=0pt,minimum size=5pt] (p2) at (5cm,-.5cm){};
	\draw (u1) -- (wfirst1) -- (u2) -- (wfirst2);
	\draw (u2) -- (u1) -- (wfirst2) -- (wfirst1);
	\draw (wfirst1) -- (wsecond1) -- (wfirst2) -- (wsecond2);
	\draw (wfirst2) -- (wfirst1) -- (wsecond2) -- (wsecond1);
	\draw (wlast1) -- (p1) -- (wlast2) -- (wlast1);
	\draw (wlast1) -- (p2) -- (wlast2) -- (wlast1);
	
	\coordinate (pmid) at (3,0);
	\node[circle,fill=black,inner sep=0pt,minimum size=3pt] (cdot1) at ($(pmid)+(-.3cm,0cm)$){};
	\node[circle,fill=black,inner sep=0pt,minimum size=3pt] (cdot1) at (pmid){};
	\node[circle,fill=black,inner sep=0pt,minimum size=3pt] (cdot1) at ($(pmid)+(+.3cm,0cm)$){};
	
	\draw[<->,line width = 2pt] ($(u1.east)+(0,.5cm)$) -- node[yshift=.5cm](){$\ceil{\ell/2}$ edges} ($(p1.west)+(0,.5cm)$);
	
	\node at ($(u1)+(-.5cm,0)$) {$v_0$};
	\node at ($(u2)+(-.5cm,0)$) {$v_0'$};
	\node at ($(p1.center)+(1.2cm,-.05cm)$) {$v_1^* \rightarrow c_{\ceil{\ell/2}}$};
	\node at ($(p2.center)+(1.2cm,-.05cm)$) {$v_2^* \rightarrow c_{\floor{\ell/2}}$};
	\end{tikzpicture}
	\caption{The graph $J_{P_1}$ for odd $\ell$. A label of the form $a \rightarrow b$ means that the vertex $a\in V(J_{P_1})$ is ``pinned'' to $b\in V(H)$ since $S^1_a=\{b\}$.}
	\label{fig:inducedCycleGadget2}
\end{figure}

This completes the definition of $J_{P_1}$. $J_{P_2}$ is defined analogously. However, the length of  $P_2$ is $q-\ell$ instead of $\ell$. Furthermore, if $q-\ell$ is even, note that the vertex in the ``middle'' of $P_2=c_{\ell}, \dots, c_{q-1}, c_0$ is $c_{(q+\ell)/2}$ rather than $c_{\ell/2}$. (Accordingly, if $q-\ell$ is odd, we use $c_{\ceil{(q+\ell)/2}}$ and $c_{\floor{(q+\ell)/2}}$ to ``pin'' to.) Formally, $J_{P_2}$ is defined as follows:
\begin{itemize}
	\item If $q-\ell$ is even, we have $V(J_{P_2})=\{w_i, w_i' \mid i\in \{0, \dots, (q-\ell)/2-1\}\} \cup \{w^*\}$ and $E(J_{P_2})= \bigl\{\{w_i, w_i'\} \mid i\in \{0, \dots, (q-\ell)/2-1\}\bigr\} \cup  \bigl\{\ucp{\{w_{i-1},w_{i-1}'\}}{\{w_{i},w_{i}'\}} \mid i\in [(q-\ell)/2-1]\bigr\} \cup \bigl\{\{w_{(q-\ell)/2-1},w^*\}, \{w_{(q-\ell)/2-1}',w^*\}\bigr\}$.
	We set $S^{2}_{w^*}=\{c_{(q+\ell)/2}\}$ and $S^{2}_w=V(H)$ for all $w\in V(J_{P_2})\setminus\{w^*\}$
	\item If $q-\ell$ is odd, we have $V(J_{P_2})=\{w_i, w_i' \mid i\in \{0, \dots, \floor{(q-\ell)/2}\}\} \cup \{w_1^*, w_2^*\}$ and $E(J_{P_2})= \bigl\{\{w_i, w_i'\} \mid i\in \{0, \dots, \floor{(q-\ell)/2}\}\bigr\} \cup \bigl\{\ucp{\{w_{i-1},w_{i-1}'\}}{\{w_{i},w_{i}'\}} \mid i\in [\floor{(q-\ell)/2}]\bigr\} \cup \bigl\{\ucp{\{w_{\floor{(q-\ell)/2}},w_{\floor{(q-\ell)/2}}'\}}{\{w^*_{1},w^*_{2}\}}\bigr\}$.
	We set $S^{2}_{w_1^*}=\{c_{\ceil{(q+\ell)/2}}\}$, $S^{2}_{w_2^*}=\{c_{\floor{(q+\ell)/2}}\}$ and $S^{2}_w=V(H)$ for all $w\in V(J_{P_2})\setminus\{w_1^*, w_2^*\}$.
\end{itemize}

We can now define the graph $J_u$ (for $u\in U$): $J_u$ is the graph obtained from $J_{P_1}$ and $J_{P_2}$ by identifying $v_0$ with $w_0$ and $v_0'$ with $w_0'$. As an example, if $P_1$ has even length and $P_2$ has odd length, the graph $J_u$ is depicted in Figure~\ref{fig:Ju}. Let $\calA(u)$ denote the set that contains the two vertices $v_0=w_0$ and $v_0'=w_0'$. The lists $\boldS^u$ of the vertices in $J_u$ are the union of $\boldS^1$ and $\boldS^2$. (Note that this is well-defined as $S_{v_0}^{1}=S_{v_0'}^{1}=S_{w_0}^{2}=S_{w_0'}^{2}=V(H)$.) 
This completes the definition of $J_u$ and $\boldS^u$.

We can finally define the instance $(J,\boldS^J)$: $J$ is the graph with vertices $V(J_u)=V(G)\setminus U \cup\bigcup_{u\in U} V(J_u)$ and edges
\begin{align*}
	E(J)=	&\bigl\{\{v,v'\} \mid \{v,v'\} \in 	E(G)\text{ and }v,v'\in V(G)\setminus U\bigr\}\\
			&\cup \bigl\{\ucp{\{v\}}{\calA(v')} \mid \{v,v'\} \in 	E(G)\text{ and }v \in V(G)\setminus U, v'\in U\bigr\}\\
			&\cup \bigl\{\ucp{\calA(v)}{\calA(v')} \mid \{v,v'\} \in E(G)\text{ and }v,v'\in U\bigr\}\\
			&\cup \bigcup_{u\in U} E(J_u).
\end{align*}
Finally, $\boldS^J=\{S_v^J \subseteq V(H) \mid v\in V(J)\}$ with 
\[
S_v^J=
\begin{cases}
S_v^G, &\text{if }v\in V(G)\setminus U\\
S^u_v &\text{if otherwise }v\in V(J_u) \text{ for some }u\in U.
\end{cases}
\]
Note that for each $v\in V(J)$ we have $S_v^J\in \calL''$ and therefore $(J, \boldS^J)$ is a valid input to $\Hom{H, \calL''}$.

To show how homomorphisms from $(J, \boldS^J)$ to $H$ relate to homomorphisms from  $(G, \boldS^G)$ to $H$ we determine some properties of the gadget $J_u$. Consider the case where $\ell$, the number of edges of $P_1$, is even, and $q-\ell$, the number of edges of $P_2$, is odd. (The other cases of $\ell$ and $q-\ell$ even or odd will be analogous.) Then $J_u$ is the gadget depicted in Figure~\ref{fig:Ju}. Let the vertices of $J_u$ be labelled accordingly.
Now let $h$ be a homomorphism from $(J, \boldS^J)$ to $H$.

\begin{figure}[ht]
	\centering
	\begin{tikzpicture}[scale=.9, baseline=0.36cm, every loop/.style={min distance=10mm,looseness=10}]
	
	\node[circle,fill=black,inner sep=0pt,minimum size=5pt] (u1) at (0,.75){};
	\node[circle,fill=black,inner sep=0pt,minimum size=5pt] (u2) at (0,-.75){};
	\node[circle,fill=black,inner sep=0pt,minimum size=5pt] (vfirst1) at (-1.5,.75){};
	\node[circle,fill=black,inner sep=0pt,minimum size=5pt] (vfirst2) at (-1.5,-.75){};
	\node[circle,fill=black,inner sep=0pt,minimum size=5pt] (vsecond1) at (-3,.75){};
	\node[circle,fill=black,inner sep=0pt,minimum size=5pt] (vsecond2) at (-3,-.75){};
	\node[circle,fill=black,inner sep=0pt,minimum size=5pt] (vlast1) at (-6,.75){};
	\node[circle,fill=black,inner sep=0pt,minimum size=5pt] (vlast2) at (-6,-.75){};
	\node[circle,fill=black,inner sep=0pt,minimum size=5pt] (p) at (-7.5,0){};
	\draw (u1) -- (vfirst1) -- (u2) -- (vfirst2);
	\draw (u2) -- (u1) -- (vfirst2) -- (vfirst1);
	\draw (vfirst1) -- (vsecond1) -- (vfirst2) -- (vsecond2);
	\draw (vfirst2) -- (vfirst1) -- (vsecond2) -- (vsecond1);
	\draw (vlast1) -- (p) -- (vlast2) -- (vlast1);
	
	\coordinate (pmid) at (-4.5,0);
	\node[circle,fill=black,inner sep=0pt,minimum size=3pt] (cdot1) at ($(pmid)+(-.3cm,0cm)$){};
	\node[circle,fill=black,inner sep=0pt,minimum size=3pt] (cdot1) at (pmid){};
	\node[circle,fill=black,inner sep=0pt,minimum size=3pt] (cdot1) at ($(pmid)+(+.3cm,0cm)$){};

	\node[circle,fill=black,inner sep=0pt,minimum size=5pt] (wfirst1) at (1.5,.75){};
	\node[circle,fill=black,inner sep=0pt,minimum size=5pt] (wfirst2) at (1.5,-.75){};
	\node[circle,fill=black,inner sep=0pt,minimum size=5pt] (wsecond1) at (3,.75){};
	\node[circle,fill=black,inner sep=0pt,minimum size=5pt] (wsecond2) at (3,-.75){};
	\node[circle,fill=black,inner sep=0pt,minimum size=5pt] (wlast1) at (6,.75){};
	\node[circle,fill=black,inner sep=0pt,minimum size=5pt] (wlast2) at (6,-.75){};
	\node[circle,fill=black,inner sep=0pt,minimum size=5pt] (p1) at (7.5,.75){};
	\node[circle,fill=black,inner sep=0pt,minimum size=5pt] (p2) at (7.5,-.75){};
	\draw (u1) -- (wfirst1) -- (u2) -- (wfirst2);
	\draw (u2) -- (u1) -- (wfirst2) -- (wfirst1);
	\draw (wfirst1) -- (wsecond1) -- (wfirst2) -- (wsecond2);
	\draw (wfirst2) -- (wfirst1) -- (wsecond2) -- (wsecond1);
	\draw (wlast1) -- (p1) -- (wlast2) -- (wlast1);
	\draw (wlast1) -- (p2) -- (wlast2) -- (wlast1);
	
	\coordinate (pmid) at (4.5,0);
	\node[circle,fill=black,inner sep=0pt,minimum size=3pt] (cdot1) at ($(pmid)+(-.3cm,0cm)$){};
	\node[circle,fill=black,inner sep=0pt,minimum size=3pt] (cdot1) at (pmid){};
	\node[circle,fill=black,inner sep=0pt,minimum size=3pt] (cdot1) at ($(pmid)+(+.3cm,0cm)$){};
	
	\draw [decorate,decoration={brace,amplitude=10pt},xshift=0pt,yshift=0pt]
	(0,-2cm) -- (-7.5cm,-2cm) node [midway, yshift=-.75cm] {$J_{P_1}$};
	\draw [decorate,decoration={brace,amplitude=10pt, mirror},xshift=0pt,yshift=0pt]
	(0,-2cm) -- (7.5cm,-2cm) node [midway, yshift=-.75cm] {$J_{P_2}$};

	\node at ($(u1)+(0,.5cm)$) {$v_0=w_0$};
	\node at ($(u2)+(0,-.5cm)$) {$v_0'=w_0'$};
	\draw[rounded corners] ($(u1)+(-.8,.8cm)$) rectangle ($(u2)+(.8,-.8cm)$);
	\node at ($(u1)+(0,1.15cm)$) {(replaces $u$ in $G$)};
	\node at ($(u1)+(0,1.6cm)$) {$\calA(u)$};
	
	\node at ($(vfirst1)+(0,.5cm)$) {$v_1$};
	\node at ($(vfirst2)+(0,-.5cm)$) {$v_1'$};
	\node at ($(vsecond1)+(0,.5cm)$) {$v_2$};
	\node at ($(vsecond2)+(0,-.5cm)$) {$v_2'$};
	\node at ($(vlast1)+(0,.5cm)$) {$v_{\ell/2-1}$};
	\node at ($(vlast2)+(0,-.5cm)$) {$v_{\ell/2-1}'$};
	\node at ($(p)+(-1,0cm)$) {$v^*\rightarrow c_{\ell/2}$};
	
	\node at ($(wfirst1)+(0,.5cm)$) {$w_1$};
	\node at ($(wfirst2)+(0,-.5cm)$) {$w_1'$};
	\node at ($(wsecond1)+(0,.5cm)$) {$w_2$};
	\node at ($(wsecond2)+(0,-.5cm)$) {$w_2'$};
	\node at ($(wlast1)+(0,.5cm)$) {$w_{\floor{\frac{q-\ell}{2}}}$};
	\node at ($(wlast2)+(0,-.5cm)$) {$w_{\floor{\frac{q-\ell}{2}}}'$};
	\node at ($(p1)+(1,.5cm)$) {$w_1^*\rightarrow c_{\ceil{\frac{q+\ell}{2}}}$};
	\node at ($(p2)+(1,-.5cm)$) {$w_2^*\rightarrow c_{\floor{\frac{q+\ell}{2}}}$};
	
	\end{tikzpicture}
	\caption{The graph $J_u$ if $\ell$ (the number of edges of $P_1$) is even and $q-\ell$ (the number of edges of $P_2$) is odd.}
	\label{fig:Ju}
\end{figure}
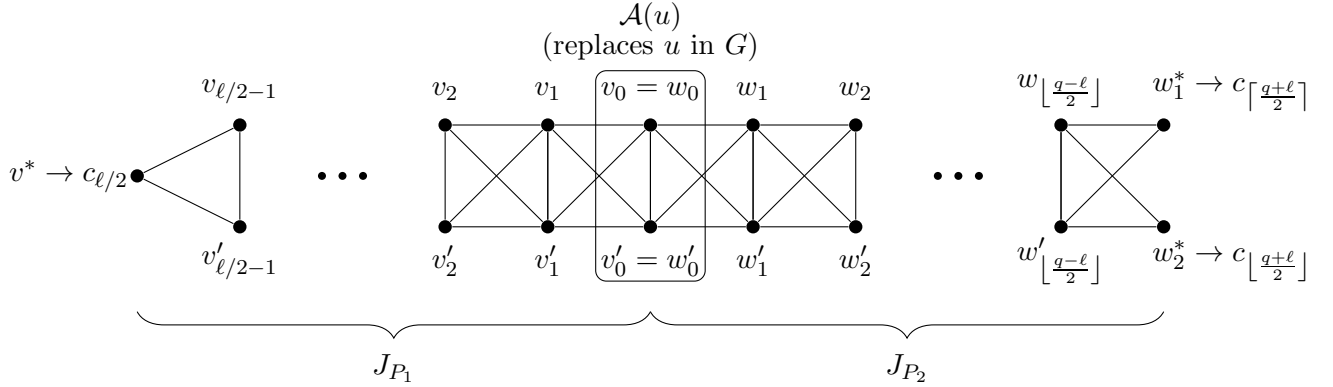

\bigskip
\noindent{\bf\boldmath Claim 1: For all $i\in \{0, \dots, \ell/2-1\}$, the image $h(\{v_i, v_i'\})$ contains at least one looped vertex. For all $j\in\{0, \dots, \floor{(q-\ell)/2}\}$, the image $h(\{w_j, w_j'\})$ contains at least one looped vertex.}

\smallskip
\noindent{\it Proof of Claim 1:} First consider $h(\{v_{\ell/2-1}, v_{\ell/2-1}'\})$. Assume that both $h(v_{\ell/2-1})$ and $h(v'_{\ell/2-1})$ are unlooped vertices of $H$. Since $\{v_{\ell/2-1}, v'_{\ell/2-1}\}$ is an edge in $J$, $h(v_{\ell/2-1})$ and $h(v'_{\ell/2-1})$ have to be connected by an edge in $H$ and therefore have to be different unlooped vertices. However, the vertices $v_{\ell/2-1}$ and $v_{\ell/2-1}'$ are also neighbours of $v^*$ which is pinned to the looped vertex $c_{\ell/2}$. It follows that $h(v_{\ell/2-1}),h(v'_{\ell/2-1}), c_{\ell/2}$ form a mixed triangle in $H$, a contradiction. This argument can be repeated iteratively for each $i= \ell/2-2, \dots, 0$ (using the fact that $h(\{v_{i+1}, v_{i+1}'\})$ contains a looped vertex).
The argument for $h(\{w_j, w_j'\})$ ($j\in\{0, \dots, \floor{(q-\ell)/2}\}$) is analogous.

\bigskip
\noindent{\bf\boldmath Claim 2: There exists a vertex $v\in \calA(u)$ with $h(v)\in \{c_0, c_\ell\}$.}

\smallskip
\noindent{\it Proof of Claim 2:}
Let $\calA(u)=\{v,v'\}$ where $h(v)$ is looped (this can be assumed without loss of generality by Claim 1). We will show that $h(v)\in\{c_0, c_{\ell}\}$: By Claim 1 and the construction of $J_{P_1}$ there exists a walk on $\ell/2$ edges in $H$ which uses looped vertices only and goes from $c_{\ell/2}$ to $h(v)$, which by assumption is looped itself. As this walk is looped it is in $H^*$, and as it contains $c_{\ell/2}$ it is in $H^{**}$. Hence,
\begin{equation}\label{equ:evenPath}
	h(v)\in \Nb{H^{**}}^{\ell/2}(c_{\ell/2}).
\end{equation}
Similarly, by the construction of $J_{P_2}$ we obtain
\begin{equation}\label{equ:oddPath}
h(v)\in \Nb{H^{**}}^{\ceil{(q-\ell)/2}}(c_{\floor{(q+\ell)/2}}) \cap \Nb{H^{**}}^{\ceil{(q-\ell)/2}}(c_{\ceil{(q+\ell)/2}}).
\end{equation}
Since $H^{**}$ is a reflexive triangle-extended cycle we have 
\[
\Nb{H^{**}}^{\ell/2}(c_{\ell/2})=\{c_0, \dots, c_{\ell}\} \cup \{d_i \mid i\in \calI\cap\{0,\dots, \ell-1\}\}
\] 
and
\[
\Nb{H^{**}}^{\ceil{(q-\ell)/2}}(c_{\floor{(q+\ell)/2}}) \cap \Nb{H^{**}}^{\ceil{(q-\ell)/2}}(c_{\ceil{(q+\ell)/2}})=\{c_\ell, \dots, c_{q-1}, c_0\} \cup \{d_i \mid i\in \calI\cap\{\ell,\dots, q-1\}\}.
\]  
Therefore, from Equations~\eqref{equ:evenPath} and~\eqref{equ:oddPath} it follows that $h(v)\in \{c_0, c_{\ell}\}$.

\bigskip
\noindent{\bf\boldmath Claim 3: Let $v\in \calA(u)$ and $h(v)\in \{c_0, c_\ell\}$. Then the image of the remaining vertices of $J_u$ under $h$ is determined completely. In particular, $h(\calA(u))=\{h(v)\}$.}

\smallskip
\noindent{\it Proof of Claim 3:}
Consider the case where $h(v)=c_0$ (the case $h(v)=c_\ell$ can be treated analogously). Since $H^{**}$ is a reflexive triangle-extended cycle and a connected component of $H^*$, the walk $c_{\ell/2}, c_{\ell/2-1}, \dots, c_0$ is the only $\ell/2$-edge walk on looped vertices in $H$ that leads from $c_{\ell/2}$ to $c_0$ (since there are no reflexive shortcuts in $C$). Thus, by Claim 1 and the construction of $J_{P_1}$, we have $\forall i\in \{0, \dots, \ell/2-1\},\ c_{i}\in h(\{v_i, v_i'\})$. We assume without loss of generality (by renaming) that for each $i\in \{0, \dots, \ell/2-1\}$ we have $h(v_i)=c_{i}$. 

Now consider the image $h(v_i')$ for some $i\in [\ell/2-1]$. The vertex $v_i'$ is a neighbour of $v_{i-1}$, $v_i$ and $v_{i+1}$ (or alternatively $v^*$ if $i=\ell/2-1$) in $J$. Therefore, by the fact that for each $i\in \{0, \dots, \ell/2-1\}$ we have $h(v_i)=c_{i}$ and by the pinning that ensures $h(v^*)=c_{\ell/2}$, we know that $h(v_i')$ has to be a neighbour of $c_{i-1}$, $c_{i}$ and $c_{i+1}$ in $H$. Since there is no edge between $c_{i-1}$ and $c_{i+1}$ we have $h(v_i')\notin \{c_{i-1}, c_{i+1}\}$. Then $h(v_i')=c_{i}$ as otherwise $c_{i-1},c_{i}, c_{i+1}, h(v_i')$ would form a square in $H$ which is a contradiction to the fact that $H$ is square-free. So we have established that for each $i\in [\ell/2-1]$ it holds that $h(v_i)=h(v_i')=c_{i}$. 

Similarly one establishes that for each $i\in [\floor{(q-\ell)/2}]$ it holds that $h(\{w_i, w_i'\})=c_{q-i}$: Note that by the construction of $J_{P_2}$ the homomorphism $h$ has to map $w_{\floor{(q-\ell)/2}}$ and $w_{\floor{(q-\ell)/2}}'$ to common neighbours of $c_{\floor{(q+\ell)/2}}$ and $c_{\ceil{(q+\ell)/2}}$. Since $H^{**}$ is a reflexive triangle-extended cycle and a connected component of $H^*$, the walk $c_{\ceil{(q+\ell)/2}} \dots, c_{q-1}, c_0$ is the only $\floor{(q-\ell)/2}$-edge walk on looped vertices in $H$ that leads from a common neighbour of $c_{\floor{(q+\ell)/2}}$ and $c_{\ceil{(q+\ell)/2}}$ to $c_0$. Therefore, we have $c_0\in h(\{w_0, w_0'\})$ and $\forall i\in [\floor{(q-\ell)/2}],\ c_{q-i}\in h(\{w_i, w_i'\})$. Then by the same arguments as before we establish that for each $i\in [\floor{(q-\ell)/2}]$ it holds that $h(w_i)=h(w_i')=c_{q-i}$.

Finally, $v_0'$ is a neighbour of $w_1$, $v_0 (=w_0)$ and $v_1$. Hence, $h(v_0')$ is a neighbour of $h(w_1)=c_{q-1}$, $h(v_0)=c_0$ and $h(v_1)=c_1$. Since there is no edge between $c_{q-1}$ and $c_{1}$ we have $h(v_i')\notin \{c_{q-1}, c_{1}\}$. Then $h(v_0')=c_{0}$ as otherwise $c_{q-1},c_{0}, c_{1}, h(v_0')$ would form a square in $H$ which is a contradiction to the fact that $H$ is square-free. We obtain $h(v_0')=h(v_0)=c_0$. This proves Claim 3.

\bigskip
From the construction of $J$ together with Claim 2 and Claim 3 we directly obtain that $\hom{(G,\boldS^G)}{H}=\hom{(J,\boldS^J)}{H}$, which gives the sought-for reduction in the case where $\ell$ is even and $q-\ell$ is odd. All other cases of $\ell$ even or odd and $q-\ell$ even or odd can be treated analogously.
\end{proof}

Now we can prove Lemma~\ref{lem:inducedCycleGeneral}.
\begin{leminducedCycleGeneral}
	\stateleminducedCycleGeneral
\end{leminducedCycleGeneral}
	
\begin{proof}
	Suppose that $H$ contains a mixed triangle as an induced subgraph, then the statement of this lemma follows from Lemmas~\ref{lem:hardtriangles1} and~\ref{lem:hardtriangles2}.
	We can now assume that $H$ does not contain any mixed triangle as an induced subgraph. 
	
	Let $C=c_0, \dots, c_{q-1}$ be the reflexive cycle of length $q\ge 5$ in $H$. Let $H^*$ be the graph induced by the looped vertices in $H$. If $H$ (and hence $H^*$) contains an induced $\WR{3}$ or an induced net then $\sat\leap \Ret{H}$ by Lemmas~\ref{lem:inducedWR3General} and~\ref{lem:inducedNet}, respectively. Otherwise, the connected component of $H^*$ that contains the cycle $C$ has to be a reflexive triangle-extended cycle by Lemma~\ref{lem:extendedCycle} and therefore $H$ fulfills the requirements of Lemma~\ref{lem:extendedCycleAnalysis}. 
	
Let $\calL = \{\{c_0, c_1\}, \{c_0, c_2\}, \dots, \{c_0, c_{q-1}\}\}$. Then $\Hom{H,\calL}\eqap \sat$ by Lemma~\ref{lem:LHomReflexiveCycleHardness}. We can use Lemma~\ref{lem:extendedCycleAnalysis} iteratively to obtain $\Hom{H,\calL} \leap \Hom{H,\bigl\{S \subseteq V(H) \mid \abs{S}\in \{1,\abs{V(H)}\}\bigr\}}$. Note that by the problem definitions we have $\Hom{H,\bigl\{S \subseteq V(H) \mid \abs{S}\in \{1,\abs{V(H)}\}\bigr\}}=\Ret{H}$. Summarising,
	\[
	\sat \eqap \Hom{H,\calL} \leap \Hom{H,\bigl\{S \subseteq V(H) \mid \abs{S}\in \{1,\abs{V(H)}\}\bigr\}}=\Ret{H}.
	\]
\end{proof}

\section{Putting the Pieces together}\label{sec:finaltheorems}
This section contains the proof of Theorem~\ref{thm:RetNoSquare} (we restate it here for convenience). We will use the following theorem, which is a consequence of the classification in~\cite{FGZRet} for approximately counting retractions to graphs of girth at least $5$, since that proof does not use the fact that $H$ is triangle-free for irreflexive $H$.

\begin{thm}[{\cite[Theorem 2.3]{FGZRet}}]\label{thm:RetIrrNoSquare}
	Let $H$ be an irreflexive square-free graph.
	{
		\renewcommand{\theenumi}{\roman{enumi})}
		\renewcommand{\labelenumi}{\theenumi}
		\begin{myenumerate}            
			\item If every connected component of $H$ is a star, then $\Ret{H}$ is in $\FP$.
			\item Otherwise, if every connected component of $H$ is a caterpillar, then $\Ret{H}$ is approximation-equivalent to $\bis$.
			\item Otherwise, $\Ret{H}$ is approximation-equivalent to $\sat$.
		\end{myenumerate}
	}
\end{thm}

In the following lemma we collect the $\sat$-hardness results which we use to prove Theorem~\ref{thm:RetNoSquare}.

\begin{lem}\label{lem:sathardCollection}
	Let $H$ be a connected square-free graph other than a reflexive clique, a member of $\bisgraphs$, or an irreflexive caterpillar. Then $\Ret{H}$ is approximation-equivalent to $\sat$.
\end{lem}
\begin{proof}
	If $H$ is irreflexive then by assumption it is not a caterpillar. Thus $\Ret{H}$ is approximation-equivalent to $\sat$ by Theorem~\ref{thm:RetIrrNoSquare}. 
	
	If $H$ is not irreflexive, i.e.~if $H$ has at least one loop, then we collect different $\sat$-hardness results proved throughout this work to show hardness. If $H$ contains a mixed triangle as induced subgraph, then $\sat \leap \Ret{H}$ by Lemmas~\ref{lem:hardtriangles1} and~\ref{lem:hardtriangles2}. If $H$ does not contain a mixed triangle as an induced subgraph but contains a $\WR{3}$, a net or a reflexive cycle of length at least $5$ as an induced subgraph, then $\sat \leap \Ret{H}$ by Lemmas~\ref{lem:inducedWR3General},~\ref{lem:inducedNet} and~\ref{lem:inducedCycleGeneral}, respectively.
	It remains to show $\sat \leap \Ret{H}$ if $H$ is a graph with the following properties:
	\begin{myitemize}
		\item $H$ is connected and square-free.
		\item $H$ has at least one looped vertex.
		\item $H$ is not a reflexive clique.
		\item $H\notin \bisgraphs$.
		\item $H$ does not contain any of the following as an induced subgraph: a mixed triangle, a $\WR{3}$, a net, a reflexive cycle of length at least $5$.
	\end{myitemize}
	Let $H^*$ be a connected component in the graph induced by the looped vertices in $H$. (It will turn out that $H^*$ is actually the only connected component in this graph.) Then by the properties of $H$ and Lemma~\ref{lem:extendedCycle} we know that $H^*$ is a reflexive triangle-extended path. We recall the definition of a reflexive triangle-extended path from Definition~\ref{def:extendedPath}: $H^*$ is a reflexive path $c_0, \dots, c_{q-1}$ together with a set $\calI\subseteq \{0, \dots, q-2\}$, and a reflexive triangle $d_i, c_i, c_{i+1}$ for each $i\in \calI$.
	Since $H^*$ is not a reflexive clique it holds that $q-1\ge 2$.
	
	Note that $H^*\in \bisgraphs$ (where the set of bristles is empty and $c_i$ corresponds to $p_i$). For all $i\in\calI$ the clique $K_i$ has size $3$, for $i\notin \calI$ it has size $2$. Since $H\notin \bisgraphs$ and $H$ is connected, there exists a vertex $u$ outside of $H^*$ with a neighbour $v$ in $H^*$. The vertex $u$ has to be unlooped as otherwise it would be part of the reflexive connected component $H^*$. We consider four disjoint cases.
	\begin{itemize}
		\item If there exists a vertex $u \notin V(H^*)$ ($u$ is unlooped) which is adjacent to a vertex $v\in V(H^*)$ and $\deg_{H}(u)\ge 2$, then consider two different cases:
		\begin{itemize}
			\item If $u$ is adjacent to a vertex $w\in\NH(v)$ with $w\neq v$, then $w\neq u$ since $u$ is unlooped and $u,v,w$ is a mixed triangle, a contradiction.
			\item If $v$ is the only neighbour of $u$ in $\NH(v)$, then the requirements of Lemma~\ref{lem:degree2bristle} are met (with $b=v$ and $g=u$) and hence $\sat \leap \Ret{H}$.
		\end{itemize}
		\item If there exists a vertex $u \notin V(H^*)$ which is adjacent to a vertex $v\in V(H^*)$, $\deg_{H}(u)=1$ and $v\in \{d_i \mid i\in \calI\}$, then $H[\NH(v)]$ is a graph of the form $\X{k_1}{0}{1}$ where $k_1\ge 1$ (cf.~Figure~\ref{fig:NeighbourhoodOfLoopedVertex}) and therefore
		\[
		\sat\leap \Hom{\X{k_1}{0}{1}} \leap \Ret{\X{k_1}{0}{1}} = \Ret{H[\NH(v)]} \leap \Ret{H},
		\]
		by Lemma~\ref{lem:hardNeighbourhood2}, Observation~\ref{obs:HomToRetToLHom} and Observation~\ref{obs:PinNeighbourhood} (in the order of the reductions used).
		\item Suppose there exists a vertex $u \notin V(H^*)$ with $\deg_{H}(u)=1$ that is adjacent to a vertex $v\in \{c_0, c_{q-1}\}$.  Without loss of generality (by renaming the vertices of $H^*$) let $v=c_0$. If $0\in \calI$ then $c_0$ is part of a reflexive triangle $d_0, c_0, c_1$ in $H^*$ and $H[\NH(v)]$ is a graph of the form $\X{k_1}{0}{1}$ where $k_1\ge 1$. Then we have $\sat \leap \Ret{H}$ by the same arguments used in the previous case. If otherwise $0\notin \calI$ then $c_1$ is the only neighbour of $c_0$ in $H^*$ and $H[\NH(v)]$ is a graph of the form $\X{k_1}{1}{0}$ where $k_1\ge 1$. Then we use Theorem~\ref{thm:RetGirth5} to infer that $\sat \leap \Ret{\X{k_1}{1}{0}}$ (since $\X{k_1}{1}{0}$ has girth at least $5$ and therefore is subject to Theorem~\ref{thm:RetGirth5}). It follows that $\sat \leap \Ret{H}$ by the same arguments as in the previous case.
		\item If for every pair $u,v$ of adjacent vertices with $u\notin V(H^*)$ and $v\in V(H^*)$ we have $\deg_H(u)=1$ ($u$ is a so-called bristle) and $v\in \{c_1, \dots, c_{q-2}\}$, then $H$ is the triangle-extended path $H^*$ together with a number of bristles (all of which are attached to a vertex in $\{c_1, \dots, c_{q-2}\}$). 
		To match the notation of $\bisgraphs$ in Definition~\ref{def:bisgraphs} we set $Q=q-2$ and, for all $i\in \{0, \dots, Q+1\}$, we set $p_i=c_i$. Further, if $i\in \calI$ then $K_i=\{p_i, d_i,p_{i+1}\}$ and otherwise $K_i=\{p_i, p_{i+1}\}$. Note that $\abs{K_i}\in \{2,3\}$ which we will use in a moment.
		For each $i\in [Q]$ let $B_i$ be the set of unlooped neighbours (bristles) of $p_i$.
		By the fact that all unlooped vertices of $H$ have degree $1$ and a neighbour in $\{c_1, \dots, c_{q-2}\}=\{p_1, \dots, p_{Q}\}$ we have $V(H)= \bigcup_{i=0}^Q K_i \cup \bigcup_{i=1}^Q B_i$ and 
		$E(H) = \bigcup_{i=0}^Q \left(\ucp{K_i}{K_i}\right) \cup \bigcup_{i=1}^Q \left(\{p_i\} \times B_i\right)$.
		Since $H^*$ is a triangle-extended path we can also verify the properties $K_{i-1}\cap K_i = \{p_i\}$ (for $i\in [Q]$) and $K_{i}\cap K_j = \emptyset$ (for $i,j\in \{0, \dots, Q\}$ with $\abs{j-i}>1$).
		
		Therefore, since $H\notin \bisgraphs$, there exists $i\in [Q]$ such that at least one of the following holds:
		\begin{enumerate}
			\renewcommand\labelenumi{(\theenumi)}
			\renewcommand{\theenumi}{\roman{enumi}}
			\item \label{item:mainThm3} $\abs{K_{i-1}}=\abs{K_{i}}=2$ and $\abs{B_i}\ge2$.
			\item \label{item:mainThm4} $\abs{K_{i-1}}=2$ and $\abs{K_{i}}=3$ (or $\abs{K_{i-1}}=3$ and $\abs{K_{i}}=2$) and $\abs{B_i}\ge3$.
			\item \label{item:mainThm5} $\abs{K_{i-1}}=\abs{K_{i}}=3$ and $\abs{B_i}\ge5$.
		\end{enumerate}
		In all three cases we will show that the neighbourhood of $p_i$ induces a $\sat$-hard subgraph, i.e.~that $\sat \leap \Ret{H[\NH(p_i)]}$. Then, by Observation~\ref{obs:PinNeighbourhood}, we obtain $\sat \leap \Ret{H}$ which completes the proof of this case and with it the proof of the theorem.
		\begin{itemize}
			\item If item~\eqref{item:mainThm3} holds, then $H[\NH(p_i)]$ is of the form $\X{k_1}{2}{0}$ where $k_1\ge 2$. The graph $\X{k_1}{2}{0}$ has girth at least $5$ and therefore is subject to Theorem~\ref{thm:RetGirth5}. Since $\X{k_1}{2}{0}$ with $k_1\ge 2$ is a mixed graph but not a partially bristled reflexive path we obtain $\sat \leap \Ret{H[\NH(p_i)]}$ by Theorem~\ref{thm:RetGirth5}.
			\item If item~\eqref{item:mainThm4} holds, then $H[\NH(p_i)]$ is of the form $\X{k_1}{1}{1}$ where $k_1\ge 3$. Then $\sat \leap \Ret{H[\NH(p_i)]}$ by Lemma~\ref{lem:hardNeighbourhood3}.
			\item If item~\eqref{item:mainThm5} holds, then $H[\NH(p_i)]$ is of the form $\X{k_1}{0}{2}$ where $k_1\ge 5$. Then $\sat \leap \Ret{H[\NH(p_i)]}$ by Lemma~\ref{lem:hardNeighbourhood4}.
		\end{itemize}
	\end{itemize}
\end{proof}

We will use the following remark to deal with graphs that have multiple connected components.
\begin{rem}[{\cite[Remark 1.15]{FGZRet}}]\label{rem:connectivity}
	Let $H$ be a graph with connected components $H_1, \dots, H_k$. On the one hand it holds that $\forall j\in [k], \Ret{H_j}\leap \Ret{H}$. On the other hand, given an oracle for each $\Ret{H_j}$, we can construct a polynomial-time algorithm for $\Ret{H}$. 
\end{rem}

{\renewcommand{\thethm}{\getrefnumber{thm:RetNoSquare}}
	\begin{thm}
		\ThmRetNoSquare
	\end{thm}
	\addtocounter{thm}{-1}
}
\begin{proof}
	If $H$ is a trivial graph then $\LHom{H}\in \FP$ by the result of Dyer and Greenhill~\cite{DG} (see Theorem 7). From $\Ret{H}\leap \LHom{H}$ (Observation~\ref{obs:HomToRetToLHom}) it also follows that $\Ret{H}\in \FP$.
	Then item~{\it i)} follows from Remark~\ref{rem:connectivity}.
	
	If $H$ is a graph for which item~{\it i)} does not hold, then $H$ has a connected component $H'$ that is not a trivial graph. By Remark~\ref{rem:connectivity} we have $\Ret{H'} \leap \Ret{H}$. Then $\bis$-hardness in item {\it ii)} follows from the reduction $\Hom{H'}\leap \Ret{H'}$ (Observation~\ref{obs:HomToRetToLHom}) together with the fact that $\bis \leap \Hom{H'}$ since $H'$ is a non-trivial connected graph~\cite[Theorem 1]{GGJBIS}. 
	
	We will now prove $\bis$-easiness in item~{\it ii)}. If $H'$ is a trivial graph we have already pointed out that $\Ret{H'}\in \FP$ and hence $\Ret{H'}$ is trivially $\bis$-easy. If $H'\in \bisgraphs$ then $\Ret{H'}\leap \bis$ by Theorem~\ref{thm:bisEasyRet}. If $H'$ is an irreflexive caterpillar then $\Ret{H'}\leap \bis$ by Theorem~\ref{thm:RetIrrNoSquare}. Hence, $\bis$-easiness in item~{\it ii)} follows from Remark~\ref{rem:connectivity}.
	
	If $H$ is a graph for which item~{\it ii)} does not hold, then $H$ has a connected component $H'$ that is not trivial, not a member of $\bisgraphs$ and not an irreflexive caterpillar. Then $\Ret{H'}$ is approximation-equivalent to $\sat$ by Lemma~\ref{lem:sathardCollection} and $\Ret{H'} \leap \Ret{H}$ by Remark~\ref{rem:connectivity}. This proves item~{\it iii)}.
\end{proof}

\bibliography{\jobname}
\end{document}